\documentclass[hidelinks,onefignum,onetabnum]{siamart251216}

\usepackage{lipsum}
\usepackage{amsfonts}
\usepackage{graphicx}
\usepackage{epstopdf}
\usepackage{algorithmic}
\ifpdf
\DeclareGraphicsExtensions{.eps,.pdf,.png,.jpg}
\else
\DeclareGraphicsExtensions{.eps}
\fi

\newsiamremark{remark}{Remark}
\newsiamremark{hypothesis}{Hypothesis}
\crefname{hypothesis}{Hypothesis}{Hypotheses}
\newsiamthm{claim}{Claim}
\newsiamthm{assumption}{Assumption}

\headers{Sparse Representations of Dynamical Networks}{\c{S}. Sab\u{a}u, A. {S}peril\u{a}, C. Oar\u{a}, and A. Jadbabaie}

\title{Sparse Representations of Dynamical Networks:  \\ A Coprime Factorization Approach\thanks{This manuscript was last updated on \today.
		\funding{The last author was supported by the ONR grant no.~N00014-23-1-2299, and by the AFOSR MURI grant no.~FA9550-25-1-0375.}}}
	
\subtitle{\color{white}\vspace{-12mm}\normalsize\textbf{(ANDREI'S VERSION)}\vspace{4mm}}

\author{\vspace{3mm}\c{S}erban Sab\u{a}u\thanks{RTX BBN Technologies, 10 Moulton Street, Cambridge MA 02138 (\email{serban.sabau@rtx.com})}
	\and Andrei {S}peril\u{a}\thanks{Faculty of Automatic Control and Computers,
		National University of Science and Technology ``Politehnica'' Bucharest, Sector 6, 060042 Romania (\email{andrei.sperila@upb.ro}, \email{cristian.oara@upb.ro})}
	\and Cristian Oar\u{a}\footnotemark[3] \and Ali Jadbabaie\thanks{Institute for Data, Systems and Society, Massachusetts Institute of Technology (MIT), Cambridge, MA 02139 USA (\email{jadbabai@mit.edu})}}

\usepackage{amsopn}

\usepackage{comment,arydshln}
\usepackage{tikz}
\usepackage[english]{babel}

\ifpdf
\hypersetup{
  pdftitle={Sparse Representations of Dynamical Networks:\\ A Coprime Factorization Approach},
  pdfauthor={\c{S}. Sab\u{a}u, A. {S}peril\u{a}, C. Oar\u{a}, and A. Jadbabaie}
}
\fi

\begin{document}
	
\tikzset{%
	block/.style    = {draw, thick, rectangle, minimum height = 1.7em,
		minimum width = 1.7em},
	block1/.style    = {draw, thick, rectangle, minimum height = 1.7em,
		minimum width = 1.7em,fill=gray!70},
	block2/.style    = {draw, thick, rectangle, minimum height = 1.7em,
		minimum width = 1.7em,fill=gray!30},
	sum/.style      = {draw, circle, node distance = 1.4cm}, 
	input/.style    = {coordinate}, 
	output/.style   = {coordinate} 
}

\maketitle

\begin{abstract}
We study a class of dynamical networks modeled by linear and time-invariant systems which are described by state-space realizations. For these networks, we investigate the relations between various types of factorizations which preserve the structure of their component subsystems' interconnection. In doing so, we provide tractable means of shifting between different types of sparsity-preserving representations and we show how to employ these factorizations to obtain distributed implementations for stabilizing and possibly stable controllers. By formulating all these results for both discrete- and continuous-time systems, we develop specialized distributed implementations that, up to this point, were only available for networks modeled as discrete-time systems.\vspace{-1mm}
\end{abstract}

\begin{keywords}
Sparse representations, distributed controllers, coprime factors, networked systems.\vspace{-1mm} 
\end{keywords}

\begin{MSCcodes}
93A14, 93B99, 93C05\vspace{-2mm}
\end{MSCcodes}

\section{Introduction}

With the growing trend of modern technical applications becoming more interconnected and spatially widespread, distributed process control has become an area of intense research in the specialized literature catering to dynamical systems. However, with its convexity refuted and its generic intractability revealed (see, for example, \cite{DateChow}), recent efforts in the area of {distributed control} for Linear and Time-Invariant (LTI) systems have shifted from the classical formulation of imposing sparsity upon the Transfer Function Matrix (TFM) of the controller \cite{Voul1, Voul2, Rotko, Shah} towards the implementation of decentralized state-space-based control laws.
\vspace{-1mm}

\subsection{Context}

Among the first attempts to formalize a state-space-oriented approach to distributed control was the work from \cite{Realiz}, in which structured realizations are employed in order to obtain control laws that inherit the sparsity structure of the controlled network. A key concept pioneered in \cite{Realiz}, which turns out to be crucial for the results presented in this manuscript, is the fact that the distributed state-space implementation of the sought-after control laws must possess certain structural properties, such as stabilizability and detectability, in order to certify closed-loop stability.

Yet, arguably, the most celebrated results in this new paradigm of distributed control came with the advent of the System-Level Synthesis (SLS) design framework from \cite{SLS}. Based upon the classical work in \cite{BBS}, the SLS approach produces decentralized control laws for discrete-time systems, by employing an affine parametrization of all the achievable closed-loop maps. Additionally, the work in \cite{Separ} has shown that, in the SLS framework, it is possible to decouple controller synthesis from closed-loop design. Closely following the developments in \cite{SLS}, the body of work presented in \cite{Explicit} aims to provide a continuous-time parallel to the SLS framework, with important contributions to the vehicular platoon problem, which was also tackled in \cite{plutonizare}. Notably, the approaches from \cite{Explicit} and \cite{plutonizare} are linked via the implicit use of coprime factorizations, and an underlying dependence upon the Youla parametrization.\newpage

In the same vein as the techniques from \cite{Explicit} and \cite{plutonizare}, a new design paradigm, based upon coprime factorizations and heavily inspired by the work from \cite{DSF}, was recently formalized in \cite{NRF}, with the aim of proposing a unitary continuous- and discrete-time approach to optimal distributed control. The theoretical developments presented in \cite{NRF} are also supported by the convex relaxation-based procedures from \cite{aug_sparse}, which are part of a larger class of design methods currently gaining traction in the distributed control literature (see \cite{Lavaei,Jovanovic, reg, Sznaier, trian}). Fortunately, the intrinsic connections between the SLS-focused methods and their Youla-based counterparts has been elucidated in \cite{Tseng} and \cite{Luca3}, while a system-theoretical middle ground between these two approaches was proposed in \cite{Luca1} and \cite{Luca2}, via the Input-Output Parametrization.

\subsection{Motivation}

Although the various techniques mentioned in the previous subsection clearly show that significant progress has been made in moving away from the traditional sparse TFM-based approach to distributed control, a number of important limitations still remain unaddressed in literature. While certain procedures, such as the ones in \cite{Realiz}, build upon highly structured realizations for the investigated networks, obtaining these exceptional representations becomes somewhat unwieldy for large-scale systems. Except for those networks endowed with particularly sparse or with recursive inter-nodal connections (see the numerical example in \cite{aug_sparse}), the effort needed to compose the interconnected plant's model out of its nodes' realizations does not scale gracefully with respect to the network's complexity. 

Similar to the results from \cite{Realiz}, other methods proposed in literature focus exclusively on networks modeled in either continuous- or discrete-time, with \cite{Explicit} and \cite{SLS} serving as respective examples. Oftentimes, this exclusivity is owed to the manner in which stability guarantees are offered for the proposed control schemes. We point out that this particular limitation, which will be further commented upon in the sequel, has been lifted following the introduction of the design framework from \cite{NRF}. In doing so, however, the control laws from \cite{NRF} take on a distinct input-output interpretation, regardless of their distributed and stabilizing state-space-based implementations, as in Section III.C in \cite{NRF}. Thus, the cost paid for this increased flexibility is a degree of opaqueness with respect to network state dynamics, at the implementation level.

\subsection{Contribution}

The aim of this paper is to propose a factorized representation, inspired by the system response TFMs from \cite{SLS}, which is suitable to both network analysis and distributed controller synthesis. This representation, dubbed a System Response-Type Realization (SRTR), captures both the dynamics and the structure of the network, regardless of whether it is modeled as either a continuous- or a discrete-time system. This factorization offers a notably keener insight into the state dynamics of the plant than the Network Realization Function (NRF) representation from \cite{NRF}, while ensuring the same guarantees of closed-loop stability as NRF-based control laws, when implementing controllers via SRTR-based feedback schemes.

In developing the theory associated with this new representation, we implicitly investigate the subtle connections between the classical left coprime factorizations discussed in \cite{V}, the system response TFMs from \cite{SLS} and the NRF pairs introduced in \cite{NRF}. Additionally, we propose numerically tractable means of shifting between the aforementioned types of representations, all while placing a particular emphasis on the preservation of three key properties: stability, coprimeness and sparsity structure. Finally, we show how the first property allows us to elegantly tackle the problem of strong stabilization in the decentralized setting, how the second property implicitly allows for the commutation between representations, and how the third property directly promotes the implementation of distributed control laws.

\subsection{Outline of the Paper} In the second section, we include a series of preliminary notions which are standard to control theory literature. The third section of the paper contains a brief outline of the theoretical concepts behind various representations for LTI systems. In the fourth section, we show that, while the NRF representation of a given LTI system is not guaranteed to be coprime, this property can be ensured for a closely related representation associated with our system, heavily inspired by those in \cite{SLS}. We also provide state-space formulas for the class of all such coprime representations, related to a particular type of realization for our LTI system. 

The fourth section also contains the main results of the paper: an explanation of the intrinsic connections between NRF pairs, the coprime representations inspired by \cite{SLS}, and the classical left coprime factorizations associated with a given TFM, along with the means of employing them to obtain distributed and stabilizing control laws. The fifth section showcases the direct application of these results through a numerical example. Finally, the sixth section contains a number of concluding remarks. 

In order to ensure that the paper is as self-contained as possible, some more specialized notions and technical results have been included in Appendices~\ref{sec:ap_A} and~\ref{sec:ap_B}, while Appendix~\ref{sec:ap_C} contains a series of auxiliary statements, employed in the proofs of the main results. Moreover, the aforementioned proofs have been placed in Appendix~\ref{sec:ap_D}, in order to promote the flow and readability of the paper.

\section{Preliminaries} \label{prelim}

\subsection{Notation}

We denote by $\mathbb{R}$ the set of real numbers and by $\mathbb{N}$ the set of natural ones, while $\mathbb{C}$ stands for the complex plane and $\mathbb{C}_\infty:=\mathbb{C}\cup\{\infty\}$ denotes the extended complex plane. Let $\mathbb{R}^{p\times m}$ be the set of $p \times m$ real matrices, for whose elements we denote by $\|\cdot\|$ the spectral norm, \emph{i.e.}, the largest singular value. Let also $\mathbb{R}(\lambda)^{p\times m}$ be the set of  $p \times m$ TFMs, which are matrices having real-rational functions as entries. Henceforth, we denote such TFMs through the use of boldface symbols. Additionally, we point out that our results apply to both continuous- and discrete-time LTI systems. Thus, we assimilate the indeterminate
$\lambda$ with  the complex variables $s$ or $z$ appearing in the
Laplace or $\mathcal{Z}$-transform, respectively, depending upon the type of system which is under investigation. 

\subsection{Realization Theory}

The main objects of study in this paper are LTI systems which are described in the time domain by the state equations
\begin{subequations}
	\begin{align} 
		\sigma{ x}  &=  A x  +  Bu  ,\  {x}(t_o)={x}_o, & \label{ss0ab} \\
		y &=  C  x  + D u, & \label{ss0c}
	\end{align}
\end{subequations}

\noindent  with $A\in\mathbb{R}^{n\times n}$ being the \emph{state matrix}, $B\in\mathbb{R}^{n\times m}$ the \emph{input matrix}, $C\in\mathbb{R}^{p\times n}$ the \emph{output matrix} and $D\in\mathbb{R}^{p\times m}$ the \emph{feedthrough matrix}, while $n$ is called the {order} of the realization. The operator $\sigma$ in \eqref{ss0ab}-\eqref{ss0c} stands for either the \emph{time derivative}, in the continuous-time context, or the \emph{forward unit-shift}, in the discrete-time case. Given any $n$-dimensional state-space realization \eqref{ss0ab}-\eqref{ss0c}  of  an LTI system, its input-output representation is expressed via the TFM denoted by
\begin{equation*}
	\mathbf G(\lambda) =  \left[\begin{array}{c|c}A-\lambda I_n & B \\ \hline C & D \end{array}\right] := D + C(\lambda I_n - A)^{-1}B.
\end{equation*}\newpage

A pair $(A,B)$ or a realization \eqref{ss0ab}-\eqref{ss0c} for which $\begin{bmatrix}
	A-\lambda I_n & B
\end{bmatrix}$ has full row rank at $\lambda\in\mathbb{C}$ is called \emph{controllable at $\lambda$}. Similarly, a pair $(C,A)$ or a realization \eqref{ss0ab}-\eqref{ss0c} that is \emph{observable at $\lambda$} is one for which the pair $(A^\top,C^\top)$ is controllable at $\lambda$. A \emph{controllable}/\emph{observable} matrix pair or realization \eqref{ss0ab}-\eqref{ss0c} is one that is controllable/observable at any $\lambda\in\mathbb{C}$.	A realization whose order is the smallest out of all others for the same $\mathbf{G}(\lambda)$ is called \emph{minimal} (see Section 6.5.1 of \cite{Kai}). Any two minimal realizations $(A,B,C,D)$ and $(\widehat{A},\widehat{B},\widehat{C},\widehat{D})$ of $\mathbf{G}(\lambda)$ having order $\nu$ are related via a state equivalence transformation $x=T {\widehat x}$, with $T$ invertible in $\mathbb{R}^{\nu\times \nu}$, given by
\begin{equation} \label{similea}
	\begin{array}{llll}
		A = T\widehat A T^{-1}, &  B = T\widehat B, &
		C = \widehat C T^{-1},  &  D = \widehat{D}.
	\end{array}
\end{equation}

\begin{definition}
	Any two realizations which are related by a state equivalence transformation, as in \eqref{similea}, are called \emph{equivalent}.
\end{definition}

Let $\mathbb{S}$ denote the \emph{domain of stability}, which is either the open left-half plane, \emph{i.e.}, $\left\{\lambda\in\mathbb{C}\ :\ \text{Re}(\lambda)<0\right\}$ in the case of continuous-time systems or the open unit disk, \emph{i.e.}, $\left\{\lambda\in\mathbb{C}\ :\ |\lambda|<1\right\}$ for discrete-time ones. A \emph{stabilizable} realization  \eqref{ss0ab}-\eqref{ss0c} or pair $(A,B)$ is one that is controllable at any $\lambda\in\mathbb{C}\backslash\mathbb{S}$. By duality, a \emph{detectable} realization \eqref{ss0ab}-\eqref{ss0c} or pair $(C,A)$ is one that is observable at any $\lambda\in\mathbb{C}\backslash\mathbb{S}$.

\subsection{Rational Matrix Theory}

The matrix polynomial $A - \lambda E$ is called a \emph{pencil}. If it is square and $\det(A-\lambda E)\not\equiv 0$, we refer to it as \emph{regular} (see chapter 12 of \cite{gantmacher} for more details). Note also that a pencil is merely a polynomial TFM.
The finite poles and zeros of any TFM $\mathbf G(\lambda)$ are rigorously defined in terms of the \emph{Smith-McMillan Form} (see Section 6.5.2 of \cite{Kai}), whereas the infinite poles and zeros of $\mathbf{G}(\lambda)$ are those of $\mathbf{G}\left(\frac{1}{\lambda}\right)$ at $\lambda=0$. Any realization \eqref{ss0ab}-\eqref{ss0c} that is both controllable and observable is also minimal (see \cite{Wonham}) and the poles of its TFM coincide with the eigenvalues of the regular pencil $A-\lambda I_n$, known as the realization's \emph{pole pencil}.

The \emph{McMillan degree} of $\mathbf{G}(\lambda)$ is the sum of the degrees of its poles (see Section 6.5.3 of \cite{Kai}), both finite and infinite. A TFM's \emph{normal rank}, which is its rank for almost all points in $\mathbb{C}_\infty$, is also rigorously defined via the TFM's Smith-McMillan Form. If a TFM has no poles at $\{\infty\}$, it is called \emph{proper}. Otherwise, it is called \emph{improper}. A TFM $\mathbf{G}\in\mathbb{R}(\lambda)^{p\times m}$ for which $\lim_{\lambda\rightarrow\infty} \mathbf{G}(\lambda)=O$ is called \emph{strictly proper}. Finally, a TFM is deemed \emph{stable} if it has no poles located in $\mathbb{C}_\infty\backslash\mathbb{S}$.

Two TFMs $\mathbf{N}(\lambda)\in\mathbb{R}(\lambda)^{p\times m}$ and $\mathbf{M}(\lambda)\in\mathbb{R}(\lambda)^{p\times p}$, with $\mathbf{M}(\lambda)$ having full normal rank, that satisfy $\mathbf{G}(\lambda)=\mathbf{M}^{-1}(\lambda)\mathbf{N}(\lambda)$ are said to be a \emph{left factorization} of $\mathbf{G}(\lambda)$. When, additionally, $\begin{bmatrix}
	\mathbf{M}(\lambda) & \mathbf{N}(\lambda)
\end{bmatrix}$ has full normal rank (this is always ensured by $\mathbf{M}(\lambda)$ being invertible) along with no zeros in $\mathbb{C}_\infty$ (see also \cite{SIMAX}), we say that $\mathbf{M}(\lambda)$ and $\mathbf{N}(\lambda)$ constitute a \emph{fully left coprime factorization} (FLCF) of $\mathbf{G}(\lambda)$ (see \cite{V, SIMAX, Rose70} for equivalent characterizations). When $\begin{bmatrix}
	\mathbf{M}(\lambda) & \mathbf{N}(\lambda)
\end{bmatrix}$ has full row normal rank along with no poles or zeros located outside of $\mathbb{S}$, the pair is called an \emph{LCF over $\mathbb{S}$}.

\section{Representations for LTI Networks} \label{adoua}

This section contains a discussion based upon various representations for networks of LTI systems, with an emphasis on those tentatively prototyped in \cite{CDC2013} and now discussed in a broader context.

\subsection{State-Space Realizations}

We start with the given system $\mathbf G(\lambda)$, described by the state equations of order $n$ given explicitly by
\begin{subequations}
	\begin{align} 
		\sigma{\widetilde x}  &= \widetilde A\widetilde x  +\widetilde Bu  ,\   \widetilde{x}(t_o)=\widetilde{x}_o, & \label{ss1ab} \\
		y &= \widetilde C\widetilde x.  & \label{ss1c}
	\end{align}
\end{subequations}\newpage

The following hypothesis is in effect throughout the paper.

\begin{assumption} \label{Regularity} The matrix $\widetilde C$ from \eqref{ss1c} has $p<n$ and full row rank.
\end{assumption}

\begin{remark}
Note that Assumption~\ref{Regularity} is in no way restrictive from a practical standpoint, where oftentimes $n\gg p$. The lack of full row rank for the matrix $\widetilde{C}$ implies that some of the sensor measurements are redundant and can be amalgamated via linear combinations to produce a new, full row rank $C$-matrix, all while preserving the property of observability (see Section 3.4 of \cite{Wonham} for the underlying connection with reduced-order observers). Notice also that when $n=p$, the matrix $\widetilde{C}$ is invertible and the problem can be reduced to the trivial, full state measurement case, in which the main representations discussed in the sequel take on the classical form given in \eqref{ss1ab}.
\end{remark}

We now choose any matrix $\overline C$ such that $T:= \begin{bmatrix} \widetilde C^\top & \overline C^\top \end{bmatrix}^\top$ is nonsingular, noticing that such a $\overline C$ always exists due to $\widetilde C$ having full row rank, and we apply the state equivalence transformation given by $T$ on \eqref{ss1ab}-\eqref{ss1c}, in order to get
\begin{subequations}
	\begin{align} 
		\begin{bmatrix} \sigma{y}  \\ \sigma{q} \end{bmatrix}  &= \begin{bmatrix}  A_{11} & A_{12} \\ A_{21} & A_{22}  \end{bmatrix}  \begin{bmatrix} y   \\ q  \end{bmatrix} + \begin{bmatrix} B_1 \\ B_2 \end{bmatrix} u,\    \label{ss2a}
		\begin{bmatrix} y(t_o) \\ q (t_o) \end{bmatrix} = \begin{bmatrix} y_o \\ q_o \end{bmatrix},\ \begin{bmatrix} y \\ q  \end{bmatrix}  = T \widetilde x,   \\
		y  &= \begin{bmatrix} \hspace{2mm}I_p &\hspace{2.25mm} O \hspace{1.75mm}\end{bmatrix} \begin{bmatrix} y  \\ q \end{bmatrix}.   & \label{ss2b}
	\end{align}
\end{subequations}

Henceforth, we will assume that the following statement holds.

\begin{assumption} \label{Observability} The pair $(\widetilde C, \widetilde A)$ in \eqref{ss1ab}-\eqref{ss1c} is observable.
\end{assumption}

\begin{remark} \label{Leuenberger} Notice that, due to the particular structure of the output matrix from \eqref{ss2b}, namely that it is given by $\begin{bmatrix}
		I_p&O
	\end{bmatrix}$, Assumption~\ref{Observability} is equivalent to the pair $( A_{12}, A_{22})$ from \eqref{ss2a} being observable.
\end{remark}

Taking now the Laplace or $\mathcal{Z}$-transform of \eqref{ss2a}, depending on the time domain, and considering all initial conditions to be 0, we get that
\begin{equation} \label{Lapla}
	\small\begin{bmatrix} \lambda I_p -A_{11} & - A_{12} \\ - A_{21}&\lambda I_{n-p} - A_{22}  \end{bmatrix} \small\begin{bmatrix} \mathbf{Y}(\lambda)\\\mathbf{Q}(\lambda) \end{bmatrix} =  \small\begin{bmatrix} B_1 \\ B_2 \end{bmatrix} \mathbf{U}(\lambda).
\end{equation}
By left-multiplying \eqref{Lapla} with the following TFM
\begin{equation} \label{Omega_i}
	\mathbf \Omega (\lambda) := \small\begin{bmatrix} I_{p} &  A_{12} ( \lambda I_{n-p} - A_{22})^{-1} \\ 
		O          & I_{n-p} \\ \end{bmatrix},
\end{equation}
\noindent we obtain the identity
\begin{equation} 
	\footnotesize\begin{bmatrix} (\lambda I_p - A_{11}) -  A_{12} (\lambda I_{n-p} - A_{22})^{-1}A_{21}  & O \\ * & * \end{bmatrix}\begin{bmatrix} \mathbf{Y}(\lambda)\\\mathbf{Q}(\lambda) \end{bmatrix}
	= \begin{bmatrix} B_1 + A_{12} ( \lambda I_{n-p} - A_{22})^{-1} B_2 \\ * \end{bmatrix}\mathbf{U}(\lambda),\normalsize\label{WV1}
\end{equation}
\noindent where $*$ denotes entries which are not of interest. Taking the inverse Laplace or $\mathcal{Z}$-transform in the first block-row of \eqref{WV1}, we are able to express (recalling that we consider all initial conditions to be 0) the following input-output relationship
\begin{multline*}
	\hspace{-2mm}\lambda \mathbf{Y}(\lambda) = \big( A_{11} +  A_{12} (\lambda I_{n-p} - A_{22})^{-1}A_{21} \big)  \mathbf{Y}(\lambda)+ \big(B_1 +  A_{12} (\lambda I_{n-p} - A_{22})^{-1}B_2 \big) \mathbf{U}(\lambda),
\end{multline*}
\noindent from which we introduce the notation 
\begin{equation*}
	\begin{array}{ll}
		\mathbf W(\lambda) :=  A_{11} + A_{12} (\lambda I_{n-p} - A_{22})^{-1}A_{21}, &
		\mathbf V(\lambda) :=  B_1 +  A_{12} (\lambda I_{n-p} - A_{22})^{-1}B_2.
	\end{array}
\end{equation*}

In doing so, we finally get the following equation, which describes the relationship between the network's manifest variables
\begin{equation} \label{WVin5} \lambda \mathbf{Y}(\lambda) =  \mathbf W(\lambda) \mathbf{Y}(\lambda) + \mathbf V(\lambda) \mathbf{U}(\lambda).
\end{equation}

\subsection{Network Realization Functions}

Promisingly, the identity from \eqref{WVin5} generates numerous avenues of investigation. We proceed to investigate one of them, by first pointing out that a routine computation produces the fact that $\displaystyle \mathbf G(\lambda) = \big(\lambda I_p -\mathbf W(\lambda)\big)^{-1}\mathbf V(\lambda)$. Notice that, since $\mathbf W(\lambda)$ is always proper, it follows that $\big(\lambda I_p -\mathbf W(\lambda)\big)$ is always invertible as a TFM. Let $\mathbf D(\lambda)$ denote the TFM obtained by taking the diagonal entries of $\mathbf W(\lambda)$, that is $\displaystyle \mathbf D(\lambda) := \mathrm{diag} \{ \mathbf W_{11}(\lambda), \mathbf W_{22}(\lambda), \dots, \mathbf W_{pp}(\lambda) \}$. Then, 
\begin{multline*}
	\mathbf G(\lambda) = \Big[ \big( \lambda I_p -\mathbf D(\lambda) \big) - \big( \mathbf W(\lambda) -\mathbf D(\lambda) \big) \Big]^{-1}\mathbf V(\lambda)=\\
	= \Big[ I_p- \big( \lambda I_p -\mathbf D(\lambda) \big)^{-1} \big( \mathbf W(\lambda) -\mathbf D(\lambda) \big) \Big]^{-1} \big( \lambda I_p -\mathbf D(\lambda) \big)^{-1}  \mathbf V(\lambda),
\end{multline*}
\noindent where $(\mathbf{W}-\mathbf{D})$ has only zeros on its diagonal and, by denoting
\begin{subequations} 
	\begin{align}
		\mathbf \Phi(\lambda) & :=   \big( \lambda I_p -\mathbf D(\lambda) \big)^{-1} \big( \mathbf W(\lambda) -\mathbf D(\lambda) \big), & \label{Q} \\
		\mathbf \Gamma(\lambda) & :=  \big( \lambda I_p -\mathbf D(\lambda) \big)^{-1}  \mathbf V(\lambda), &  \label{P} 
	\end{align}
\end{subequations}

\noindent we get that $\displaystyle \mathbf G(\lambda) = \big( I_p -\mathbf \Phi(\lambda)\big)^{-1}\mathbf \Gamma(\lambda)$ or, equivalently, that
\begin{equation} \label{NRF}
	\mathbf{Y}(\lambda) = \mathbf \Phi(\lambda) \mathbf{Y}(\lambda) + \mathbf \Gamma(\lambda) \mathbf{U}(\lambda).
\end{equation}

\begin{remark} \label{AliObs}
	The splitting and the ``extraction'' of the diagonal in \eqref{Q} are intended to ensure that $\mathbf \Phi(\lambda)$ has the sparsity (and the meaning) of the {\em adjacency matrix} of the graph, \emph{i.e.}, the square matrix which indicates the adjacency between any two of the graph's vertices, thus describing the causal relationships between the system's various outputs. Consequently, $\mathbf \Phi(\lambda)$ will always have zero entries on its diagonal, which is a key factor in facilitating any implementations based upon \eqref{Q}-\eqref{P}.
\end{remark}

We now formally introduce the notion of NRF pair, discussed at length in \cite{NRF}. 

\begin{definition} \label{NRFdefinw} Given a $\mathbf G(\lambda)\in\mathbb{R}^{p \times m}(\lambda)$ and any two TFMs $\mathbf \Phi(\lambda)\in\mathbb{R}^{p \times p}(\lambda)$ and $\mathbf \Gamma(\lambda)\in\mathbb{R}^{p \times m}(\lambda)$, satisfying $\displaystyle \mathbf G(\lambda) = \big( I_p -\mathbf \Phi(\lambda) \big)^{-1}\mathbf \Gamma(\lambda)$ and with $\mathbf \Phi(\lambda)$ having zero entries on its diagonal, the pair $\small\big( \mathbf \Phi(\lambda), \mathbf \Gamma(\lambda) \big)$ \normalsize is called an \emph{NRF pair} of $\mathbf G(\lambda)$.
\end{definition}

\begin{remark}\label{rem:NRF_MIMO}
	Note that the outputs of the network's nodes need not be scalar and that constraining only the diagonal entries of $\mathbf \Phi(\lambda)$ represents the most generic case. Indeed, when implementing controllers via their NRF pairs, it may be convenient for certain nodes to compute more than one command signal, with the latter representing our network's outputs. Therefore, it is natural to impose identically zero elements on the block-diagonal of $\mathbf \Phi(\lambda)$, in order to avoid the redundant use of local information. This merely amounts to an additional sparsity constraint on the entries of $\mathbf \Phi(\lambda)$ and we point out that the results from \cite{NRF} are not dependent upon the structure of the aforementioned TFM, except for its identically zero diagonal. Moreover, as explained in Remark III.8 of \cite{NRF}, the block-row implementations which arise naturally from imposing identically zero blocks on the block-diagonal of $\mathbf \Phi(\lambda)$ retain all the desirable properties of the row implementations provided in the results of Section III from \cite{NRF}.
\end{remark}

\begin{remark}\label{Kron}
	Another benefit of the these NRF representations consists in the preservation of the network's adjacency graph, both with respect to the input-output relationships that characterize it, as modeled by $\mathbf{\Gamma}(\lambda)$, and also of the intrinsic interconnections of its manifest variables, as encoded through the structure of $\mathbf{\Phi}(\lambda)$. These graph-oriented representations have received significant attention in the context of Kron reduction for electrical networks (see \cite{Kron1,Kron3,Kron2} and the references therein).
\end{remark}

It is for the reasons discussed in this subsection that NRF pairs present themselves as tractable representations for the purpose of network modeling and control. Their versatility in all areas of control theory, such as system identification, controller synthesis and even model reduction, is one of the chief motivators for investigating their relationships with the other types of representations from this paper.

\subsection{An Example of NRF Modelling}\label{subsec:mdl}

\begin{figure}[t]
	\centering
	\begin{tikzpicture}[scale=0.3]
		\draw[xshift=0.1cm, >=stealth ] [->] (4,0) -- (6,0);
		\draw[xshift=0.1cm]    (3,0) node {$z_1$};
		\draw[ thick, xshift=0.1cm]  (6,-1) rectangle +(4,2);
		\draw [xshift=0.1cm](8,0)   node {{${\bf \Gamma_1(\lambda)}$}};
		\draw[xshift=0.1cm, >=stealth ] [->] (10,0) -- (12,0);
		\draw[ xshift=0.1cm ]  (12.5,0) circle(0.5);
		\draw[ xshift=0.1cm,  >=stealth] [->] (13,0) -- (19,0);
		\draw [xshift=0.1cm](16,0.75)   node {$u_1$} ;
		\draw[ thick, xshift=0.1cm]  (19,-1) rectangle +(4,2);
		\draw [xshift=0.1cm](21,0)   node {{${\bf \Phi_2(\lambda)}$}};
		\draw[ xshift=0.1cm ]  (25.5,0) circle(0.5);
		\draw[ thick, xshift=0.1cm]  (23.5,2.5) rectangle +(4,2);
		\draw [xshift=0.1cm](25.5,3.5)   node {{${\bf \Gamma_2(\lambda)}$}};
		\draw[xshift=0.1cm, >=stealth ] [->] (21.5,3.5) -- (23.5,3.5);
		\draw[xshift=0.1cm]    (20.5,3.5) node {$z_2$};
		\draw[xshift=0.1cm, >=stealth ] [->] (23,0) -- (25,0);
		\draw[xshift=0.1cm, >=stealth ] [->] (25.5,2.5) -- (25.5,0.5);
		\draw[xshift=0.1cm, >=stealth ] [-]  (26,0) -- (32,0);
		\draw[xshift=0.1cm]    (33.5,0) node {$u_2$};
		\draw[xshift=0.1cm, >=stealth ] [-]  (32,0) -- (32,-2.5);
		\draw[ thick, xshift=0.1cm]  (30,-4.5) rectangle +(4,2);
		\draw [xshift=0.1cm](32,-3.5)   node {{${\bf \Phi_3(\lambda)}$}};
		\draw[ xshift=0.1cm,  >=stealth] [->] (32,-4.5) -- (32,-6);
		\draw[ xshift=0.1cm ]  (32,-6.5) circle(0.5);
		\draw[ xshift=0.1cm,  >=stealth] [->] (12.5,-2.5) -- (12.5,-0.5);
		\draw[ thick, xshift=0.1cm]  (10.5,-4.5) rectangle +(4,2);
		\draw [xshift=0.1cm](12.5,-3.5)   node {{${\bf \Phi_1(\lambda)}$}};
		\draw[ xshift=0.1cm,  >=stealth] [->] (12.5,-6.5) -- (12.5,-4.5);
		\draw[ xshift=0.1cm,  >=stealth] [-]  (12.5,-6.5) -- (31.5,-6.5);
		\draw[ xshift=0.1cm ] [->] (34.5,-6.5) -- (32.5,-6.5);
		\draw[ thick, xshift=0.1cm]  (34.5,-7.5) rectangle +(4,2);
		\draw [xshift=0.1cm](36.5,-6.5)   node {{${\bf \Gamma_3(\lambda)}$}};
		\draw[ xshift=0.1cm ] [->] (40.5,-6.5) -- (38.5,-6.5);
		\draw[xshift=0.1cm]    (41.5,-6.5) node {$z_3$};
		\draw[xshift=0.1cm]    (25.5,-5.75) node {$u_3$};
	\end{tikzpicture}
	\caption{Three-hop network architecture}
	\label{fig:three_hop}
	\hrulefill
\end{figure}
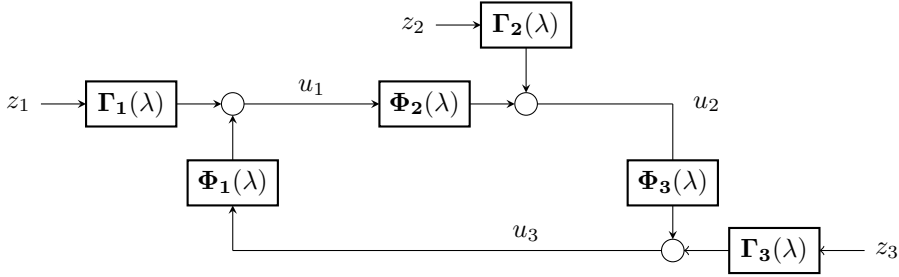

We now showcase a practical example, in order to illustrate the inherent potential of NRF representations and to provide motivation for the means (which will be discussed in the sequel) of obtaining them from the various coprime factorizations of the network's TFM.

We will focus on the 3-hop ring network in Fig.~\ref{fig:three_hop}, where all the $\mathbf \Phi_i(\lambda)$  and $\mathbf \Gamma_i(\lambda)$ blocks represent scalar TFMs of LTI systems, with all $\mathbf \Gamma_i(\lambda)$ blocks being strictly proper. We will denote by $\mathbf G(\lambda)$ the TFM from all the input signals $u_i$ to all the outputs $y_i$. By directly inspecting the signal flow graph in Fig.~\ref{fig:three_hop}, we can write 
\begin{equation}
	\label{RingRing}
	\small\begin{bmatrix} \mathbf{U}_1(\lambda) \\ \mathbf{U}_2(\lambda) \\ \mathbf{U}_3(\lambda) \end{bmatrix}  \hspace{-1mm}=\hspace{-1mm} \small\begin{bmatrix} O & O & \mathbf \Phi_{1}(\lambda) \\ 
		\mathbf \Phi_{2}(\lambda) & O  & O\\ O &  \mathbf \Phi_{3}(\lambda) &O\end{bmatrix} \normalsize\small\begin{bmatrix} \mathbf{U}_1(\lambda) \\  \mathbf{U}_2(\lambda) \\ \mathbf{U}_3(\lambda) \end{bmatrix} + \small\begin{bmatrix} \mathbf \Gamma_{1}(\lambda) & O&O\\ 
		O & \mathbf \Gamma_{2}(\lambda) &O\\ O&O&\mathbf \Gamma_{3}(\lambda) \end{bmatrix}   \normalsize\small\begin{bmatrix}  \mathbf{Z}_1(\lambda) \\ \mathbf{Z}_2(\lambda) \\ \mathbf{Z}_3(\lambda) \end{bmatrix},\normalsize
\end{equation}
where $\mathbf{U}_i(\lambda)$ and $\mathbf{Y}_i(\lambda)$ are the Laplace or $\mathcal{Z}$-transforms of $u_i$ and, respectively, of $y_i$, while also introducing the notation
	\begin{equation}
	\label{lokala1}
	\mathbf \Phi(\lambda):=\small\begin{bmatrix} O & O & \mathbf \Phi_{1}(\lambda) \\ 
		\mathbf \Phi_{2}(\lambda) & O  & O\\ O &  \mathbf \Phi_{3}(\lambda) &O\end{bmatrix}, \quad  \mathbf \Gamma(\lambda):=  \small\begin{bmatrix} \mathbf \Gamma_{1}(\lambda) & O&O\\ 
		O & \mathbf \Gamma_{2}(\lambda) &O\\ O&O&\mathbf \Gamma_{3}(\lambda) \end{bmatrix},\normalsize
\end{equation}
in order to state that $\big( \mathbf \Phi(\lambda),\mathbf \Gamma(\lambda) \big)$ is the NRF pair associated with the TFM $\mathbf G(\lambda)$.

\begin{remark}
	A noteworthy observation consists in the fact that the structure of the subsystem interconnections in Fig.~\ref{fig:three_hop} is no longer recognizable from the input-output relation described by the TFM of the system, $\displaystyle  \mathbf G(\lambda)=\big( I - \mathbf \Phi(\lambda) \big)^{-1}\mathbf \Gamma(\lambda)$,  since $\mathbf G(\lambda)$ does not have any sparsity pattern, nor (in general) does it have any other particularities. Notice, however, that the aforementioned structure is successfully captured in the sparsity patterns of $\mathbf \Phi(\lambda)$ and $\mathbf \Gamma(\lambda)$. This remarkable property makes the NRF a well-suited theoretical instrument for modeling any LTI network.
\end{remark}

Additionally, we wish to further illustrate how the NRF determines, via \eqref{NRF}, the topology of the LTI network describing the given TFM $\mathbf G(\lambda)$.  By considering $\mathbf \Phi_{13}(\lambda)$ identically zero, which would imply the ``breaking'' of the ring network from Fig.~\ref{fig:three_hop}, the latter becomes a cascade connection and $\mathbf G(\lambda)$ can be seen as a ``line'' network. This is relevant in the motion control of vehicles moving in ``platoon'' formation \cite{plutonizare}.

\subsection{System Response-Type Realizations}

With the utility of NRF representation now established, we return to the pair of TFMs introduced just before \eqref{WVin5}, and which were employed in obtaining \eqref{Q}-\eqref{P}. Given the need to impose a certain sparsity structure on the $\mathbf{\Phi}(\lambda)$ term of an NRF pair, in order to facilitate distributed implementations (as discussed in Remark~\ref{rem:NRF_MIMO}), and the opaqueness of the NRF formalism with respect to state dynamics, it becomes desirable to investigate a new type of representation that alleviates these limitations. It is in this context that we introduce one of the central objects of this paper, via the following definition.

\begin{definition}\label{obiectu} Let the system described by the realization \eqref{ss2a}-\eqref{ss2b}, of order $n$, have a TFM be denoted $\mathbf G(\lambda)$. Then, a pair $\big( \mathbf W(\lambda),\mathbf V(\lambda) \big)$ of {\em proper} TFMs, where $\mathbf W(\lambda)\in\mathbb{R}^{p \times p}(\lambda)$ and $\mathbf V(\lambda)\in\mathbb{R}^{p \times m}(\lambda)$, with $\begin{bmatrix}\mathbf{W}(\lambda) &\mathbf{V}(\lambda)\end{bmatrix}$ having a McMillan degree of at most $n-p$ and which satisfy
	\begin{equation} \label{WV}
		\mathbf G(\lambda) = \big( \lambda I_p -\mathbf W(\lambda) \big)^{-1}\mathbf V(\lambda),
	\end{equation}
	is called a {\em System Response-Type Realization} of the system given in \eqref{ss2a}-\eqref{ss2b}.
\end{definition}

\begin{remark}\label{rem:implem}
	It is worthwhile to point out that the two TFMs which make up the SRTR pair denoted $\big( \mathbf W(\lambda),\mathbf V(\lambda) \big)$ fulfill the same role as the TFMs $\widetilde{\mathbf{R}}^+(z)$ and $\widetilde{\mathbf{N}}(z)$ from the state iteration equation of (18) from \cite{SLS}. Denoting by $\mathbf{B}(z)$ the $\mathcal{Z}$-transform of the state vector belonging to the controller, the aforementioned state iteration is expressed in the same manner as \eqref{WVin5} via the identity
	\begin{equation}\label{eq:SLS_implem}
		z\mathbf{B}(z) = \widetilde{\mathbf{R}}^+(z)\mathbf{B}(z)+\widetilde{\mathbf{N}}(z)\mathbf{Y}(z).
	\end{equation}
	It is this very fact which brought about the designation of \emph{System Response-Type Realization} and which can be exploited to obtain sparse and scalable representations for systems in the context of distributed control. Note that, in contrast to the implementations from \cite{SLS}, the one from \eqref{WVin5} is also available for continuous-time systems.
\end{remark}

\begin{remark}
	In the sequel, SRTR pairs of controllers will be employed to implement control laws in a distributed fashion. However, in contrast to the approach from \cite{SLS} and its implementations of type \eqref{eq:SLS_implem}, we point out that closed-loop maps are not intrinsically linked to the definition of an SRTR pair. Therefore, the newly introduced representation bears all the hallmarks of the so-called \emph{implementation matrices} from \cite{Separ}, in which a mutation of the classical SLS paradigm allows for the decoupling of controller synthesis from closed-loop design. Although the framework proposed in the sequel shares this benefit, we continue to refer to SRTR pairs as such due to \eqref{eq:SLS_implem} being the inspiration for the proposed representation and, more importantly, due to not wishing to overshadow their significant potential for system analysis, which would inevitably result from naming them ``implementation matrices''.
\end{remark}

\begin{remark} \label{sprs} Another feature of the SRTR representation is that any pair $\big(\hspace{-0.25mm} \mathbf W\hspace{-0.25mm}(\lambda),$ $\mathbf V(\lambda) \big)$ possesses the same sparsity pattern with its subsequent NRF pair $\big( \mathbf \Phi(\lambda), \mathbf \Gamma(\lambda) \big)$, obtained via \eqref{Q}-\eqref{P}. Therefore, $\mathbf W(\lambda)$ is lower triangular if and only if $\mathbf \Phi(\lambda)$ is lower triangular and, similarly, $\mathbf V(\lambda)$ is tridiagonal if and only if $\mathbf \Gamma(\lambda)$ is tridiagonal.
\end{remark}

In light of Remark~\ref{sprs} and the TFMs from \eqref{Q}-\eqref{P}, it becomes clear that, in order to obtain a rich parametrization of factored representations, for a given $\mathbf G(\lambda)$ and having a particular sparsity pattern, it suffices to study the set of all SRTR pairs $\big( \mathbf W(\lambda),\mathbf V(\lambda) \big)$ associated with the investigated network. This can be achieved via the following result, that will prove critical to many of the results presented in the sequel.

\begin{theorem} \label{Main} For the system described by the realization given in \eqref{ss2a}-\eqref{ss2b}, we have that:
	
	\begin{enumerate}
		\item[$\mathbf{a)}$] An SRTR pair $\big( \mathbf W(\lambda),\mathbf V(\lambda) \big)$ associated with the aforementioned system is expressed via
	\end{enumerate}
	\begin{equation} \label{VWdef}
		\small\begin{bmatrix}
			\mathbf W(\lambda) &\hspace{-2mm} \mathbf V(\lambda)
		\end{bmatrix}= \left[\begin{array}{c|cc}
			A_{22}+ K A_{12}  -\lambda I_{n-p}  & K A_{11} - KA_{12}K + A_{21} -A_{22}K   &  KB_1+B_2  \\ 
			\hline  A_{12}   & A_{11} - A_{12}K & B_1  \\ \end{array}\right],
	\end{equation}
	\begin{enumerate}
		\item[\phantom{$\mathbf{a)}$}] where $K$ is any matrix  in $\mathbb{R}^{(n-p) \times p}$;
		
		\item[$\mathbf{b)}$] All SRTR pairs associated with the system described by \eqref{ss2a}-\eqref{ss2b} are expressed via \eqref{VWdef} for some $K\in\mathbb{R}^{(n-p)\times p}$ if and only if the aforementioned $n$-th order realization is minimal.
	\end{enumerate}

\end{theorem}
\begin{proof} See Appendix~\ref{sec:ap_D}.
\end{proof}

\begin{remark} \label{stable}
	The poles of $\mathbf W(\lambda)$ and $\mathbf V(\lambda)$ can be placed at will in the complex plane, by a suitable choice of the matrix $K$ and the observability of the  pair $(A_{12}, A_{22})$ (recalling Assumption~\ref{Observability}). Thus, when implementing a controller with a realization of type \eqref{ss2a}-\eqref{ss2b} via its SRTR pair, it may be possible to obtain a distributed control law of type \eqref{WVin5} using TFMs whose entries are all stable. If this is the case, then this implementation can be employed even for a plant which is not strongly stabilizable, \emph{i.e.}, there exist no (centralized) controllers having stable TFMs which render all closed-loop maps stable. This possibility will be investigated in the sequel.
\end{remark}

\begin{figure*}
	
	\begin{center}
		\begin{tikzpicture}[every node/.style={transform shape},scale=.70]
			
			\node(n1)[minimum height=0.5cm, minimum width=0.5cm,draw]at(0,0) {$
				\begin{array}{ccl}
					\begin{bmatrix} \sigma{y}  \\ \sigma{q} \end{bmatrix}  &\hspace{-2mm}=\hspace{-2mm}& \begin{bmatrix}  A_{11} & A_{12} \\ A_{21} & A_{22}  \end{bmatrix}  \begin{bmatrix} y   \\ q  \end{bmatrix} + \begin{bmatrix} B_1 \\ B_2 \end{bmatrix} u\\
					y  &\hspace{-2mm}=\hspace{-2mm}& \begin{bmatrix} \hspace{2mm}I_p & \hspace{1.5mm}O\hspace{2mm} \end{bmatrix}\hspace{0.5mm} \begin{bmatrix} y  \\ q \end{bmatrix}
				\end{array}
				$
			};
			
			\node(n2)[minimum height=0.5cm, minimum width=0.5cm,draw]at(11,0) {$
				\begin{array}{l}
					\lambda \mathbf{Y}(\lambda) =  \mathbf{W}(\lambda) \mathbf{Y}(\lambda) +  \mathbf{V}(\lambda) \mathbf{U}(\lambda)\\\phantom{}\vspace{-2mm}\\
					\big(\lambda I_p-\mathbf{W}(\lambda),\mathbf{V}(\lambda)\big)\text{ FLCF}
				\end{array}
				$
			};
			
			\node(n3)[minimum height=0.5cm, minimum width=0.5cm,draw]at(11,-4) {$
				\begin{array}{lcl}
					\mathbf{Y}(\lambda) = \mathbf \Phi(\lambda) \mathbf{Y}(\lambda) + \mathbf \Gamma(\lambda) \mathbf{U}(\lambda)\\\phantom{}\vspace{-2mm}\\
					\mathbf{\Phi}_{ii}(\lambda)=0,\ \forall\, 1\leq i\leq p
				\end{array}
				$
			};
			
			\node(n4)[minimum height=0.5cm, minimum width=0.5cm,draw]at(0,-4) {$
				\begin{array}{l}
					\mathbf{Y}(\lambda)=\mathbf{M}^{-1}(\lambda)\mathbf{N}(\lambda)\mathbf{U}(\lambda)\\\phantom{}\vspace{-2mm}\\
					\big(\mathbf{M}(\lambda),\mathbf{N}(\lambda)\big)\text{ LCF over $\mathbb{S}$}
				\end{array}
				$
			};
			
			\draw[-latex](2.95,0.2)--(8.1,0.2) node[midway, above]{Theorem~\ref{Main}};
			\draw[-latex](8.1,-0.2)--(2.95,-0.2) node[midway, below]{\eqref{eq:SRTR_init}-\eqref{eq:idGr}};
			\draw[-latex](-0.2,-0.95)--(-0.2,-3.35) node[midway, left]{Theorem~\ref{SIMAX99}$\ $};
			\draw[-latex](0.2,-3.35)--(0.2,-0.95) node[midway, right]{$\begin{array}{c}
					\eqref{Kontroller}\\\text{via}\\\eqref{dc2}
				\end{array}$};
			\draw[-latex](8.55,-0.65)--(2.35,-3.6) node[midway, right]{$\quad\quad$Theorem~\ref{step2}};
			\draw[-latex](1.9,-3.35)--(8.1,-0.4) node[midway, left]{Theorem~\ref{inversa}$\quad\quad$};
			\draw[-latex](11,-0.65)--(11,-3.35) node[midway, right]{$\ $\eqref{Q}-\eqref{P}};
			\draw[-latex](2.3,-4)--(8.3,-4) node[midway, below]{(9) in \cite{NRF}};
			
		\end{tikzpicture}
	\end{center}
	\caption{Diagram showing the means of switching between various LTI network representations}\label{fig:thm_diagram}
	\hrulefill\vspace{-2mm}
\end{figure*}
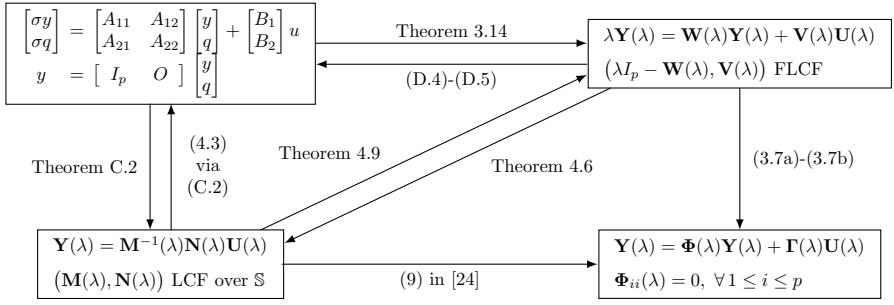

\section{Main Results}\label{sec:main}

The ultimate goal of this line of research is the computation of controllers whose factorized representation possesses a desired sparsity pattern or which inherits the structure of a closely related factorization, from which we may shift by employing numerically reliable procedures. This would allow us, for instance, to compute controllers that can be implemented as a ``ring'' network (recall Fig.~\ref{fig:three_hop}) or as a ``line'' network, which is important for motion control of vehicles moving in ``platoon'' formation (see \cite{plutonizare}). A key insight in this regard is the fact that the most prominent classical results in control theory, such as the Youla Parameterization, state the expression of a stabilizing controller via a coprime factorization over $\mathbb{S}$ of its TFM.  

As a further step towards employing Youla-like methods for the synthesis of controllers featuring either structured NRF or SRTR pairs, we need to understand the connections between the (stable) left coprime factorizations (of a given stabilizing controller) and its NRF- or SRTR-based factorizations. To help in keeping track of these various representations, the diagram displayed in Fig.~\ref{fig:thm_diagram} at the top of this page illustrates how the various results presented in this paper interact with one another.

\begin{remark}
	A poignant insight into the relationship between the SRTR and NRF representations is hinted at in Fig.~\ref{fig:thm_diagram}. As can be seen from the indicated diagram, an SRTR pair bears much stronger connections with a network's state dynamics than its NRF counterpart. Indeed, it is this very affinity which serves to motivate the study of SRTR representations, in contrast to the more input-output oriented perspective offered by an NRF pair. This also makes SRTR pairs a far more appealing tool for system analysis than their correspondents in the NRF formalism, as the latter was primarily developed with the aim of controller synthesis (see \cite{NRF} and \cite{aug_sparse}).  
\end{remark}

\subsection{Coprimeness of SRTR Pairs} In this subsection, we will prove that (by chance, rather than by design), for any SRTR  $\big(\mathbf W(\lambda),\mathbf V(\lambda) \big)$ associated with a system described by a minimal realization \eqref{ss2a}-\eqref{ss2b} and having $\mathbf G(\lambda)$ as its TFM, it follows that $\displaystyle  \big(\lambda I_p -\mathbf W(\lambda), \mathbf V(\lambda)\big)$  is an FLCF of $\mathbf G(\lambda)$. Before doing so, we recall the pivotal role of controllability, as established by point $\mathbf{b)}$ in Theorem~\ref{Main} and, from now onward, we assume that the following statement holds.

\begin{assumption} \label{Controllability}
	  The realization \eqref{ss2a}-\eqref{ss2b} is controllable.
\end{assumption} 

\begin{remark}
	A direct consequence of point $\mathbf{b)}$ in Theorem~\ref{Main} is that the set of all SRTR pairs that can be obtained from a system described by a minimal realization \eqref{ss2a}-\eqref{ss2b}, and having a TFM denoted $\mathbf{G}(\lambda)$, coincides with the set of all SRTR pairs expressed via \eqref{VWdef} for any other system described by a minimal realization of the same $\mathbf{G}(\lambda)$. Given the fact that, in the sequel, we consider only realizations of type \eqref{ss2a}-\eqref{ss2b} which are minimal, the subsequently presented SRTR-based results are free to employ any minimal realization of type \eqref{ss2a}-\eqref{ss2b}, for a given TFM.
\end{remark}

As previously discussed, coprimeness is an important property for output feedback stabilization, since  classical results require stable coprime factors of the plant, in addition to the stable coprime factors of the controllers. Therefore, we now characterize a class of SRTR-based coprime factorizations via the following result. 

\begin{theorem} \label{CupruMin} Consider a system given by a minimal realization \eqref{ss2a}-\eqref{ss2b}, having $\mathbf G(\lambda)$ as its TFM. Then, for any SRTR pair $\big(\mathbf W(\lambda),\mathbf V(\lambda) \big)$ of the latter system expressed as in \eqref{VWdef}, we have that the pair $\displaystyle \big(\lambda I_p -\mathbf W(\lambda),\mathbf V(\lambda)\big)$ is an FLCF of $\mathbf G(\lambda)$.
\end{theorem}
\begin{proof} See Appendix~\ref{sec:ap_D}.
\end{proof}

\begin{remark} \label{nevercop}
	While any SRTR pair $\big(\mathbf W(\lambda),\mathbf V(\lambda) \big)$ produced by point $\mathbf{b)}$ of Theorem~\ref{Main} and associated with the system having $\mathbf{G}(\lambda)$ as its TFM offers an FLCF $\mathbf G(\lambda)=\big( \lambda I_p- \mathbf W(\lambda)\big)^{-1}\mathbf V(\lambda)$, the NRF representation $\displaystyle \mathbf G(\lambda) = \big( I_p -\mathbf \Phi(\lambda)\big)^{-1}\mathbf \Gamma(\lambda)$ is \emph{not guaranteed} to be fully left coprime. This is owed to the fact that some of the finite poles belonging to $\big(\lambda I_p - \mathbf D(\lambda)\big)$ may not cancel out when multiplying its inverse with $\begin{bmatrix} \lambda I_p -\mathbf W(\lambda) &  \mathbf V(\lambda) \end{bmatrix}$ to obtain, via \eqref{Q}-\eqref{P}, the fact that
	\begin{equation*}
		\begin{bmatrix}
			I_p-\mathbf{\Phi}&\mathbf{\Gamma}
		\end{bmatrix}=\big(\lambda I_p - \mathbf D(\lambda)\big)^{-1}\begin{bmatrix} \lambda I_p -\mathbf W(\lambda) &  \mathbf V(\lambda) \end{bmatrix}.
	\end{equation*}
	We point out that, due to the diagonal structure of $\big(\lambda I_p - \mathbf D(\lambda)\big)^{-1}$, the statements from Remark~\ref{sprs} regarding the sparsity pattern of the resulting NRF are reinforced.
\end{remark}

The inherent coprimeness of SRTR pairs turns out to be an important property, allowing for the commutation between these pairs and the classical coprime factorizations of the network's model. Building upon this property, we proceed to investigate in the sequel reliable means of shifting between these two sets of coprime representations. 

\subsection{Getting from SRTR pairs to Left Coprime Factorizations}\label{subsec:s2c}
In this subsection, we show that any SRTR pair $\big( \mathbf W(\lambda), \mathbf V(\lambda) \big)$ with both $\mathbf W(\lambda)$ and $\mathbf V(\lambda)$ having no poles located outside of $\mathbb{S}$ can produce an entire class of factorizations which are left coprime over $\mathbb{S}$. Notice that, for any SRTR pair $\big(\mathbf W(\lambda), \mathbf V(\lambda) \big)$, the TFM $\begin{bmatrix} \lambda I_p -\mathbf W(\lambda) &\hspace{-2mm}  \mathbf V(\lambda) \end{bmatrix}$ is {\em improper} and has exactly $p$ poles at infinity of multiplicity one, and recall that the poles of $\mathbf W(\lambda)$ and $\mathbf V(\lambda)$ can be placed at will, as per Remark~\ref{stable}.

We aim to showcase the means of shifting from a stable SRTR pair $\big(\mathbf W(\lambda),\mathbf V(\lambda) \big)$ of the system having $\mathbf G(\lambda)$ as its TFM to an LCF over $\mathbb{S}$ of $\mathbf G(\lambda) = \mathbf M^{-1}(\lambda)\mathbf N(\lambda)$. Crucial to this endeavor, we provide in Appendix~\ref{sec:ap_C} the characterization for a class of LCFs over $\mathbb{S}$ for our network's TFM. It is this very result which we now employ, in order to obtain a procedure which extracts an LCF over $\mathbb{S}$ from a stable SRTR pair.

\begin{theorem}\label{step2} Consider a system given by a minimal realization \eqref{ss2a}-\eqref{ss2b}, having $\mathbf G(\lambda)$ as its TFM, and let $\big(\mathbf W(\lambda),\mathbf V(\lambda) \big)$ be a stable SRTR pair associated with this system. Define
	\begin{equation} \label{Theta}
		\mathbf \Theta(\lambda):= \left[\begin{array}{c|c} A_x  -\lambda I_p & B_x \\
			\hline C_x   & O \end{array}\right],
	\end{equation}
	\noindent for any $A_x\in\mathbb{R}^{p\times p}$ having all its eigenvalues located in $\mathbb{S}$, along with any $B_x\in\mathbb{R}^{p\times p}$ and any $C_x\in\mathbb{R}^{p\times p}$ which are both invertible. Then, we have that:
	\begin{enumerate}
		\item[$\mathbf{a)}$] The following TFM\end{enumerate}
		\begin{equation}\label{eq:multip}
			\begin{bmatrix} \mathbf M(\lambda) &  \mathbf N(\lambda) \end{bmatrix}{=}\  \mathbf \Theta(\lambda) \begin{bmatrix}  \lambda I_p -\mathbf W(\lambda)  &  \mathbf V(\lambda) \end{bmatrix}\normalsize
		\end{equation}\begin{enumerate}
		\item[\phantom{$\mathbf{a)}$}] designates an LCF over $\mathbb{S}$ of $\mathbf G(\lambda)=  \mathbf M^{-1}(\lambda) \mathbf N(\lambda)$;
		
		\item[$\mathbf{b)}$] The multiplication with $\mathbf{\Theta}(\lambda)$ from \eqref{eq:multip} preserves the poles of the SRTR pair in the $\mathbf M(\lambda)$ and $\mathbf N(\lambda)$ factors.
	\end{enumerate}

\end{theorem}

\begin{proof} See Appendix~\ref{sec:ap_D}.
\end{proof}

\begin{remark} \label{buna} Notice that, for any \emph{diagonal} $A_x$ having all its eigenvalues located in $\mathbb{S}$, the factor $\mathbf \Theta(\lambda)=(\lambda I_p - A_x)^{-1}$ can be used to produce an LCF over $\mathbb{S}$ of $\mathbf G(\lambda)$ with precisely the same sparsity structure as the SRTR pair $\big(\mathbf W(\lambda),\mathbf V(\lambda) \big)$.
\end{remark}

\subsection{Getting from Left Coprime Factorizations to SRTR pairs}\label{subsec:c2s}

In this subsection, we show that, given an LCF over $\mathbb{S}$ of $\mathbf G(\lambda) = \mathbf M^{-1}(\lambda)\mathbf N(\lambda)$, which is described by a realization of type \eqref{ss2a}-\eqref{ss2b} and for which $\begin{bmatrix} \mathbf M(\lambda) &  \mathbf N(\lambda) \end{bmatrix}$ has the same McMillan degree as $\mathbf{G}(\lambda)$, it is possible to use a certain type of realization for the aforementioned factorization, in order to retrieve a stable SRTR pair  $\big(\mathbf W(\lambda),\mathbf V(\lambda) \big)$.

To begin with, recall from the previous subsection that Appendix~\ref{sec:ap_C} includes a characterization for a particular class of LCFs over $\mathbb{S}$ for our network's TFM. Thus, we start with a given LCF over $\mathbb{S}$ of $\mathbf G(\lambda)$ from this class, where the latter is described by an $n$-th order minimal realization. In this case, let the given LCF over $\mathbb{S}$ have a state-space representation as in \eqref{dc2} from the Appendix~\ref{sec:ap_C}, to which we apply a state equivalence transformation described by $T\in \mathbb{R}^{n \times n}$, so that $\displaystyle CT^{-1} = \begin{bmatrix} I_p & O \end{bmatrix}$. Such a $T$ matrix always exists, due to Assumption~\ref{Regularity}, and it follows that \eqref{dc2} turns into
\begin{equation} \label{Kontroller}
	\begin{bmatrix}\mathbf M(\lambda) & \mathbf N(\lambda) \end{bmatrix} =  \footnotesize\left[\begin{array}{cc|cc} A_{11} +F_1  -\lambda I_p& A_{12}  & F_1& B_1 \\
		A_{21} +F_2 & A_{22} -\lambda I_{n-p} & F_2 & B_2 \\ 
		\hline U & O  & U & O  \end{array}\right].
	\normalsize
\end{equation}

\noindent
The key to retrieving stable SRTR pairs rests with a particular class of solutions to the continuous-time non-symmetric algebraic Riccati equation (CTNARE) given by
\begin{equation} \label{Riccati}
	K(A_{11}+F_1)    -KA_{12}K  +  (A_{21} + F_2) -A_{22}K  = O.
\end{equation}

As will be shown in the sequel, it is these solutions which underpin the correspondence between the set of systems from \eqref{Kontroller} and their related SRTR pairs. Thus, given their paramount importance, we now formally introduce these particular solutions.

\begin{definition}\label{def:stab}
	A matrix $K$ that is a solution of \eqref{Riccati} and which ensures that $(A_{11}+F_1-A_{12}K)$ has all its eigenvalues located in $\mathbb{S}$ is called a \emph{right stabilizing solution} of the CTNARE.
\end{definition} 

The existence of these solutions is shown in Appendix~\ref{sec:ap_B} to be a generic property (in light of Assumption~\ref{Observability}), with the aforementioned appendix offering means of reliably computing them. This being said, we proceed to employ these right stabilizing solutions, in order to state the central result of this subsection.

\begin{theorem} \label{inversa}
	Consider a system given by a minimal realization \eqref{ss2a}-\eqref{ss2b}, having $\mathbf G(\lambda)$ as its TFM, and such that $\mathbf G(\lambda)=\mathbf M^{-1}(\lambda)\mathbf N(\lambda)$ is an LCF over $\mathbb{S}$, having a realization of type \eqref{Kontroller}. Let also $K$ be a right stabilizing solution of the CTNARE from \eqref{Riccati} and denote $A_x:=A_{11} + F_1 - A_{12}K$. Then, we have that:
	\begin{enumerate}
		\item[$\mathbf{a)}$] The realization 
	\end{enumerate}
		\begin{equation} \label{MN1}
		\begin{bmatrix} \mathbf M(\lambda) &  \hspace{-2mm}\mathbf N(\lambda) \end{bmatrix} =  \left[\scriptsize\begin{array}{cc|cc} \hspace{-1mm}A_x  -\lambda I_{p}  &  A_{12}  & A_x - A_{11} +A_{12}K   & B_1 \\
			O & \hspace{-2mm}A_{22}+KA_{12} -\lambda I_{n-p} & KA_{12}K+A_{22}K-KA_{11}-A_{21}  & \hspace{-2mm}KB_1+B_2 \\ 
			\hline U & O & U  & O \\
		\end{array}\right] \normalsize
	\end{equation}\normalsize
	\begin{enumerate}
		\item[\phantom{$\mathbf{a)}$}] is equivalent to the one from \eqref{Kontroller};
		\item[$\mathbf{b)}$] A stable SRTR pair $\big( \mathbf W(\lambda),\mathbf V(\lambda) \big)$ of the system can be obtained, from the available LCF over $\mathbb{S}$, via
	\end{enumerate}
		\begin{equation}\label{eq:lcf2srtr}
			\begin{bmatrix}
				\mathbf W(\lambda) &\mathbf V(\lambda)
			\end{bmatrix}=\begin{bmatrix}
				\lambda I_p & O
			\end{bmatrix}+(\lambda I_p-A_x)U^{-1}\begin{bmatrix}
				-\mathbf{M}(\lambda)&\mathbf{N}(\lambda)
			\end{bmatrix};
		\end{equation}
		\begin{enumerate}
		\item[$\mathbf{c)}$] The stable SRTR pair from \eqref{VWdef}, with $K$ taken as a right stabilizing solution of \eqref{Riccati}, is identical to the stable SRTR pair from \eqref{eq:lcf2srtr} for the same $K$.
	\end{enumerate}
\end{theorem}
\begin{proof}  See Appendix~\ref{sec:ap_D}.
\end{proof}

\begin{remark}
	The realization of the stable SRTR pair, mentioned in point $\mathbf{c)}$ of Theorem~\ref{inversa}, can also be obtained from \eqref{MN1}. Notice that the state-space representation from \eqref{VWdef} can be formed by simply extracting the appropriate submatrices which make up the realization given in \eqref{MN1}.
\end{remark}

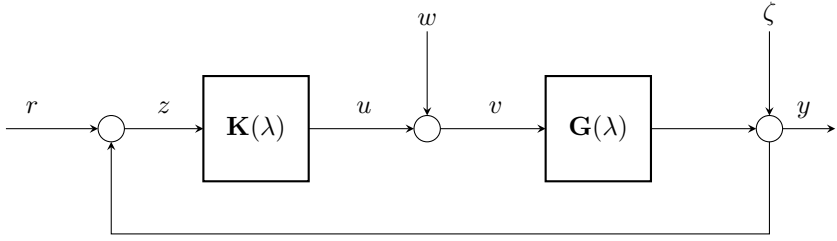
\begin{figure}[t]
	\centering
	\begin{tikzpicture}[scale=0.35]
		\draw[xshift=0.1cm, >=stealth ] [->] (0,0) -- (3.5,0);
		\draw[ xshift=0.1cm ]  (4,0) circle(0.5);
		\draw[xshift=0.1cm] (3,1)   node {\bf{ }} (1,0.8) node {$r$};
		\draw [xshift=0.1cm](6,0.8)   node {$z$} ;
		\draw[ xshift=0.1cm,  >=stealth] [->] (4.5,0) -- (7.5,0);
		\draw[ thick, xshift=0.1cm]  (7.5,-2) rectangle +(4,4);
		\draw [xshift=0.1cm](9.5,0)   node {{${\bf K(\lambda)}$}} ;
		\draw[ xshift=0.1cm,  >=stealth] [->] (11.5,0) -- (15.5,0);
		\draw[ xshift=0.1cm ]  (16,0) circle(0.5cm);
		\draw [xshift=0.1cm](13.6,0.8)   node {$u$} ;
		\draw [xshift=0.1cm](18.6,0.8)   node {$v$} ;
		\draw [xshift=0.1cm] (14.8,0.7)   node {\bf{ }};
		\draw[  xshift=0.1cm,  >=stealth] [->] (16,3.7) -- (16,0.5);
		\draw [xshift=0.1cm] (16,3.6)  node[anchor=south] {$w$}  (15.3,1.5)  node {\bf{ }};
		\draw[  xshift=0.1cm,  >=stealth] [->] (16.5,0) -- (20.5,0);
		\draw[ thick, xshift=0.1cm ]  (20.5,-2) rectangle +(4,4) ;
		\draw [xshift=0.1cm] (22.5,0)   node {{${\bf G(\lambda)}$}} ;
		\draw[ xshift=0.1cm,  >=stealth] [->] (24.5,0) -- (28.5,0);
		\draw[ xshift=0.1cm ] (29,0)  circle(0.5);
		\draw [xshift=0.1cm] (29,3.6)  node[anchor=south] {$\zeta$}  (28.5,1.5)  node {\bf{ }};
		\draw [xshift=0.1cm] (28,0.7)   node {\bf{ }};
		\draw [xshift=0.1cm] (30.3,0.7)   node {$y$};
		\draw[  xshift=0.1cm,  >=stealth] [->] (29.5,0) -- (31.5,0);
		\draw[  xshift=0.1cm,  >=stealth] [->] (29,3.7) -- (29,0.5);
		\draw[ xshift=0.1cm,  >=stealth] [->] (29,-0.5) -- (29,-4) -- (4,-4)-- (4, -0.5);
		\draw [xshift=0.1cm] (3.3,-1.3)   node {};
		\useasboundingbox (0,0.1);
	\end{tikzpicture}
	\caption{Feedback loop of the plant $\bf G$ with the controller $\bf K$}
	\label{fig:2Block}
	\hrulefill
\end{figure}

\subsection{Stable and Stabilizing SRTR Implementations}

In this subsection, we elaborate upon Remark~\ref{stable} and present the means of implementing distributed controllers via their SRTR representations. The results given here are, in many ways, the analogues of Theorems III.4 and III.6 from \cite{NRF}, thus highlighting the connections being showcased in Fig.~\ref{fig:thm_diagram}. Before doing so, we introduce the following definitions.

\begin{definition}\label{def:int_stab}
	Let $\mathbf{K}\in\mathbb{R}(\lambda)^{p\times m}$ and $\mathbf{G}\in\mathbb{R}(\lambda)^{m\times p}$. A closed-loop interconnection (such as the one in Fig.~\ref{fig:2Block}) between $\mathbf{K}(\lambda)$ and $\mathbf{G}(\lambda)$ for which the closed-loop maps from all exogenous signals ($r$, $w$ and $\zeta$) to all internal ones ($y$, $u$, $z$ and $v$) are stable TFMs is called \emph{internally stable}.
\end{definition}

\begin{definition}
	A system, described by its TFM $\mathbf{K}\in\mathbb{R}(\lambda)^{p\times m}$, which ensures that the closed-loop interconnection with $\mathbf{G}\in\mathbb{R}(\lambda)^{m\times p}$ from Fig.~\ref{fig:2Block} is internally stable, is termed a \emph{stabilizing controller} for $\mathbf{G}(\lambda)$.
\end{definition}

We start by proving that the specialized implementations of type \eqref{WVin5}, corresponding to the TFMs of stabilizing controllers, are themselves stabilizing.

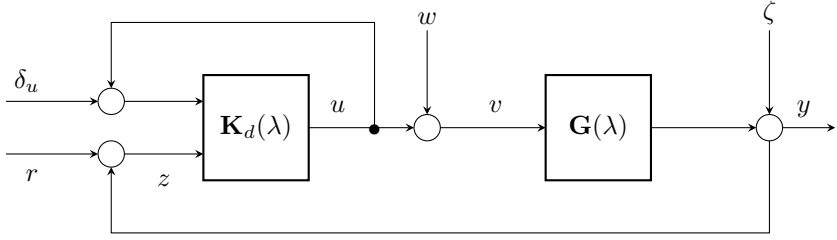
\begin{figure}[t]
	\centering
	\begin{tikzpicture}[scale=0.35]
		\draw[xshift=0.1cm, >=stealth ] [->] (0,-1) -- (3.5,-1);
		\draw[ xshift=0.1cm ]  (4,-1) circle(0.5);
		\draw[xshift=0.1cm] (3,1)   node {\bf{ }} (1,-1.8) node {$r$};
		\draw [xshift=0.1cm](6,-2)   node {$z$} ;
		\draw[ xshift=0.1cm,  >=stealth] [->] (4.5,-1) -- (7.5,-1);
		\draw[ thick, xshift=0.1cm]  (7.5,-2) rectangle +(4,4);
		\draw [xshift=0.1cm](9.5,0)   node {{$\mathbf{K}_{d}(\lambda)$}} ;
		\draw[ xshift=0.1cm,  >=stealth] [->] (11.5,0) -- (15.5,0);
		\draw[ xshift=0.1cm ]  (16,0) circle(0.5cm);
		\draw [xshift=0.1cm](12.6,0.8)   node {$u$} ;
		\draw[  xshift=0.1cm,  >=stealth] [->] (14,0) -- (14,4)--(4,4)--(4,1.5);
		\draw [xshift=0.1cm](18.6,0.8)   node {$v$} ;
		\draw [xshift=0.1cm](14,-0.1)   node {$\bullet$} ;
		\draw [xshift=0.1cm] (14.8,0.7)   node {\bf{ }};
		\draw[  xshift=0.1cm,  >=stealth] [->] (16,3.7) -- (16,0.5);
		\draw [xshift=0.1cm] (0.8,1)  node[anchor=south] {$\delta_u$};
		\draw[ xshift=0.1cm ]  (4,1) circle(0.5);
		\draw[  xshift=0.1cm,  >=stealth] [->] (4.5,1) -- (7.5,1);
		\draw[  xshift=0.1cm,  >=stealth] [->] (0,1) -- (3.5,1);
		\draw [xshift=0.1cm] (16,3.6)  node[anchor=south] {$w$}  (15.3,1.5)  node {\bf{ }};
		\draw[  xshift=0.1cm,  >=stealth] [->] (16.5,0) -- (20.5,0);
		\draw[ thick, xshift=0.1cm ]  (20.5,-2) rectangle +(4,4) ;
		\draw [xshift=0.1cm] (22.5,0)   node {{${\bf G(\lambda)}$}} ;
		\draw[ xshift=0.1cm,  >=stealth] [->] (24.5,0) -- (28.5,0);
		\draw[ xshift=0.1cm ] (29,0)  circle(0.5);
		\draw [xshift=0.1cm] (29,3.6)  node[anchor=south] {$\zeta$}  (28.5,1.5)  node {\bf{ }};
		\draw [xshift=0.1cm] (28,0.7)   node {\bf{ }};
		\draw [xshift=0.1cm] (30.3,0.7)   node {$y$};
		\draw[  xshift=0.1cm,  >=stealth] [->] (29.5,0) -- (31.5,0);
		\draw[  xshift=0.1cm,  >=stealth] [->] (29,3.7) -- (29,0.5);
		\draw[ xshift=0.1cm,  >=stealth] [->] (29,-0.5) -- (29,-4) -- (4,-4)-- (4, -1.5);
		\draw [xshift=0.1cm] (3.3,-1.3)   node {\bf{}};
		\useasboundingbox (0,0.1);
	\end{tikzpicture}
	\caption{Feedback loop of the plant ${\bf G}$ with the controller  ${\bf K}$ in an SRTR-based implementation $\lambda\mathbf{U}(\lambda)={\bf W}(\lambda) (\mathbf{U}(\lambda)+\mathbf{\Delta_U}(\lambda))+ {\bf V}(\lambda)\mathbf{Z}(\lambda)$}
	\label{fig:2BlockAgain}
	\hrulefill
\end{figure}

\begin{theorem}\label{thm:IO_stab}
		Let $\mathbf{K}\in\mathbb{R}(\lambda)^{p\times m}$ be a stabilizing controller for $\mathbf{G}\in\mathbb{R}(\lambda)^{m\times p}$, given by a minimal realization \eqref{ss2a}-\eqref{ss2b} and by a stable SRTR pair $\big(\mathbf W(\lambda),\mathbf V(\lambda) \big)$ associated with the latter via \eqref{VWdef}, as in Definition~\ref{obiectu}. Then, the implementation
	\begin{subequations}
		\begin{align}\label{eq:SRTR_ctl}
			\mathbf{K}_d(\lambda):=&\begin{bmatrix}
				\frac{1}{\lambda}\mathbf{W}(\lambda) & \frac{1}{\lambda}\mathbf{V}(\lambda)
			\end{bmatrix},\\
			\mathbf{U}(\lambda)=&\ \mathbf{K}_d(\lambda)\begin{bmatrix}
				\mathbf{U}(\lambda)^\top+\mathbf{\Delta}^\top_{\mathbf{U}}(\lambda)&\mathbf{Z}^\top(\lambda)
			\end{bmatrix}^\top,\label{eq:SRTR_conv}
		\end{align}
	\end{subequations}
	internally stabilizes the closed-loop interconnection in Fig.~\ref{fig:2BlockAgain}, where $\delta_u$ designates a bounded communication disturbance and $\mathbf{\Delta_U}(\lambda)$ denotes its Laplace or $\mathcal{Z}$-transform.
\end{theorem}

\begin{proof}
	See Appendix~\ref{sec:ap_D}.
\end{proof}

\begin{remark}\label{rem:discr_stab}
	Notice that the stabilizing and distributed controller from \eqref{eq:SRTR_ctl} always has a stable TFM in the discrete-time case, as previously anticipated in Remark~\ref{stable}. Recalling the parallels discussed in Remark~\ref{rem:implem}, this conclusion comes naturally and serves to further emphasize the close connections between the SRTR representations and the celebrated SLS framework \cite{SLS}, for discrete-time systems.
\end{remark}

\begin{remark}
	In the continuous-time case, the sparsity structure of the distributed controller from \eqref{eq:SRTR_ctl} is much more versatile than the ones obtained in \cite{Explicit}. We do not restrict $\mathbf{K}_d(\lambda)$ to possessing only banded nonzero entries and the generic sparsity structure it inherits from $\mathbf{W}(\lambda)$ and $\mathbf{V}(\lambda)$ offers more flexibility with respect to the network's communication infrastructure. Note that, in contrast to the results from \cite{Explicit}, we make no assumptions on the spatial invariance of the investigated networks.
\end{remark}

Although the TFM from \eqref{eq:SRTR_ctl} is guaranteed to be stable in the discrete-time case, we now show that this feature is unavailable for continuous-time controllers, when considering implementations obtained via SRTR pairs produced by Theorem~\ref{Main}.

\begin{proposition}\label{prop:not_stab}
	In the continuous-time case, $\mathbf{K}_d(s)$ from \eqref{eq:SRTR_ctl} is guaranteed to be an unstable TFM when it is formed from any stable SRTR pair, obtained from a minimal realization \eqref{ss2a}-\eqref{ss2b} of $\mathbf{K}(s)$ via \eqref{VWdef}, as considered in Theorem~\ref{thm:IO_stab}.\vspace{-1mm}
\end{proposition}

\begin{proof}
	See Appendix~\ref{sec:ap_D}.
\end{proof}

We now show that a particularly useful state-space implementation of \eqref{eq:SRTR_conv} can be used to drive the closed-loop system's state vector to the origin for any finite initial conditions. This critical result is formalized via the following theorem.

\begin{theorem}\label{thm:implem}
	Let $\mathbf{G}\in\mathbb{R}(\lambda)^{m\times p}$ be given by a stabilizable and detectable realization and let $\mathbf{K}\in\mathbb{R}(\lambda)^{p\times m}$ be a stabilizing controller of $\mathbf{G}(\lambda)$. Let $\mathbf{K}(\lambda)$ be given by a minimal realization \eqref{ss2a}-\eqref{ss2b} and by a stable SRTR pair $(\mathbf{W}(\lambda),\mathbf{V}(\lambda))$ associated with the latter via \eqref{VWdef}. Let $e_i$ be the $i^{\text{th}}$ vector in the canonical basis of $\mathbb{R}^{p}$. By implementing stabilizable and detectable realizations for $e_i^\top\mathbf{K}_d(\lambda)$, with $i\in1:p$, and by computing the command signals as in \eqref{eq:SRTR_conv}, the state dynamics of the closed-loop system from
	Fig.~\ref{fig:2BlockAgain} are made asymptotically stable (see Section 5.3 in \cite{zhou}).\vspace{-1mm}
\end{theorem}

\begin{proof}
	See Appendix~\ref{sec:ap_D}.
\end{proof}

\begin{remark}\label{rem:NRFlike}
	For the means to obtain realizations of the type described in Theorem~\ref{thm:implem}, consult Remark III.7 in \cite{NRF}. Fortuitously, Remark III.8 in \cite{NRF} also holds for block-row implementations (recall Remark~\ref{rem:NRF_MIMO}) of the control laws from \eqref{eq:SRTR_ctl}-\eqref{eq:SRTR_conv}. Additionally, the synthesis techniques described in \cite{NRF} and \cite{aug_sparse} can be employed in order to obtain the stabilizing controllers that are required by the results given in this subsection. Such a design scenario, in which the techniques from \cite{NRF} and \cite{aug_sparse} are used to obtain a stabilizing controller which generates sparse SRTR-based control laws, will be showcased in Section~\ref{sec:num_ex} through a comprehensive numerical example.
\end{remark}

\begin{remark}
	Among the chief benefits of the distributed control schemes proposed in Theorems~\ref{thm:IO_stab} and \ref{thm:implem} is that they provide stability guarantees in both the continuous- and discrete-time contexts. In contrast to this statement, we now refer to the proof of sufficiency for Theorem 2 from \cite{SLS} and to the proof of Lemma 4 from \cite{SLS}. The fact that $\frac{1}{z}I_p$ is a discrete-time stable TFM has been employed in these proofs so as to provide guarantees for both internal stability (recall Definition~\ref{def:int_stab}) and for the fact that the free evolution of the plant's and the controller's states tends to zero from any finite initial conditions (see the proof of Theorem~\ref{thm:implem} in Appendix~\ref{sec:ap_D}), in closed-loop configuration. Thus, the arguments made in \cite{SLS} cannot be employed for the continuous-time setting, since $\frac{1}{s}I_p$ is not a stable TFM in the latter time domain.
\end{remark}

Given that we now posses theoretical guarantees for closed-loop distributed stabilization via SRTR-based control laws, it becomes natural to investigate the means by which these representations can be constrained to possess various sparsity structures, meant to promote decentralized implementation. In the following subsection, we turn our attention to this problem, with the aim of finding reliable numerical procedures that enforce a desired sparsity pattern upon an SRTR pair.\vspace{-1mm}

\subsection{Forming structured and stabilizing SRTR-based implementations}\label{subsec:struc}

In light of Theorem~\ref{inversa} and of the results presented in Section III of \cite{aug_sparse}, it is tempting to envision a procedure that produces stable SRTR pairs by way of structured LCFs over $\mathbb{S}$. However, a closer inspection of Theorem~\ref{thm:implem} suggests a different approach. 

Recall that, as per Remark~\ref{rem:NRFlike}, we aim to implement state-space realizations of (block-)rows belonging to the TFM from \eqref{eq:SRTR_ctl}. Moreover, notice that the latter has the same sparsity structure as the TFM given in \eqref{VWdef} and that, for the purpose of computational efficiency when calculating the command signals, we wish to obtain realizations for these (block-)rows which have orders that are as small as possible.

It turns out that a convenient solution to these objectives can be obtained by employing the class of realizations parametrized by the constant matrices $K\in\mathbb{R}^{(n-p) \times p}$ and given explicitly in \eqref{VWdef}. Given two binary matrices, which encode the desired sparsity structure of the stable SRTR pair, and a set of integers, which impose the orders of the sought-after realizations, the following result offers tractable solutions.

\begin{theorem}\label{thm:MM}
	Consider a system given by a minimal realization \eqref{ss2a}-\eqref{ss2b} and let $e_{q,i}$ denote the $i^{th}$ vector from the canonical basis of $\mathbb{R}^q$, for any nonzero $q\in\mathbb{N}$. For each $i\in1:p$, denote by $Q_i\in\mathbb{R}^{(n-p)\times(n-p)}$ any orthogonal matrix which satisfies
	\begin{equation}\label{eq:compress}
		e_{p,i}^\top A_{12}^{}Q_i^\top=\left\|e_{p,i}^\top A_{12}^{}\right\|e_{(n-p),(n-p)}^\top\ ,
	\end{equation}
	along with an integer $n_i\in1:(n-p)$. Consider also $\mathcal{B}_{\mathbf{W}}^{}\in\mathbb{R}^{p\times p}$ and $\mathcal{B}_{\mathbf{V}}^{}\in\mathbb{R}^{p\times m}$ with $e_{p,i}^\top\mathcal{B}_{\mathbf{W}}^{}e_{p,j}^{}\in\{0,1\}$ and $e_{p,i}^\top\mathcal{B}_{\mathbf{V}}^{}e_{m,k}^{}\in\{0,1\}$, for any $i,j\in1:p$ and any $k\in1:m$. Assume that, for some $K\in\mathbb{R}^{(n-p) \times p}$ and for each $i\in1:p$, we have:\smallskip
	
	\begin{enumerate}
		\item[$i)$] $e_{p,i}^\top(A_{11}-A_{12}K)e_{p,j}^{}\left(1-e_{p,i}^\top\mathcal{B}_{\mathbf{W}}^{}e_{p,j}^{}\right)=0,\ \forall j\in1:p$;\smallskip
		
		\item[$ii)$] $e_{p,i}^\top B_1e_{m,k}^{}\left(1-e_{p,i}^\top\mathcal{B}_{\mathbf{V}}^{}e_{m,k}^{}\right)=0,\ \forall k\in1:m$;\smallskip			
	\end{enumerate}

	\noindent
	and, whenever $\left\|e_{p,i}^\top A_{12}^{}\right\|\neq 0$, we also have:\smallskip
	
	\begin{enumerate}
		\item[$iii)$] $\begin{bmatrix}O &\hspace{-2mm} I_{n_i}\end{bmatrix}\hspace{-1mm} Q_i( K A_{11} - KA_{12}K + A_{21} -A_{22}K)e_{p,j}^{}\left(1-e_{p,i}^\top\mathcal{B}_{\mathbf{W}}^{}e_{p,j}^{}\right)\hspace{-1mm}=O,\forall j\in1:p$;\smallskip	
		
		\item[$iv)$] $\begin{bmatrix}O &\hspace{-2mm} I_{n_i}\end{bmatrix}\hspace{-1mm} Q_i(KB_1+B_2)e_{m,k}^{}\left(1-e_{p,i}^\top\mathcal{B}_{\mathbf{V}}^{}e_{m,k}^{}\right)=O,\forall k\in1:m$;\smallskip
		
		\item[$v)$] $\begin{bmatrix}O &\hspace{-2mm} I_{n_i}\end{bmatrix}\hspace{-1mm} Q_i(A_{22}+KA_{12})Q_i^\top\begin{bmatrix}I_{n-p-n_i} & O\end{bmatrix}^\top=O$;\smallskip
		
		\item[$vi)$] $\left\{\lambda\in\mathbb{C}\ :\ \det\left(\lambda I_{n_i}-\begin{bmatrix}O &\hspace{-2mm} I_{n_i}\end{bmatrix} Q_i(A_{22}+KA_{12})Q_i^\top\begin{bmatrix}O &\hspace{-2mm} I_{n_i}\end{bmatrix}^\top\right)=0\right\}\subseteq\mathbb{S}$.\smallskip			
	\end{enumerate}

	\noindent
	Then, it follows that:\smallskip
	
	\begin{enumerate}
		\item[$\mathbf{a)}$] The TFM $e_{p,i}^\top\begin{bmatrix}
			\mathbf{W}(\lambda) & \mathbf{V}(\lambda)
		\end{bmatrix}$ can be expressed, when $\left\|e_{p,i}^\top A_{12}^{}\right\|= 0$, via the constant matrix
	\end{enumerate}
	\begin{equation}\label{eq:linec}\tag{4.9a}
	\hspace{-1mm}e_{p,i}^\top\begin{bmatrix}
	\mathbf{W}(\lambda) & \mathbf{V}(\lambda)
	\end{bmatrix}=e_{p,i}^\top\begin{bmatrix}
	A_{11}-A_{12}K& B_1
	\end{bmatrix};
	\end{equation}
	\begin{enumerate}
		\item[\phantom{$\mathbf{a)}$}] or,	when $\left\|e_{p,i}^\top A_{12}^{}\right\|\neq 0$, via the realization of order $n_i$ given by
	\end{enumerate}
		\begin{equation}\label{eq:linenc}\tag{4.9b}
			e_{p,i}^\top\begin{bmatrix}
				\mathbf W(\lambda) &\mathbf V(\lambda)
			\end{bmatrix}= \left[\begin{array}{c|c}
				\begin{bmatrix}O & I_{n_i}\end{bmatrix} Q_i(A_{22}+KA_{12})Q_i^\top\begin{bmatrix}O & I_{n_i}\end{bmatrix}^\top  -\lambda I_{n_i}  & \widetilde{B}_{i2}  \\ 
				\hline  e_{p,i}^\top A_{12}Q_i^\top\begin{bmatrix}O & I_{n_i}\end{bmatrix}\hspace{-1mm}\phantom{}^\top   &  \widetilde{D}_i  \\ \end{array}\right],
		\end{equation}
	\begin{enumerate}
		\item[\phantom{$\mathbf{a)}$}] having defined \stepcounter{equation}
	\end{enumerate}
		\begin{subequations}
			\begin{align}
				\widetilde{B}_{i2}:=&\begin{bmatrix}O & I_{n_i}\end{bmatrix} Q_i\begin{bmatrix}
					K A_{11} - KA_{12}K + A_{21} -A_{22}K   &  KB_1+B_2
				\end{bmatrix},\label{eq:aux_defa}\\
				\widetilde{D}_i:=&\ e_{p,i}^\top\begin{bmatrix}
					A_{11} - A_{12}K & B_1
				\end{bmatrix};\label{eq:aux_defb}
			\end{align}
		\end{subequations}
	\begin{enumerate}
		\item[$\mathbf{b)}$] The SRTR pair $(\mathbf{W}(\lambda),\mathbf{V}(\lambda))$ obtained from $K$ via the realization from \eqref{VWdef} is formed of stable TFMs that, for any $i,j\in1:p$ and any $k\in1:m$, satisfy
	\end{enumerate}
	\begin{subequations}
		\begin{align}
			e_{p,i}^\top\mathcal{B}_{\mathbf{W}}^{}e_{p,j}^{}=0\ \Longrightarrow &\ \  e^\top_{p,i}\mathbf{W}(\lambda)e_{p,j}^{}\equiv0,\label{eq:sufW}\\
			e_{p,i}^\top\mathcal{B}_{\mathbf{V}}^{}e_{m,k}^{}=0\ \Longrightarrow &\ \  e^\top_{p,i}\mathbf{V}(\lambda)e_{m,k}^{}\equiv0.\label{eq:sufV}
		\end{align}
	\end{subequations}
\end{theorem}
\begin{proof}
	See Appendix~\ref{sec:ap_D}.
\end{proof}

\begin{remark}
	Notice the fact that the proposed procedure also caters to the situations in which we desire to implement block-rows of the obtained SRTR pair. Indeed, the sparsity pattern and stability constraints are satisfied if and only if all rows of the considered block-rows abide by the aforementioned constraints. Additionally, the dimensionality aspect of the problem can also be tackled by minimizing the individual orders of the realizations belonging to a block-row's component rows. Should the block-row problem be tackled directly, the algebraic conditions from points $i)$ through $vi)$ of Theorem~\ref{thm:MM} can easily be altered to accommodate this change of perspective, with the resulting procedure amounting to a mere mutation of the provided technique.\vspace{-2mm}
\end{remark}

\begin{remark}\label{rem:SRTR_via_NRF}
	In contrast to the approach from \cite{Realiz}, our technique places little emphasis on the sparsity pattern belonging to the realization of type \eqref{ss2a}-\eqref{ss2b}, employed in the enforcement of sparsity patterns on the resulting SRTR pairs. Furthermore, as will be shown in the numerical example from Section~\ref{sec:num_ex}, the proposed technique can produce SRTR pairs with distinctive network topologies and high degrees of sparsity, even when the realization employed by the procedure possesses a greatly reduced degree of sparsity. In spite of this, the arguments made in \cite{Realiz} remain relevant even in the considered case, which is to say that the discussed problem's feasibility is highly dependent upon the network's intrinsic flow of information. Thus, the aforementioned example will employ the procedure from \cite{aug_sparse} to obtain a controller with just such an information flow, imparted by a structured NRF representation.\vspace{-2mm}
\end{remark}

Before moving on to the numerical example, we first discuss the computational aspects of the algorithms used to satisfy conditions $i)$ through $vi)$ from Theorem~\ref{thm:MM}. Firstly, the orthogonal matrices $Q_i$ with $i\in1:p$, which compress the individual rows of $A_{12}$, are straightforward to compute via the use of permutations and Householder transformations. Additionally, notice that conditions $i)$, $iv)$ and $v)$ only amount to a set of equality constraints in terms of linear combinations of the entries belonging to $K$, whereas condition $ii)$ must only be checked against the desired sparsity pattern.\vspace{-2mm}

\begin{remark}
	It has been shown via the identity from \eqref{eq:all_Gs} in Appendix~\ref{sec:ap_D} that $B_1$ is invariant with respect to the class of all minimal realizations of a system showcasing the structure from \eqref{ss2a}-\eqref{ss2b}. Therefore, the choice of minimal realization to be employed in Theorem~\ref{thm:MM} does not restrict the result's applicability and means of inducing sparsity in $B_1$ will be presented in the numerical example from Section~\ref{sec:num_ex}, by employing the synthesis procedures from \cite{aug_sparse}, as previously indicated in Remark~\ref{rem:SRTR_via_NRF}.\vspace{-2mm}
\end{remark}

We conclude that the main difficulty in applying Theorem~\ref{thm:MM} rests with satisfying conditions $iii)$ and $vi)$. Notably, the former involves a series of equality constraints containing bilinear terms of the entries which belong to the matrix $K$, whereas the latter requires that we restrict the eigenvalues of an affine expression of $K$ to $\mathbb{S}$. Fortunately, these types of problems have been jointly tackled in Section IV of \cite{aug_sparse} and Algorithm~1 therein can easily be adapted to accommodate the conditions laid out in Theorem~\ref{thm:MM}. While the procedures from \cite{aug_sparse} deal with continuous-time versions of condition $vi)$ from the aforementioned theorem, we point out that the discrete-time case can be handled just as elegantly, by employing Lemma 6 from \cite{Realiz} alongside the convexification techniques given in \cite{aug_sparse}. Alternatively, one may employ Lemma 8 from \cite{Realiz} to restrict the eigenvalues of $A_{22}+KA_{12}$ to $\mathbb{S}$ which, combined with condition $v)$ of Theorem~\ref{thm:MM}, is sufficient to satisfy condition $vi)$ of the aforementioned theorem.\vspace{-2mm}

\begin{remark}
	The high computational performance of the indicated algorithm, as highlighted in Table 1 of \cite{aug_sparse} and discussed comparatively in Section V.B of \cite{aug_sparse}, serves to justify the chosen approach and its distinctive formulation, given in Theorem~\ref{thm:MM}. Furthermore, a code implementation for this type of numerical procedure is publicly available, and can be accessed via the link in the footnote from the first page of \cite{aug_sparse}.
\end{remark}

\section{Numerical example}\label{sec:num_ex}

We now present an example of SRTR-based design for distributed control. To this end, we return to the network illustrated in Fig.~\ref{fig:three_hop} and modeled in Section~\ref{subsec:mdl}. Recall that, although this network possesses a highly particular topology, the latter's structure is completely obscured when inspecting the network via its TFM. The aim of this section is to illustrate how the framework developed in this paper can be used to obtain a controller, for the aforementioned network, which has the same ring topology and which can be implemented in a distributed manner, via a set of sparse and stabilizing control laws. The chosen example, inspired by the one showcased in Section V of \cite{aug_sparse}, will be presented in the continuous-time setting. We point out that, in doing so, our aim is to provide further contrast with respect to the framework discussed in \cite{SLS}, which caters exclusively to discrete-time systems. 

Consider a network of the type described in Section~\ref{subsec:mdl} which, instead of the $3$ nodes shown in Fig.~\ref{fig:three_hop}, is more generically formed from the ring-like interconnection of $p\in\mathbb{N}$ nodes, with $p\geq 2$. Thus, its NRF-like dynamics are captured via the identities
\begin{equation}\label{eq:NRF_ex}
	\mathbf{Y}_{(i\text{ mod }p)+1}(s) = \mathbf{\Phi}_\mathbf{G}(s)\mathbf{Y}_{((i-1)\text{ mod }p)+1}(s) + \mathbf{\Gamma}_\mathbf{G}(s)\mathbf{U}_{(i\text{ mod }p)+1}(s),\ \forall i\in 1:p,
\end{equation}
where $\mathbf{Y}_i(s)$ and $\mathbf{U}_i(s)$ denote the Laplace transforms of the network's scalar outputs and inputs, respectively, while the interconnection and input dynamics are given by the scalar TFMs $\mathbf{\Phi}_\mathbf{G}(s)$ and $\mathbf{\Gamma}_\mathbf{G}(s)$, respectively. Notice that these two scalar TFMs do not depend upon $i\in 1:p$, which indicates homogeneous dynamics among the network's nodes (see \cite{plutonizare} for a practical scenario where such a situation can occur).

The aim is to design a set of control laws which exhibit the same type of ring-like structure, \emph{i.e.}, each command signal may be computed using only information that is local or which originates from the node directly behind it in the ring. Moreover, given the homogeneous dynamics occurring in each of the controlled network's nodes, we wish to preserve this trait with respect to the controller's SRTR implementation. Therefore, we desire control laws of type \eqref{eq:SRTR_conv} which, for all $i\in 1:p$, are given by
\begin{align}\nonumber
	\mathbf{U}_{(i\text{ mod }p)+1}(s) =
	\tfrac{1}{s}\begin{bmatrix}
		\mathbf{W}_{local}(s) & \mathbf{V}_{local}(s)
	\end{bmatrix}\begin{bmatrix}
	\mathbf{U}_{(i\text{ mod }p)+1}^\top(s) & \mathbf{Y}_{(i\text{ mod }p)+1}^\top(s)
\end{bmatrix}^\top+\\+
	\tfrac{1}{s}\begin{bmatrix}
		\mathbf{W}_{prev}(s) & \mathbf{V}_{prev}(s)
	\end{bmatrix}\begin{bmatrix}
		\mathbf{U}_{((i-1)\text{ mod }p)+1}^\top(s) & \mathbf{Y}_{((i-1)\text{ mod }p)+1}^\top(s)
	\end{bmatrix}^\top.\label{eq:SRTR_ex}
\end{align}
Once more, note that the four stable and scalar TFMs, $\mathbf{W}_{local}(s)$, $\mathbf{V}_{local}(s)$, $\mathbf{W}_{prev}(s)$ and $\mathbf{V}_{prev}(s)$, which form the sought-after control law, do not depend upon $i\in 1:p$.

In this example, we consider that the network is characterized by $p=6$ and by
\begin{equation}\label{eq:net_ex}
	\mathbf{\Phi}_\mathbf{G}(s) = \left[\begin{array}{r|r}
		-0.5 - s & -0.1 \\\hline 1 & 0
	\end{array}\right] = \frac{-1}{10s+5}\ ,\quad \mathbf{\Gamma}_\mathbf{G}(s) = \left[\begin{array}{r|r}
		1 - s & 1 \\\hline 1 & 0
	\end{array}\right]= \frac{1}{s-1}\ .
\end{equation}
By employing the identities from \eqref{eq:NRF_ex} along with the realizations from \eqref{eq:net_ex}, and by defining the matrix which models the interconnection structure of the $\mathbf{\Phi}_\mathbf{G}(s)$ terms
\begin{equation}\label{eq:F_def}
	F_b:={\scriptsize\begin{bmatrix}
			0 & 0 & 0 & 0 & 0 & 1\\
			1 & 0 & 0 & 0 & 0 & 0\\
			0 & 1 & 0 & 0 & 0 & 0\\
			0 & 0 & 1 & 0 & 0 & 0\\
			0 & 0 & 0 & 1 & 0 & 0\\
			0 & 0 & 0 & 0 & 1 & 0
	\end{bmatrix}},
\end{equation}
we will express a state-space realization for our network as follows
\begin{equation}\label{eq:net_real}
	\mathbf{G}(s)=\left[\begin{array}{cc|c}
		(-0.5-s) I_6 - 0.1 F_b & -0.1 F_b & O \\
		O & (1-s)I_6 & I_6\\\hline
		I_6 & I_6 &O
	\end{array}\right].
\end{equation}

The state vector of the realization from \eqref{eq:net_real} is, as per the identities in \eqref{eq:net_ex}, the concatenation of the state vectors belonging to all $\mathbf{\Phi}_\mathbf{G}(s)$ subsystems and of those belonging to all $\mathbf{\Gamma}_\mathbf{G}(s)$ subsystems making up the interconnection. Moreover, routine computations show that the network's realization from \eqref{eq:net_real} is minimal, which also makes it both stabilizable and detectable. Thus, Theorem~\ref{thm:implem} can be freely applied.

\begin{remark}
	The observation pertaining to the state vector of the realization from \eqref{eq:net_real} is crucial. By devising distributed control laws which stabilize the state dynamics of the system given in \eqref{eq:net_real} in closed-loop configuration, this controller's distributed implementation will guarantee that the states of all the $\mathbf{\Phi}_\mathbf{G}(s)$ and $\mathbf{\Gamma}_\mathbf{G}(s)$ subsystems forming the network will remain bounded when they are excited by bounded exogenous signals (as a consequence of the closed-loop state dynamics being made asymptotically stable, see Section 5.3 in \cite{zhou} and Corollary III.9 from \cite{NRF}, along with its proof), and that these states (along with the state variables of the controller's implementation) will tend to zero when evolving freely from any finite initial conditions.
\end{remark}

In order to produce the control laws described by the identities from \eqref{eq:SRTR_ex}, we first employ the numerical procedures from \cite{aug_sparse} to obtain a controller whose NRF representation has the desired type of sparsity structure. This controller is given by a minimal state-space realization of order $n=12$ and type \eqref{ss2a}-\eqref{ss2b} denoted by
\begin{equation}\label{eq:ctl_real}
	\mathbf{K}(s)=\left[\begin{array}{cc|c} A_{11}   -\lambda I_6& A_{12}  &  B_1 \\
		A_{21}  & A_{22} -\lambda I_{6}  & B_2 \\ 
		\hline I_6 & O   & O  \end{array}\right],
\end{equation}
where we have that
\begin{align*}
	A_{11} =& \left[\scriptsize\begin{array}{rrrrrr}
		-11.8035 &  -2.5928 &   0.6483 &  -0.1498 &  -0.0059 &   1.5415\\
		  1.5415 & -11.8035 &  -2.5928 &   0.6483 &  -0.1498 &  -0.0059\\
		 -0.0059 &   1.5415 & -11.8035 &  -2.5928 &   0.6483 &  -0.1498\\
		 -0.1498 &  -0.0059 &   1.5415 & -11.8035 &  -2.5928 &   0.6483\\
		  0.6483 &  -0.1498 &  -0.0059 &   1.5415 & -11.8035 &  -2.5928\\
		 -2.5928 &   0.6483 &  -0.1498 &  -0.0059 &   1.5415 & -11.8035
	\end{array}\right],\\
	A_{12} =& \left[\scriptsize\begin{array}{rrrrrr}
		 -7.3460 & -12.0158 &  29.1222 & -12.1621 &   9.7523 & \hspace{1.5mm}-5.1182\\
		-16.0509 &  10.1202 &  25.3239 &  -8.1033 & -15.2385 &                 4.0548\\
		-23.7679 &  21.4320 &   8.6752 &  13.6752 &   1.5280 &                -5.1783\\
		-27.0198 &   9.8895 & -11.8960 &  -7.4120 &  16.8991 &                 2.4006\\
		-26.0184 & -17.3463 & -14.7520 &  -7.3268 &  -4.8881 &                -6.5282\\
		-15.8018 & -29.0018 &   9.5454 &  10.1566 &   5.3141 &                 1.5347
	\end{array}\right],\\
	A_{21} =& \left[\scriptsize\begin{array}{rrrrrr}
		 0.1795  &  0.0818  &  0.4039  &  0.4133  &  0.4137  &  \phantom{-}0.3992\\
		 0.4333  & -0.0452  & -0.2691  & -0.3175  &  0.0051  &             0.4693\\
		-0.4047  & -0.5266  & -0.2458  & -0.0123  &  0.2775  &             0.1615\\
		-0.4200  &  0.5344  &  0.0288  & -0.4924  &  0.4689  &             0.0625\\
		-0.2869  & -0.2283  &  0.5717  & -0.3033  & -0.4110  &             0.4398\\
		-0.2019  &  0.2441  & -0.2441  &  0.2778  & -0.2152  &             0.2834
	\end{array}\right],\\
	A_{22} =& \left[\scriptsize\begin{array}{rrrrrr}
		-2.2413  &  1.0141  &  0.3885  & -0.5740  & -0.4535  &  0.2336\\
		-1.0913  & -2.1134  & -2.7942  & -0.1511  &  0.1878  &  0.0216\\
		-0.3266  &  2.7957  & -2.1347  & -0.4690  &  0.0492  &  0.1115\\
		 0.6269  &  0.2681  &  0.3424  & -2.8351  &  4.0507  & -0.0716\\
		 0.2603  & -0.3236  & -0.2731  & -4.0526  & -2.8072  &  0.0350\\
		 0.1026  &  0.1898  & -0.0850  &  0.0140  & -0.1408  & -4.6168
	\end{array}\right],\\
	B_{1} =& \left[\scriptsize\begin{array}{rrrrrr}
		-1.0779  &  0.0000  &  0.0000  &  0.0000  &  0.0000  & 15.8441\\
		15.8441  & -1.0779  &  0.0000  &  0.0000  &  0.0000  &  0.0000\\
		 0.0000  & 15.8441  & -1.0779  &  0.0000  &  0.0000  &  0.0000\\
		 0.0000  &  0.0000  & 15.8441  & -1.0779  &  0.0000  &  0.0000\\
		 0.0000  &  0.0000  &  0.0000  & 15.8441  & -1.0779  &  0.0000\\
		 0.0000  &  0.0000  &  0.0000  &  0.0000  & 15.8441  & -1.0779
	\end{array}\right],\\
	B_{2} =& \left[\scriptsize\begin{array}{rrrrrr}
		 0.6417  &  1.1275  &  1.4058  &  1.4121  &  1.0183  &  0.5263\\
		-0.1623  & -1.0111  & -0.7933  &  0.4321  &  1.4028  &  1.0261\\
		-1.5048  & -0.7775  &  0.3057  &  0.7985  & -0.0171  & -1.2371\\
		 0.5132  & -0.2538  & -0.0823  &  0.3877  & -0.1357  &  0.1614\\
		 0.2803  &  0.2855  & -0.5768  & -0.2496  & -0.0568  & -0.3892\\
		 0.0177  &  0.0315  &  0.0806  &  0.1186  &  0.1352  &  0.0834
	\end{array}\right].
\end{align*}

\begin{remark}
	Notably, the obtained controller's TFM displays no distinctive sparsity pattern, having only nonzero entries. This is straightforward to check by simply computing the value of $\mathbf{K}(0)=\begin{bmatrix}
		-I_6 & O \end{bmatrix}\begin{bmatrix} A_{11} & A_{12} \\ A_{21} & A_{22} \end{bmatrix}^{-1}\begin{bmatrix} B_1^\top & B_2^\top \end{bmatrix}^\top$.
	As discussed throughout Section~\ref{adoua}, the classical TFM-based representation of a system is oftentimes opaque to its intrinsic structure, as otherwise captured by NRF or SRTR pairs.
\end{remark}

For the obtained realization of the controller, given in \eqref{eq:ctl_real}, we will employ a set of orthogonal matrices in order to compress each individual row of $A_{12}$. In the notation employed in Theorem~\ref{thm:MM}, these matrices are given explicitly by
\begin{align*}
	Q_{1} =& \left[\scriptsize\begin{array}{rrrrrr}
		-0.1411 &  -0.0389 &   0.0942 &  -0.0393 &   0.0315 &   0.9834\\
		0.2689   & 0.0741   &-0.1795   & 0.0750   & 0.9399   & 0.0315\\
		-0.3353  & -0.0924  &  0.2239  &  0.9065  &  0.0750  & -0.0393\\
		0.8029   & 0.2212   & 0.4639   & 0.2239   &-0.1795   & 0.0942\\
		-0.3313  &  0.9087  &  0.2212  & -0.0924  &  0.0741  & -0.0389\\
		-0.2025  & -0.3313  &  0.8029  & -0.3353  &  0.2689  & -0.1411
	\end{array}\right],\\
	Q_{2} =& \left[\scriptsize\begin{array}{rrrrrr}
		0.1118   &-0.0216   &-0.0541   & 0.0173   & 0.0326   & 0.9913\\
		-0.4201  &  0.0813  &  0.2033  & -0.0651  &  0.8776  &  0.0326\\
		-0.2234  &  0.0432  &  0.1081  &  0.9654  & -0.0651  &  0.0173\\
		0.6982   &-0.1350   & 0.6621   & 0.1081   & 0.2033   &-0.0541\\
		0.2790   & 0.9460   &-0.1350   & 0.0432   & 0.0813   &-0.0216\\
		-0.4425  &  0.2790  &  0.6982  & -0.2234  & -0.4201  &  0.1118
	\end{array}\right],\\
	Q_{3} =& \left[\scriptsize\begin{array}{rrrrrr}
		-0.1428  &  0.0510  &  0.0206  &  0.0325  &  0.0036  &  0.9877\\
		0.0421   &-0.0150   &-0.0061   &-0.0096   & 0.9989   & 0.0036\\
		0.3770   &-0.1346   &-0.0545   & 0.9141   &-0.0096   & 0.0325\\
		0.2392   &-0.0854   & 0.9654   &-0.0545   &-0.0061   & 0.0206\\
		0.5909   & 0.7891   &-0.0854   &-0.1346   &-0.0150   & 0.0510\\
		-0.6553  &  0.5909  &  0.2392  &  0.3770  &  0.0421  & -0.1428
	\end{array}\right],\\
	Q_{4} =& \left[\scriptsize\begin{array}{rrrrrr}
		0.0662&   -0.0103&    0.0124 &   0.0078&   -0.0177&    0.9975\\
		0.4659   &-0.0728   & 0.0876    &0.0546   & 0.8756   &-0.0177\\
		-0.2043  &  0.0319  & -0.0384   & 0.9761  &  0.0546  &  0.0078\\
		-0.3280  &  0.0512  &  0.9384   &-0.0384  &  0.0876  &  0.0124\\
		0.2726   & 0.9574   & 0.0512    &0.0319   &-0.0728   &-0.0103\\
		-0.7449  &  0.2726  & -0.3280   &-0.2043  &  0.4659  &  0.0662
	\end{array}\right],\\
	Q_{5} =& \left[\scriptsize\begin{array}{rrrrrr}
		-0.1800  & -0.0501  & -0.0426  & -0.0212  & -0.0141  &  0.9811\\
		-0.1348  & -0.0375  & -0.0319  & -0.0159  &  0.9894  & -0.0141\\
		-0.2020  & -0.0563  & -0.0478  &  0.9762  & -0.0159  & -0.0212\\
		-0.4067  & -0.1133  &  0.9037  & -0.0478  & -0.0319  & -0.0426\\
		-0.4782  &  0.8668  & -0.1133  & -0.0563  & -0.0375  & -0.0501\\
		-0.7173  & -0.4782  & -0.4067  & -0.2020  & -0.1348  & -0.1800
	\end{array}\right],\\
	Q_{6} =& \left[\scriptsize\begin{array}{rrrrrr}
		0.0423   & 0.0236  & -0.0078  & -0.0083  & -0.0043  &  0.9988\\
		0.1465   & 0.0816  & -0.0269  & -0.0286  &  0.9850  & -0.0043\\
		0.2800   & 0.1559  & -0.0513  &  0.9454  & -0.0286  & -0.0083\\
		0.2632   & 0.1466  &  0.9518  & -0.0513  & -0.0269  & -0.0078\\
		-0.7996  &  0.5547 &   0.1466 &   0.1559 &   0.0816 &   0.0236\\
		-0.4356  & -0.7996 &   0.2632 &   0.2800 &   0.1465 &   0.0423		
	\end{array}\right].
\end{align*}

With all of these elements now computed, we proceed to employ the procedure given in Section~\ref{subsec:struc}, with the aim of obtaining a structured and stabilizing SRTR representation for our controller. To this end, we will employ Theorem~\ref{thm:MM} and we point out the fact that $\left\|e_{p,i}^\top A_{12}^{}\right\|\neq 0,\ \forall i\in1:6$. We then inspect the structure of the control law given in \eqref{eq:SRTR_ex} and we select the two binary matrices from the statement of Theorem~\ref{thm:MM} as follows $\mathcal{B}_{\mathbf{W}}=\mathcal{B}_{\mathbf{V}}=I_6+F_b$, recalling the matrix defined in \eqref{eq:F_def}. 

Additionally, it is desirable to find low-order state-space realizations for each row of the TFM given in \eqref{eq:SRTR_ctl} and, therefore, we choose $n_i=1,\ \forall i\in1:6$. Finally, recall that we wish to obtain an SRTR-based control law in which each command signal is computed via the expression given in \eqref{eq:SRTR_ex}. For the controller produced by the procedures from \cite{aug_sparse}, we obtain precisely these types of control laws when, in addition to the constraints stated in Theorem~\ref{thm:MM}, we also impose that $A_{22}+KA_{12}$ be diagonal and that it have the same element on each position of its diagonal. 

\begin{remark}\label{rem:modif}
	Notice that if the latter two constraints from the previous paragraph which pertain to $A_{22}+KA_{12}$ are satisfied, then condition $v)$ from Theorem~\ref{thm:MM} is implicitly fulfilled. Therefore, the constraints arising from condition $v)$ can be omitted from the optimization procedure's statement, for the sake of avoiding redundancy.
\end{remark}

A solution to the modified problem described in Remark~\ref{rem:modif} is given by the matrix
\begin{equation*}
	K=\left[\scriptsize\begin{array}{rrrrrr}
		-0.0259  &  0.0404  &  0.0610  &  0.0693  &  0.0776  &  0.0821\\
		 0.0563  & -0.0062  & -0.0746  & -0.0792  &  0.0163  &  0.1318\\
		-0.1053  & -0.0838  & -0.0574  &  0.0127  &  0.1007  &  0.0124\\
		 0.0257  &  0.1610  & -0.1578  & -0.0006  &  0.1571  & -0.1560\\
		-0.1736  &  0.1068  &  0.0956  & -0.1956  &  0.0737  &  0.0580\\
		 0.1935  & -0.1927  &  0.1912  & -0.1835  &  0.1914  & -0.1767
	\end{array}\right],
\end{equation*}
which produces, via \eqref{VWdef}, the SRTR-based control laws of type \eqref{eq:SRTR_ex} given by
\begin{equation}\label{eq:SRTR_obt}
	\begin{array}{ll}
		\mathbf{W}_{local}(s) = \dfrac{-5.255s - 55.9}{s+9.34}\ , & \mathbf{V}_{local}(s) = \dfrac{-1.078 s - 94.28}{s+9.34}\ ,\\
		\phantom{ }\\
		\mathbf{W}_{prev}(s)  \hspace{0.25mm}= \dfrac{-15.84}{s+9.34}\ , & \mathbf{V}_{prev}(s)  \hspace{0.25mm}= \dfrac{15.84 s - 15.84}{s+9.34}\ .\\
	\end{array}
\end{equation}

The obtained SRTR pair of the controller from \eqref{eq:ctl_real} is (continuous-time) stable, as can be seen from \eqref{eq:SRTR_obt}. Additionally, since this SRTR pair belongs to a stabilizing controller of the network, whose realization from \eqref{eq:net_real} is both stabilizable and detectable, the obtained control laws of type \eqref{eq:SRTR_conv} can be implemented as in Theorem~\ref{thm:implem} with the aim of stabilizing the network's state dynamics, as given in \eqref{eq:net_real}.

\section{Conclusions} 
In this paper, we have introduced the SRTR representation for networked systems and we have presented a comprehensive discussion on the intrinsic connections between SRTR pairs, their NRF counterparts introduced in \cite{NRF}, the response TFMs from \cite{SLS} with their associated implementation matrices in \cite{Separ}, and the consecrated coprime factorizations treated extensively in \cite{V}. Akin to our previous results from \cite{NRF}, we have provided a unitary mathematical framework that is not restricted to either continuous- (see \cite{Explicit}) or discrete-time (see \cite{SLS}) network models, while also highlighting the proposed paradigm's rich potential for system analysis. Finally, we have extended the closed-loop internal stability guarantees from \cite{NRF} to the current setting, thus paving the way for distributed SRTR-based implementations of control laws and validating our proposed framework's capacity for controller design.\bigskip

\appendix

\section{Rosenbrock-type realizations}\label{sec:ap_A}\medskip

We state here several results related to a type of realization which has been fully investigated in \cite{Rose70}. These remarkable representations are the generalizations of state-space \cite{Wonham}, descriptor \cite{Verg81} and pencil \cite{CF,SIMAX} realizations, being instrumental in constructing many of the proofs given in Appendix~\ref{sec:ap_D}.

\begin{definition} \label{rose_real}
	A collection of $4$ \emph{polynomial} matrices $(T(\lambda),U(\lambda),V(\lambda),W(\lambda))$ with coefficients in $\mathbb{R}$ and with $\det \left(T(\lambda)\right)\not\equiv 0$ that satisfy
	\begin{equation}\label{eq:rose_def}
		\mathbf{G}(\lambda)=V(\lambda)T^{-1}(\lambda)U(\lambda)+W(\lambda)=:\left[\footnotesize\begin{array}{r|r}
			-T(\lambda)&U(\lambda)\\\hline V(\lambda)&W(\lambda)
		\end{array}\right]
	\end{equation}
	is called a \emph{Rosenbrock-type realization} of $\mathbf{G}(\lambda)\in\mathbb{R}(\lambda)^{p\times m}$.
\end{definition}

\begin{definition}
	For a Rosenbrock-type realization of $\mathbf{G}(\lambda)$ as in \eqref{eq:rose_def}, the polynomial matrix denoted as
	\begin{equation}\label{eq:sysmat}
		\mathcal{S}_{\mathbf{G}}(\lambda):=\left[\small\begin{array}{rr}
			-T(\lambda)&U(\lambda)\\ V(\lambda)&W(\lambda)
		\end{array}\right]
	\end{equation}
	is called a \emph{system matrix} of $\mathbf{G}(\lambda)$.
\end{definition}

\begin{proposition}\label{prop:fnr}
	$\mathbf{G}(\lambda)$ has full row/column normal rank if and only if $\mathcal{S}_{\mathbf{G}}(\lambda)$ has full row/column normal rank.
\end{proposition}

The use of Rosenbrock-type realizations in Appendix~\ref{sec:ap_D} becomes justified due to the fact that improper TFMs, such as the $\lambda I_p-\mathbf{W}(\lambda)$ factor generated by an SRTR pair, cannot be modeled by using state-space realizations. Furthermore, we are particularly interested in a certain class of Rosenbrock-type realizations.

\begin{definition}\label{def:ired}
	A representation of type \eqref{eq:rose_def} for which $\begin{bmatrix}
		-T(\lambda)&U(\lambda)
	\end{bmatrix}$ and $\begin{bmatrix}
		-T^\top(\lambda)&V^\top(\lambda)
	\end{bmatrix}^\top$ have full row and, respectively, column rank $\forall\lambda\in\mathbb{C}$ is called an \emph{irreducible} Rosenbrock-type realization.
\end{definition}

The next result plays a key role in proving Theorem~\ref{CupruMin}.

\begin{theorem}{\cite[pp. 111, Theorem 4.1]{Rose70}}\label{thm:iso}
	Let $\mathbf{G}(\lambda)$ be given by an irreducible realization of type \eqref{eq:rose_def}. Then:
	\begin{enumerate}
		\item[$\mathbf{a)}$] The finite poles of $\mathbf{G}(\lambda)$ are the finite zeros of $T(\lambda)$;
		
		\item[$\mathbf{b)}$] The finite zeros of $\mathbf{G}(\lambda)$ are the finite zeros of $\mathcal{S}_{\mathbf{G}}(\lambda)$.
	\end{enumerate}
\end{theorem}

\bigskip

\section{Non-symmetrical Algebraic Riccati Equations}\label{sec:ap_B}\medskip

We show here that the right stabilizing solutions of \eqref{Riccati}, as introduced in Definition~\ref{def:stab}, can be readily expressed and reliably computed by employing a particular invariant subspace (see Section 1.3 of \cite{genricc}) of the matrix
\begin{equation}\label{eq:A_plus}
	A^+ :=\begin{bmatrix} \phantom{-(}A_{11} +F_1\phantom{)}  & - A_{12} \\  -(A_{21} +F_2) &  \phantom{-}A_{22}  \end{bmatrix}.
\end{equation}
Let $\begin{bmatrix}V_1^\top &  V_2^\top \end{bmatrix}^\top$ have $w$ columns and full column rank, for some integer $w\in1:n$, with $V_1$ having $p$ rows and $V_2$ having $n-p$ rows. By point 2 of Theorem 1.3.2 in \cite{genricc}, we have that if
\begin{equation}\label{eq:invar}
	\begin{bmatrix} \phantom{-(}A_{11} +F_1\phantom{)}  & - A_{12} \\  -(A_{21} +F_2) &  \phantom{-}A_{22}  \end{bmatrix}\begin{bmatrix}V_1 \\  V_2 \end{bmatrix}=\begin{bmatrix}V_1 \\  V_2 \end{bmatrix} S
\end{equation}
for some $w\times w$ matrix $S$, then $\begin{bmatrix}V_1^\top &  V_2^\top \end{bmatrix}^\top$ is a basis that spans an invariant subspace of dimension $w$ belonging to $A^+$. By the same result, we also have that the collection (which allows for repeated elements) of $w$ eigenvalues belonging to $S$ is included in the collection (also allowing for repeated elements) of $n$ eigenvalues belonging to $A^+$.

\begin{remark}
	An invariant subspace of $A^+$ will be termed \emph{disconjugate} if it admits a basis $\begin{bmatrix}V_1^\top & V_2^\top \end{bmatrix}^\top$ with $V_1$ having full column rank. Moreover, the disconjugacy of an invariant subspace is \emph{independent} with respect to the choice of basis for it.
\end{remark}

\begin{definition}
	When $S$ has all its eigenvalues located in $\mathbb{S}$, the subspace being spanned by the basis from $\eqref{eq:invar}$ is called \emph{$\mathbb{S}$-invariant}.
\end{definition}

Using these notions, we now express the following result, which characterizes the right stabilizing solutions of CTNAREs.

\begin{proposition}[Corollary 4.2 in \cite{EJC}] \label{Ric}
	The CTNARE from \eqref{Riccati} has a right stabilizing solution $K$ if and only if $A^+$ has a disconjugate $\mathbb{S}$-invariant subspace of dimension $p$. In such cases, if we denote any basis for such a subspace as $\begin{bmatrix}V_1^\top &  V_2^\top \end{bmatrix}^\top$, then $V_1$ is invertible and $K=V_2V_1^{-1}$ is a right stabilizing solution of \eqref{Riccati}.
\end{proposition}

\begin{remark} \label{Cristian} Since we consider only LCFs over $\mathbb{S}$ with representations of type \eqref{Kontroller}, and since $A^+$ and the state matrix of the state-space realization from \eqref{Kontroller} are related by a similarity transformation given by $\footnotesize\begin{bmatrix}
		I_p&O\\O&-I_{n-p}
	\end{bmatrix}$, it follows that $A^+$ has all its eigenvalues located in $\mathbb{S}$. Therefore, all its invariant subspaces, such as those spanned by the eigenvectors of $A^+$, are actually $\mathbb{S}$-invariant. By freely placing the eigenvalues of $A^+$ in distinct points of $\mathbb{S}\cap \mathbb{R}$ (due to Assumption~\ref{Observability}), the $n\times n$ matrix $T_{eig}$ whose columns are chosen from the resulting real eigenvectors of $A^+$ (selecting one eigenvector for each distinct eigenvalue) will be invertible. Thus, $\begin{bmatrix}
		I_p&O
	\end{bmatrix}T_{eig}$ has full row rank or, equivalently, at least one $p\times p$ nonzero minor. Due to this minor existing, we may always choose those $p$ columns of $T_{eig}$ that correspond to (one of) the nonzero $p\times p$ minor(s) of $\begin{bmatrix}
		I_p&O
	\end{bmatrix}T_{eig}$ to obtain a basis for a disconjugate $\mathbb{S}$-invariant subspace of dimension $p$ belonging to $A^+$. Therefore, Assumption~\ref{Observability} ensures that the CTNARE from \eqref{Riccati} can always be made to have a (real) right stabilizing solution.\bigskip
	
\end{remark}

\section{Auxiliary Results}\label{sec:ap_C}\medskip

We discuss here a number of technical results, which will be employed in proving the main theorems presented in this paper. The first of these results gives necessary and sufficient conditions, under the assumptions made in Sections~\ref{adoua} and~\ref{sec:main}, for both of the TFMs that make up an SRTR pair to be stable.

\begin{lemma}\label{lem:SRTR_stab}
	The TFM $\begin{bmatrix}
		\mathbf{W}(\lambda)&\mathbf{V}(\lambda)
	\end{bmatrix}$, as given in \eqref{VWdef}, is stable if and only if $A_{22}+KA_{12}$ has all its eigenvalues located in $\mathbb{S}$.
\end{lemma}
\begin{proof}
	The sufficiency of $A_{22}+KA_{12}$ having all its eigenvalues located in $\mathbb{S}$ is clear, since this implies that the TFM from \eqref{VWdef} cannot have poles located outside of $\mathbb{S}$. For the necessity of this condition, we begin by showing that \eqref{VWdef} is a minimal realization. Notice first that the pair $(A_{12},A_{22}+KA_{12})$ is observable $\forall\ K\in\mathbb{R}^{(n-p)\times p}$ if and only if the pair $(A_{12},A_{22})$ is observable. The latter is ensured via Assumption~\ref{Observability} and via Remark~\ref{Leuenberger}. Moving on, define now the following pencil
	\begin{multline}\label{eq:ctrb}
		\left[\begin{array}{c}
			S_1(\lambda)\\\hdashline S_2(\lambda)
		\end{array}\right]:=\left[\footnotesize\begin{array}{ccc}
			A_{11}-A_{12}K-\lambda I_p & A_{12} &  B_1\\\hdashline
			K A_{11} - KA_{12}K + A_{21} -A_{22}K   \hspace{-1mm}& A_{22}+ K A_{12}  -\lambda I_{n-p}  \hspace{-1mm}& KB_1+B_2
		\end{array}\right]=\\=\begin{bmatrix}I_p&O\\K&I_{n-p}\end{bmatrix}\begin{bmatrix}
			A-\lambda I_n & B
		\end{bmatrix}\begin{bmatrix}
			I_p&O&O\\-K&I_{n-p}&O\\O&O&I_m
		\end{bmatrix},\normalsize
	\end{multline}
	and notice that, due to Assumption~\ref{Controllability}, $\begin{bmatrix}
		S_1^\top(\lambda) &\hspace{-2mm} S_2^\top(\lambda)
	\end{bmatrix}^\top$ has full row rank $\forall\lambda\in\mathbb{C}$ and $\forall K\in\mathbb{R}^{(n-p)\times p}$. Then, $S_2(\lambda)$ must share this property, which makes the realization from \eqref{VWdef} controllable. Since this realization is also observable, we get that it is minimal and that all the eigenvalues of $A_{22}+KA_{12}$ are the poles of its TFM. To be stable, these eigenvalues must be located in $\mathbb{S}$, thus proving the statement's necessity.
\end{proof}

The following result is a classical state-space derivation, originating from \cite{Nett} and generalized in \cite{Lucic} (see also Proposition 1 in \cite{Lucic}). It proves crucial in establishing the connections between stable SRTR pairs and LCFs over $\mathbb{S}$, in Sections~\ref{subsec:s2c} and~\ref{subsec:c2s}.

\begin{theorem}
	\label{SIMAX99}
	Let $\mathbf G(\lambda)\in\mathbb{R}^{p \times m}(\lambda)$ be a strictly proper TFM. All LCFs over $\mathbb{S}$ of $\mathbf G(\lambda) = \mathbf{M}^{-1}(\lambda )\mathbf{N}(\lambda)$, with $\begin{bmatrix} \mathbf M(\lambda) &  \mathbf N(\lambda) \end{bmatrix}$ having the same McMillan degree as $\mathbf{G}(\lambda)$, are given by
	\begin{equation}\label{dc2}
		\begin{bmatrix} \mathbf M(\lambda) &  \mathbf N(\lambda) \end{bmatrix}\normalsize
		= \left[\begin{array}{c|cc} A + FC - \lambda I_n & F & B \\ \hline
			UC & U & O \end{array}\right]\normalsize,
	\end{equation}
	where $A, B, C, F$ and $U$ are real matrices such that:
	
	\begin{enumerate}
		
		\item[$\mathbf{a)}$] $U$ is an invertible matrix in $\mathbb{R}^{p\times p}$;
		
		\item[$\mathbf{b)}$] $F$ is a matrix which ensures that $A+FC$ has all its eigenvalues located in $\mathbb{S}$;
		
		\item[$\mathbf{c)}$] $\mathbf G(\lambda) =  \left[\begin{array}{c|c}A - \lambda I_n & B \\ \hline C & O\end{array}\right]$ is both controllable and observable.
	\end{enumerate}

\end{theorem}

A direct consequence of the previous result is the following lemma, which will be employed in proving the main result of Section~\ref{subsec:s2c}.

\begin{lemma}\label{lem:indeed_coprime}
	Let $\mathbf{G}(\lambda)$ be given by a minimal realization of type \eqref{ss2a}-\eqref{ss2b} and let $A_x\in\mathbb{R}^{p\times p}$ have all its eigenvalues located in $\mathbb{S}$. Let both $B_x\in\mathbb{R}^{p\times p}$ and $C_x\in\mathbb{R}^{p\times p}$ be invertible and let $K\in\mathbb{R}^{(n-p)\times p}$ ensure that $A_{22}+KA_{12}$ has all its eigenvalues located in $\mathbb{S}$. Then, the system
	\begin{equation} \label{MN}
		\begin{bmatrix} \mathbf M(\lambda) & \hspace{-2mm}\mathbf N(\lambda) \end{bmatrix} =\left[\scriptsize\begin{array}{cc|cc} \hspace{-1mm}A_x  -\lambda I_{n-p}& B_x A_{12}  & A_xB_x -B_x A_{11} +B_xA_{12}K  & B_xB_1 \\
			O & \hspace{-2mm}A_{22}+KA_{12} -\lambda I_p & \hspace{-1mm}KA_{12}K+A_{22}K-KA_{11}-A_{21}  &\hspace{-1mm} KB_1+B_2 \\ 
			\hline C_x & O & C_xB_x  & O \\
		\end{array}\right] 
	\end{equation}
	designates an LCF over $\mathbb{S}$ of $\mathbf{G}(\lambda)=\mathbf{M}^{-1}(\lambda)\mathbf{N}(\lambda)$.
\end{lemma}
\begin{proof}
	Let $\mathbf G(\lambda)$ be given by a minimal realization of type \eqref{ss2a}-\eqref{ss2b}, namely
	$$\mathbf G(\lambda)=\left[\begin{array}{cc|c}  A_{11} - \lambda I_p&  A_{12}  &B_1 \\ 
		A_{21}  & A_{22} - \lambda I_{n-p}   & B_2 \\ \hline
		I_p & O & O \\ \end{array}\right].$$
	We wish to retrieve \eqref{MN} by employing Theorem~\ref{SIMAX99}. Thus, we first apply a state equivalence transformation given by $\displaystyle T:=\footnotesize\begin{bmatrix} B_x & O \\ K & I_{n-p} \end{bmatrix}$ in order to obtain the minimal realization expressed as
	\begin{equation} \label{oooo2}
			\mathbf G(\lambda) = \left[\small\begin{array}{cc|c}B_x (A_{11}- A_{12} K)B_x^{-1}-\lambda I_p  & B_x A_{12}  & B_x B_1 \\ 
				\big( K A_{11} - KA_{12}K + A_{21} -A_{22}K  \big)B_x^{-1}  &KA_{12}+ A_{22} - \lambda I_{n-p}  & K B_1+B_2 \\  \hline
				B_x^{-1} & O & O \\ \end{array}\right].\normalsize
	\end{equation}
	Define now
	$
	F:=\scriptsize\begin{bmatrix}  A_xB_x - B_x A_{11} + B_xA_{12}K
		\\ KA_{12}K+A_{22}K-KA_{11}-A_{21} \end{bmatrix}$ and $U:=C_xB_x$. Since $B_x$ and $C_x$ are invertible, then so is $U$. Moreover, denote by $A_\mathbf{G}$ and $C_\mathbf{G}$ the state and, respectively, the output matrix of the realization from \eqref{oooo2}. Since $A_x$ and $A_{22}+KA_{12}$ have all their eigenvalues located in $\mathbb{S}$, we have that $A_\mathbf{G}+FC_\mathbf{G}=\scriptsize\begin{bmatrix}
		A_x  & B_x A_{12}\\
		O & A_{22}+KA_{12}
	\end{bmatrix}$ has all its eigenvalues located in $\mathbb{S}$. Since \eqref{oooo2} is minimal, points $\mathbf{a)}$, $\mathbf{b)}$ and $\mathbf{c)}$ of Theorem~\ref{SIMAX99} are satisfied and, thus, \eqref{MN} designates a left coprime factorization over $\mathbb{S}$ of $\mathbf{G}(\lambda)$.
\end{proof}\newpage

The final auxiliary result stated here will be crucial in proving Theorem~\ref{Main}.

\begin{lemma}\label{lem:SRTR_aux}
	Let $\mathbf{G}_1\in\mathbb{R}(\lambda)^{p\times p}$ and $\mathbf{G}_2\in\mathbb{R}(\lambda)^{p\times m}$ such that we can express the following realization of order $n-p>0$, given explicitly by
	\begin{equation}\label{eq:auxr1}
		\begin{bmatrix}\mathbf{G}_1(\lambda)&\mathbf{G}_2(\lambda)\end{bmatrix}=\left[\scriptsize\begin{array}{c|cc}{A}-\lambda I_{n-p}&{B}_1&{B}_2\\\hline{C}&{D}_1&{D}_2\end{array}\right].
	\end{equation}
	Then, we can obtain the following realization of order $n$ for the TFM defined as
	\begin{equation}\label{eq:auxr2}
		\overline{\mathbf{G}}(\lambda):=\big(I_p-\lambda^{-1}\mathbf{G}_1(\lambda)\big)^{-1}\big(\lambda^{-1}\mathbf{G}_2(\lambda)\big)=\left[\scriptsize\begin{array}{cc|c}\hspace{-1mm}{D}_1-\lambda I_{p}&{C}&{D}_2\\{B}_1&\hspace{-1mm}{A}-\lambda I_{n-p}\hspace{-1mm}&{B}_2\\\hline I_p&O&O\end{array}\right].
	\end{equation}
\end{lemma}
\begin{proof}
	By first employing the realization $\frac{1}{\lambda}I_p=\left[\scriptsize\begin{array}{c|c}-\lambda I_p&I_p\\\hline I_p&O\end{array}\right]$ alongside the one given in \eqref{eq:auxr1}, routine state-space-based computations yield the following identity
	\begin{align}\label{eq:aux1}
		\overline{\mathbf{G}}(\lambda)=&\ \big(I_p-\lambda^{-1}\mathbf{W}(\lambda)\big)^{-1}\big(\lambda^{-1}\mathbf{V}(\lambda)\big)\\=&\left(I_p-\left[\footnotesize\begin{array}{c|c}
			-\lambda I_{p}&I_p\\\hline I_p&O
		\end{array}\right]\left[\footnotesize\begin{array}{c|c}
			{A}-\lambda I_{n-p}&{B}_1\\\hline{C}&{D}_1
		\end{array}\right]\right)^{-1}\left[\footnotesize\begin{array}{c|c}
			-\lambda I_{p}&I_p\\\hline I_p&O
		\end{array}\right]\left[\footnotesize\begin{array}{c|c}
			{A}-\lambda I_{n-p}&{B}_2\\\hline{C}&{D}_2
		\end{array}\right]\nonumber\\
		=&\footnotesize\left[\begin{array}{cc|c}
			-\lambda I_{p}&{C}&{D}_1\\
			O&{A}-\lambda I_{n-p}&{B}_1\\\hline
			-I_p&O&I_p
		\end{array}\right]^{-1}\hspace{-1mm}\footnotesize\left[\begin{array}{cc|c}
			-\lambda I_{p}&{C}&{D}_1\\
			O&{A}-\lambda I_{n-p}&{B}_1\\\hline
			I_p&O&O
		\end{array}\right]\nonumber\\=&\scriptsize\left[\begin{array}{cccc|c}{D}_1-\lambda I_{p}&{C}&{D}_1&O&O\\{B}_1&{A}-\lambda I_{n-p}&{B}_1&O&O\\
			O&O&-\lambda I_p&{C}&{D}_2\\
			O&O&O&{A}-\lambda I_{n-p}&{B}_2\\\hline I_p&O&I_p&O&O\end{array}\right].\nonumber
	\end{align}
	Applying now a state equivalence transformation $T_o:=\tiny\begin{bmatrix}
		I_p&O&I_p&O\\O&I_{n-p}&O&I_{n-p}\\O&O&I_p&O\\O&O&O&I_{n-p}
	\end{bmatrix}$ to the state-space realization of order $2n$ from \eqref{eq:aux1}, we obtain the fact that
	\begin{align}\label{eq:aux2}
		\overline{\mathbf{G}}(\lambda)
		=\hspace{-0.5mm}\left[\scriptsize\begin{array}{cccc|c}\hspace{-1mm}{D}_1-\lambda I_{p}&{C}&O&O&{D}_2\\{B}_1&{A}-\lambda I_{n-p}&O&O&{B}_2\\
			O&O&-\lambda I_p&{C}&{D}_2\\
			O&O&O&{A}-\lambda I_{n-p}&{B}_2\\\hline I_p&O&O&O&O\end{array}\right].
	\end{align}
	Finally, after eliminating the last $n$ unobservable components from the state vector of the realization from \eqref{eq:aux2}, we retrieve the realization of $\overline{\mathbf{G}}(\lambda)$ given in \eqref{eq:auxr2}.
\end{proof}

\section{Proofs of the main results}\label{sec:ap_D}\medskip

{\bf Proof of Theorem~\ref{Main}} We begin by stating that, throughout this entire proof, the TFM of the system given by \eqref{ss2a}-\eqref{ss2b} will be denoted by $\mathbf{G}(\lambda)$, just as in Definition~\ref{obiectu}. In order to prove point $\mathbf{a)}$, we begin by first defining the following matrix $A_K:=K A_{11} - KA_{12}K + A_{21} -A_{22}K$ and by employing Lemma~\ref{lem:SRTR_aux} for the $\big( \mathbf W(\lambda), \mathbf V(\lambda) \big)$ pair along with its associated realization given in \eqref{VWdef}, to get that
\begin{equation}\label{eq:idGt}
	\overline{\mathbf{G}}(\lambda):=\big(I_p-\lambda^{-1}\mathbf{W}(\lambda)\big)^{-1}\big(\lambda^{-1}\mathbf{V}(\lambda)\big)\hspace{-0.5mm}=\hspace{-0.5mm}\left[\tiny\begin{array}{cc|c}\hspace{-1mm}A_{11}-A_{12}K-\lambda I_{p}&A_{12}&B_1\\A_K&\hspace{-1mm}A_{22}+KA_{12}-\lambda I_{n-p}\hspace{-1mm}&\hspace{-1mm}KB_1+B_2\hspace{-1mm}\\\hline I_p&O&O\end{array}\right],
\end{equation}
and by applying a state equivalence transformation given by $\widetilde{T}:=\scriptsize\begin{bmatrix}I_p&O\\-K&I_{n-p}\end{bmatrix}$ to the realization in \eqref{eq:idGt}, we obtain precisely the realization from \eqref{ss2a}-\eqref{ss2b}. Therefore
\begin{equation}\label{eq:idGGt}
	\mathbf{G}(\lambda)=\overline{\mathbf{G}}(\lambda)=\big(I_p-\lambda^{-1}\mathbf{W}(\lambda)\big)^{-1}\big(\lambda^{-1}\mathbf{V}(\lambda)\big)=\big(\lambda I_p-\mathbf{W}(\lambda)\big)^{-1}\mathbf{V}(\lambda).
\end{equation}

In addition to the equality from \eqref{eq:idGGt}, notice that $\mathbf W(\lambda)$ and $\mathbf V(\lambda)$ are expressed via state-space realizations, which means that they are guaranteed to be proper TFMs. Moreover, from the expression of $\begin{bmatrix}
	\mathbf W(\lambda)&\mathbf V(\lambda)
\end{bmatrix}$ given in \eqref{VWdef}, we get that the latter TFM's McMillan degree is at most $n-p$. thus, by Definition~\ref{obiectu}, any $\displaystyle \big( \mathbf W(\lambda), \mathbf V(\lambda) \big)$ obtained via \eqref{VWdef} is indeed an SRTR pair of the system from the result's statement.

To prove point $\mathbf{b)}$, we first show that minimality is necessary for \eqref{VWdef} to produce, given the plant's realization, every SRTR pair associated with the latter realization and having a McMillan degree of at most $n-p$. Let us assume that this is not the case and consider the system described by the realization of type \eqref{ss2a}-\eqref{ss2b}
\begin{equation*}
	\mathbf{G}(\lambda)=\left[\scriptsize\begin{array}{cc|c}
		-1-\lambda&1&1\\0&-1-\lambda&0\\\hline 1&0&0
	\end{array}\right],
\end{equation*}
which is observable (as per Assumption~\ref{Leuenberger}), but not controllable, and for which \eqref{VWdef} produces the class of SRTR pairs
\begin{align}
	\begin{bmatrix}
		\mathbf W(\lambda) & \mathbf V(\lambda)
	\end{bmatrix}= &\left[\footnotesize\begin{array}{c|cc}
		k-1&-k^2&k\\\hline 1&-1-k&1 \end{array}\right]=\begin{bmatrix}
		-\frac{(k+1)\lambda+1}{\lambda+1-k}&\ \frac{\lambda+1}{\lambda+1-k}
	\end{bmatrix},\ \forall\ k\in\mathbb{R}.\label{eq:clk}
\end{align}
Note also that $\begin{bmatrix}
	\overline{\mathbf W}(\lambda) & \overline{\mathbf V}(\lambda)
\end{bmatrix}=\begin{bmatrix}
	-2&
	\frac{\lambda+2}{\lambda+1}
\end{bmatrix}$ forms an SRTR pair of the same system having a McMillan degree equal to $1$ and that this representation cannot be obtained for any real $k$ in \eqref{eq:clk}. This contradiction shows, therefore, that minimality of the realization from \eqref{ss2a}-\eqref{ss2b} is necessary for \eqref{VWdef} to characterize the desired class.

To show the sufficiency of starting from a system given by minimal realization \eqref{ss2a}-\eqref{ss2b} of order $n$, having $\mathbf{G}(\lambda)$ as its TFM, we will first consider an arbitrary SRTR pair $\displaystyle \big( \mathbf W(\lambda), \mathbf V(\lambda) \big)$ associated with the aforementioned system. Thus, by Definition~\ref{obiectu}, we can always obtain a realization of order $n-p$ for the following TFM\vspace{-1mm}
\begin{equation}\label{eq:SRTR_init}
	\begin{bmatrix}\mathbf{W}(\lambda)&\mathbf{V}(\lambda)\end{bmatrix}=\left[\scriptsize\begin{array}{c|cc}\overline{A}-\lambda I_{n-p}&\overline{B}_1&\overline{B}_2\\\hline\vspace{-2mm}\\\overline{C}&\overline{D}_1&\overline{D}_2\end{array}\right],\vspace{-1mm}
\end{equation}
which we will employ, in order to obtain a minimal realization of $\mathbf{G}(\lambda)$. Since \eqref{WV} holds, then $\mathbf{G}(\lambda)=\big(I_p-\lambda^{-1}\mathbf{W}(\lambda)\big)^{-1}\big(\lambda^{-1}\mathbf{V}(\lambda)\big)$ and, by employing Lemma~\ref{lem:SRTR_aux} for the $\big( \mathbf W(\lambda), \mathbf V(\lambda) \big)$ pair along with its associated realization from \eqref{eq:SRTR_init}, we get that
\begin{equation}\label{eq:idGr}
	\mathbf{G}(\lambda)
	=\left[\scriptsize\begin{array}{cc|c}\overline{D}_1-\lambda I_{p}&\overline{C}&\overline{D}_2\\\overline{B}_1&\overline{A}-\lambda I_{n-p}&\overline{B}_2\\\hline I_p&O&O\end{array}\right].
\end{equation}
Notice that the obtained realization from \eqref{eq:idGr} has the same order as the McMillan degree of $\mathbf{G}(\lambda)$, making it minimal. We proceed to exploit the minimality of this realization and its particular structure, in order to show that \eqref{eq:SRTR_init} is related to \eqref{VWdef} by a state equivalence transformation, for some $K\in\mathbb{R}^{(n-p)\times p}$.

Consider all matrices $\overline{T}:=\scriptsize\begin{bmatrix}
	I_p&O\\LK&L
\end{bmatrix}$, where $K$ is arbitrary in $\mathbb{R}^{(n-p)\times p}$ and $L$ is any invertible matrix in $\mathbb{R}^{(n-p)\times(n-p)}$. Then, $\overline{T}^{-1}=\scriptsize\begin{bmatrix}
	I_p&O\\-K&L^{-1}
\end{bmatrix}$ represents the class of all invertible matrices in $\mathbb{R}^{n\times n}$ having the first block-row equal to $\begin{bmatrix}
	I_p&O
\end{bmatrix}$. Thus, we get that $\begin{bmatrix}
	I_p&O
\end{bmatrix}=\begin{bmatrix}
	I_p&O
\end{bmatrix}\overline{T}^{-1}$ if and only if $\overline{T}^{-1}$ has the indicated structure. Starting the minimal realization of type \eqref{ss2a}-\eqref{ss2b} from the result's statement and by applying the state-equivalence transformation given by $\overline{T}$, we obtain the class of all minimal realizations of $\mathbf{G}(\lambda)$ having an output matrix equal to $\begin{bmatrix}
	I_p&O
\end{bmatrix}$, given by
	\begin{equation}\label{eq:all_Gs}
		\mathbf{G}(\lambda)=\left[\scriptsize\begin{array}{cc|c}
			A_{11}-A_{12}K-\lambda I_p&A_{12}L^{-1}&B_1\\
			\hspace{-1mm}L(K A_{11} - KA_{12}K + A_{21} -A_{22}K  ) & L(A_{22}+ K A_{12})L^{-1}  -\lambda I_{n-p}&L(KB_1+B_2)\\\hline
			I_p&O&O
		\end{array}\right].
	\end{equation}

Thus, there must exist a $K$ and an invertible $L$ such that the rightmost term from \eqref{eq:idGr} coincides with \eqref{eq:all_Gs}, which makes \eqref{eq:SRTR_init} coincide with the following realization\vspace{1mm}
\begin{equation}\label{eq:all_SRTRs}
	\scriptsize\begin{bmatrix}\mathbf{W}(\lambda)&\hspace{-2mm}\mathbf{V}(\lambda)\end{bmatrix}= \left[\scriptsize\begin{array}{c|cc}
		\hspace{-1mm}L(A_{22}+ K A_{12})L^{-1}  -\lambda I_{n-p}  & L(K A_{11} - KA_{12}K + A_{21} -A_{22}K  ) &  \hspace{-2mm}L(KB_1+B_2)  \\ 
		\hline  A_{12}L^{-1}   & A_{11} - A_{12}K & B_1  \\ \end{array}\right]\normalsize\hspace{-1mm}.
\end{equation}
Apply now a state equivalence transformation $T_L:=L^{-1}$ to \eqref{eq:all_SRTRs} in order to conclude that, when starting from a minimal realization of type \eqref{ss2a}-\eqref{ss2b} of $\mathbf{G}(\lambda)$, there must exist a $K\in\mathbb{R}^{(n-p)\times p}$ which allows $\begin{bmatrix}\mathbf{W}(\lambda)&\mathbf{V}(\lambda)\end{bmatrix}$ to be expressed via \eqref{VWdef}. 

\medskip

{\bf Proof of Theorem~\ref{CupruMin}} In order to prove our statement, we first require a Rosenbrock-type realization of $\begin{bmatrix} \lambda I_p -\mathbf W(\lambda) &  \mathbf V(\lambda) \end{bmatrix}$. Obtain first a state-space realization of $\begin{bmatrix} -\mathbf W(\lambda) &  \mathbf V(\lambda) \end{bmatrix}$ by using \eqref{VWdef} and multiplying its input and feedthrough matrices to the right with $\footnotesize\begin{bmatrix}
	-I_p&O\\O&I_m
\end{bmatrix}$. Write now that $\begin{bmatrix}
	\lambda I_p& O
\end{bmatrix}=\left[\footnotesize\begin{array}{c|cc}
	I_p&-\lambda I_p&O\\\hline I_p&O&O
\end{array}\right]$ and, by expressing the parallel connection formula between this Rosenbrock-type representation and the computed realization of $\begin{bmatrix} -\mathbf W(\lambda) &  \mathbf V(\lambda) \end{bmatrix}$, use it to form the sought-after Rosenbrock-type realization of $\begin{bmatrix} \lambda I_p -\mathbf W(\lambda) &  \mathbf V(\lambda) \end{bmatrix}$, given explicitly by
\begin{multline} \label{coprime}
		\begin{bmatrix} \lambda I_p-\mathbf W(\lambda) &\hspace{-2mm} \mathbf V(\lambda) \end{bmatrix} = \left[\footnotesize\begin{array}{c|c}
			\widetilde{A}-\lambda\widetilde{E}&\widetilde{B}-\lambda\widetilde{F}  \\\hline
			\widetilde{C}&\widetilde{D} 
		\end{array}\right]:= \\
		 :=\left[\footnotesize\begin{array}{cc|cc} \hspace{-2mm}A_{22}+ K A_{12}  -\lambda I_{n-p} & O  & \hspace{-1mm}KA_{12}K + A_{22}K -K A_{11} - A_{21} & \hspace{-2mm}KB_1+B_2 \hspace{-1mm} \\
			O & I_p &  - \lambda I_p  & O \\ 
			\hline A_{12} & I_p & -A_{11} + A_{12}K  & B_1 \\
		\end{array}\right]\hspace{-1mm}.\normalsize
\end{multline}

If the realization from \eqref{coprime} is irreducible then, by Theorem~\ref{thm:iso}, the finite zeros of $\begin{bmatrix} \lambda I_p -\mathbf W(\lambda) &  \mathbf V(\lambda) \end{bmatrix}$ are the finite zeros of the realization's system matrix, namely\vspace{-1mm}
	\begin{equation} \label{mare}
	\mathcal{S}(\lambda) := \scriptsize\begin{bmatrix} A_{22}+ K A_{12}  -\lambda I_{n-p}  & O  & KA_{12}K + A_{22}K -K A_{11} - A_{21}   & KB_1+B_2 \\
		O & I_p & - \lambda I_p  & O \\ 
		A_{12} & I_p &  -A_{11} + A_{12}K  & B_1 \\
	\end{bmatrix}.\normalsize\vspace{-1mm}
\end{equation}
Additionally, the infinite zeros of $\begin{bmatrix} \lambda I_p -\mathbf W(\lambda) &  \mathbf V(\lambda) \end{bmatrix}$ can be deduced by using an irreducible Rosenbrock-type realization of $\begin{bmatrix}
	\frac{1}{\lambda}I_p-\mathbf{W}\left(\frac{1}{\lambda}\right)&\mathbf{V}\left(\frac{1}{\lambda}\right)
\end{bmatrix}$ and by computing the latter TFM's zeros at $\lambda=0$.

Thus, we divide the proof into two distinct parts. In part {\bf I)}, we show that the realization from \eqref{coprime} is indeed irreducible and that $\mathcal{S}(\lambda)$ from \eqref{mare} has full row rank $\forall\lambda\in\mathbb{C}$. In part {\bf II)}, we compute an irreducible Rosenbrock-type realization for $\begin{bmatrix}
	\frac{1}{\lambda} I_p-\mathbf{W}\left(\frac{1}{\lambda}\right)&\mathbf{V}\left(\frac{1}{\lambda}\right)
\end{bmatrix}$ and show that its own system matrix has full row rank at $\lambda=0$. Therefore, $\begin{bmatrix} \lambda I_p -\mathbf W(\lambda) &  \mathbf V(\lambda) \end{bmatrix}$ must have no finite or infinite zeros, and it must also have full row normal rank, due to Proposition~\ref{prop:fnr}.

{\bf I)} Due to Assumption~\ref{Observability} and Remark~\ref{Leuenberger}, notice that\vspace{-1mm}
\begin{equation*}
	\footnotesize\begin{bmatrix}
		\widetilde{A}-\lambda\widetilde{E}\\\widetilde{C}
	\end{bmatrix}=\footnotesize\begin{bmatrix} A_{22}+ K A_{12}  -\lambda I_{n-p}  & O \\
		O  & I_p \\
		A_{12} & I_p \end{bmatrix} \normalsize   =\footnotesize\begin{bmatrix} I_{n-p} & -K & K \\ O & I_p & O \\ O & O & I_p  \end{bmatrix} \footnotesize\begin{bmatrix} A_{22}  -\lambda I_{n-p}  & O \\
		O  & I_p \\
		A_{12} & I_p \end{bmatrix}\vspace{-1mm}
\end{equation*}
has full column rank $\forall\lambda\in\mathbb{C}$. Moreover, due to Assumption~\ref{Controllability}, note also that the pencil from the right-hand term of the following equality \vspace{-1mm}
\begin{equation*}
	\begin{array}{c}
		
		\footnotesize\begin{bmatrix} A_{22}+ K A_{12}  -\lambda I_{n-p}  & O  & KA_{12}K + A_{22}K -K A_{11} - A_{21}   & KB_1+B_2 \\
			O & I_p &   - \lambda I_p  & O \end{bmatrix}
		\normalsize=\\
		
		\footnotesize\begin{bmatrix}
			K&I_{n-p}\\-I_p&O
		\end{bmatrix} \small\begin{bmatrix} \lambda I_p -A_{11} & - A_{12} & B_1 & -I_p\\
			-A_{21}  &\lambda I_{n-p} - A_{22}  & B_2  & K\\  \end{bmatrix}\scriptsize\begin{bmatrix}
			O&O&I_p&O\\
			-I_{n-p}&O&-K&O\\
			O&O&O&I_m\\
			A_{12}&I_p&A_{12}K-A_{11}&B_1\\
		\end{bmatrix} 
	\end{array}
\end{equation*}\vspace{-1mm}
has full row rank $\forall\lambda\in\mathbb{C}$, which means that so does $\small\begin{bmatrix}
	\widetilde{A}-\lambda\widetilde{E}&\widetilde{B}-\lambda\widetilde{F}
\end{bmatrix}$.\newpage

By Definition~\ref{def:ired}, the realization from \eqref{coprime} is irreducible and we now inspect its finite zeros. It follows from the following identity
\begin{equation}\label{eq:huge}
	\begin{array}{c}
		\footnotesize\begin{bmatrix} A_{22}+ K A_{12}  -\lambda I_{n-p}& O  & KA_{12}K + A_{22}K -K A_{11} - A_{21} & KB_1+B_2 \\
			O & I_p & - \lambda I_p  & O \\ 
			A_{12} & I_p & -A_{11} + A_{12}K  & B_1 \\
		\end{bmatrix}\normalsize  =\\
		
		\footnotesize\begin{bmatrix}
			K&I_{n-p}&O\\
			O&O&I_p\\
			I_p&O&I_p
		\end{bmatrix}\footnotesize\begin{bmatrix}  \lambda I_p - A_{11} &  -A_{12} & B_1 &O \\
			-A_{21}  &  \lambda I_{n-p} - A_{22}    & B_2 &O \\
			-\lambda I_p & O & O & I_p \end{bmatrix}\footnotesize\begin{bmatrix}
			O&O&I_p&O\\
			-I_{n-p}&O&-K&O\\
			O&O&O&I_m\\
			O&I_p&O&O
		\end{bmatrix}\normalsize
	\end{array}
\end{equation}
that $\mathcal{S}(\lambda)$ has full row rank $\forall\lambda\in\mathbb{C}$ since so does the pencil from the right-hand term in \eqref{eq:huge}, once again due to Assumption~\ref{Controllability}. We conclude, by Proposition~\ref{prop:fnr}, that $\begin{bmatrix} \lambda I_p -\mathbf W(\lambda) &  \mathbf V(\lambda) \end{bmatrix}$ has full row normal rank and, by Theorem~\ref{thm:iso}, no finite zeros.

{\bf II)} Recall the Rosenbrock-type realization from \eqref{coprime} to get
\begin{equation}\label{eq:lcf_mu}
{\begin{bmatrix}
	\frac{1}{\lambda}I_p-\mathbf{W}\left(\frac{1}{\lambda}\right)&\mathbf{V}\left(\frac{1}{\lambda}\right)
\end{bmatrix}}=\widetilde{C}(\lambda\widetilde{A}-\widetilde{E})^{-1}(\widetilde{F}-\lambda\widetilde{B})+\widetilde{D}.
\end{equation}
Remark that, since the realization given by \eqref{coprime} is irreducible, then we have that $\begin{bmatrix}
\widetilde{E}-\lambda\widetilde{A}&\widetilde{F}-\lambda\widetilde{B}
\end{bmatrix}$ and $\begin{bmatrix}
\widetilde{E}^\top-\lambda\widetilde{A}^\top&\widetilde{C}^\top
\end{bmatrix}^\top$ have full row and, respectively, column rank $\forall\lambda\in\mathbb{C}\backslash\{0\}$. It is straightforward to check that
$
\begin{bmatrix}
\widetilde{E}&\widetilde{F}
\end{bmatrix}=\footnotesize\begin{bmatrix}
I_{n-p}&O&O&O\\O&O&I_p&O
\end{bmatrix}
$
has full row rank, while
$
\begin{bmatrix}
\widetilde{E}^\top&\widetilde{C}^\top
\end{bmatrix}^\top=\footnotesize\begin{bmatrix}
I_{n-p}&O&A_{12}^\top\\O&O&I_p
\end{bmatrix}^\top
$
has full column rank. Therefore, we conclude that \eqref{eq:lcf_mu} is an irreducible Rosebrock-type realization and we proceed to define the associated system matrix
$\widetilde{\mathcal{S}}(\lambda):=\small\begin{bmatrix}
\widetilde{E}-\lambda\widetilde{A}&\widetilde{F}-\lambda\widetilde{B}\\\widetilde{C}&\widetilde{D}
\end{bmatrix}$ in order to investigate the latter's zeros at $\lambda=0$. Notice that 
$
\normalsize\widetilde{\mathcal{S}}(0)=\scriptsize\begin{bmatrix}
I_{n-p}&O&O&O\\O&O&I_p&O\\A_{12}&I_p&-A_{11}+A_{12}K&B_1
\end{bmatrix}
$
has full row rank and, thus, it follows that \eqref{eq:lcf_mu} has no zeros at $\lambda=0$, which implies that $\begin{bmatrix} \lambda I_p -\mathbf W(\lambda) &  \mathbf V(\lambda) \end{bmatrix}$ does not have any zeros at infinity.

By the conclusions drawn in parts {\bf I)} and {\bf II)} of the proof, we get the fact that $\begin{bmatrix} \lambda I_p -\mathbf W(\lambda) &  \mathbf V(\lambda) \end{bmatrix}$ has full row normal rank along with no zeros, be they finite or infinite. Then, it follows that this compound TFM designates an FLCF of $\mathbf{G}(\lambda)$. \medskip

{\bf Proof of Theorem~\ref{step2}} 
Point $\mathbf{a})$ boils down to the displacement of poles at infinity, via the procedure from Theorem 3.1 of \cite{SIMAX}. We will apply this result using only state-space-based arguments, to avoid using the \emph{pencil} realizations from \cite{SIMAX}. 

Begin by recalling Theorem~\ref{Main} to conclude that the given SRTR pair must be expressed through a realization of type \eqref{VWdef} for some $K$ and a minimal realization of type \eqref{ss2a}-\eqref{ss2b}, mentioned in the result's statement. Thus, we can write that
\begin{equation}\normalsize
\mathbf \Theta(\lambda) \begin{bmatrix}  \lambda I_p -\mathbf W(\lambda)  &\hspace{-1mm}  \mathbf V(\lambda) \end{bmatrix}=\begin{bmatrix}
\lambda\mathbf{\Theta}(\lambda)&\hspace{-1mm}O
\end{bmatrix}+\mathbf \Theta(\lambda)\begin{bmatrix} \mathbf W(\lambda)  &\hspace{-1mm}  \mathbf V(\lambda) \end{bmatrix}\footnotesize\begin{bmatrix}
-I_p&O\\
O&I_m
\end{bmatrix}\normalsize,\label{eq:sum}
\end{equation}
where we have that
\begin{subequations}
\begin{align}
\begin{bmatrix}
	\lambda\mathbf{\Theta}(\lambda)&O
\end{bmatrix}=&\begin{bmatrix}
	C_x(\lambda I_p-A_x)^{-1}(\lambda I_p-A_x+A_x)B_x&O
\end{bmatrix}=\left[\footnotesize\begin{array}{c|cc}
	A_x-\lambda I_p&A_xB_x&O  \\\hline
	C_x&C_xB_x&O 
\end{array}\right]\label{eq:sum1_aux}\\
=&\left[\footnotesize\begin{array}{cc|cc}
	A_x-\lambda I_p&B_xA_{12}&A_xB_x&O  \\
	O&A_{22}+KA_{12}-\lambda I_{n-p}&O&O\\\hline
	C_x&O&C_xB_x&O
\end{array}\right].\label{eq:sum1}
\end{align}
\end{subequations}

The identity given in \eqref{eq:sum1} is obtained by introducing $n-p$ uncontrollable modes in the realization from \eqref{eq:sum1_aux}. We denote $A_K:=K A_{11} - KA_{12}K + A_{21} -A_{22}K$ and, using \eqref{VWdef} and \eqref{Theta}, we write
\begin{equation}\label{eq:sum2}
\mathbf \Theta(\lambda)\footnotesize\begin{bmatrix} \mathbf W(\lambda)  & \hspace{-2mm} \mathbf V(\lambda) \end{bmatrix}\begin{bmatrix}
-I_p&O\\
O&I_m
\end{bmatrix}\hspace{-1mm}=\hspace{-1mm}\left[\footnotesize\begin{array}{cc|cc}
\hspace{-2mm}A_x-\lambda I_p&\hspace{-3mm}B_xA_{12}&\hspace{-1mm}B_x(A_{12}K-A_{11}) &\hspace{-3mm} B_xB_1\hspace{-2mm} \\
\hspace{-2mm}O&\hspace{-3mm}A_{22}+KA_{12}-\lambda I_{n-p}\hspace{-1mm}&-A_K &\hspace{-3mm}  KB_1+B_2\hspace{-2mm}\\\hline
\hspace{-2mm}C_x&\hspace{-3mm}O&O&\hspace{-3mm}O\hspace{-2mm}
\end{array}\right].\normalsize
\end{equation}

Notice now that the state-space realizations from \eqref{eq:sum1} and \eqref{eq:sum2} have the same state and output matrices and, therefore, the parallel connection from \eqref{eq:sum} can be computed by adding together their input and feedthrough matrices, from which we directly obtain \eqref{MN}. Since the submatrices from the employed realization of the SRTR pair come from a minimal realization of type \eqref{ss2a}-\eqref{ss2b} for $\mathbf{G}(\lambda)$, we apply Lemma~\ref{lem:indeed_coprime} to get that the system from \eqref{eq:sum} describes an LCF over $\mathbb{S}$ of $\mathbf{G}(\lambda)$.

To prove point $\mathbf{b})$, begin by recalling the proof of Lemma~\ref{lem:SRTR_stab}, in order to state that the realization from \eqref{VWdef} is minimal, with all the eigenvalues of $A_{22}+KA_{12}$ being poles of the TFM. Notice also that, due to the Assumption~\ref{Observability}, Remark~\ref{Leuenberger} and the invertibility of $B_x$ and $C_x$, the realization given in \eqref{MN} is observable. Additionally, recall the matrix pencil from \eqref{eq:ctrb}, which has full row rank $\forall\,\lambda\in\mathbb{C}$. It follows that

\begin{equation*}
\begin{bmatrix}
B_x&O\\O&I_p
\end{bmatrix}\begin{bmatrix}
B_x^{-1}(A_x-\lambda I_p)&S_1(\lambda)\\O&S_2(\lambda)
\end{bmatrix}\begin{bmatrix}
I_p&O&B_x&O\\
O&O&-I_p&O\\
O&I_{n-p}&O&O\\
O&O&O&I_m
\end{bmatrix}
\end{equation*}
has full row rank $\forall\,\lambda\in\mathbb{C}$, since $B_x$ is invertible, which makes \eqref{MN} controllable. Since it is also observable, we get that the realization from \eqref{MN} is minimal, with all the eigenvalues of its pole pencil being poles of its associated TFM. Therefore, due to the minimality of both realizations from \eqref{VWdef} and \eqref{MN} and due to the shared eigenvalues of $A_{22}+KA_{12}$, the poles of \eqref{VWdef} are preserved among those of \eqref{MN}. \medskip

{\bf Proof of Theorem~\ref{inversa}} To prove point $\mathbf{a)}$, we first obtain from the definition of $A_x$ in the Theorem's statement and from the fact that $K$ is a solution of \eqref{Riccati} the following identity $
\begin{bmatrix}  F_1 \\ F_2\end{bmatrix}{=} \begin{bmatrix} A_x - A_{11} +A_{12}K
\\ A_{22}K  -KA_x - A_{21} \end{bmatrix}
$. We then plug this equality into the expression of \eqref{Kontroller} and we apply a state equivalence transformation which is given by $T:=\begin{bmatrix}I_p&O\\K&I_{n-p}\end{bmatrix}$, in order to obtain the realization from \eqref{MN1}.

To prove point $\mathbf{b)}$, we will simply reverse the procedure used to prove point $\mathbf{a)}$ in Theorem~\ref{step2}. This is done by choosing $B_x=I_p$, $C_x=U$ and $A_x$ as in the Theorem's statement, and then left-multiplying \eqref{MN1} with the inverse of \eqref{Theta}, as in \eqref{eq:lcf2srtr}. By doing so, we get that
\begin{align}
\footnotesize\begin{array}{r}
	(\lambda I_p-A_x)U^{-1}\hspace{-1mm}\begin{bmatrix}
		-\mathbf{M}(\lambda)&\hspace{-2mm}\mathbf{N}(\lambda)
	\end{bmatrix}
\end{array}\hspace{-2mm}\hspace{-1mm}=&\hspace{-1mm}\ (\lambda I_p-A_x)\begin{bmatrix}
I_p&O
\end{bmatrix}\times\nonumber\\
&\hspace{-1mm}\times\hspace{-1mm}\footnotesize\begin{bmatrix}
(\lambda I_p-A_x)^{-1}&\hspace{-2mm}(\lambda I_p-A_x)^{-1}A_{12}(\lambda I_{n-p}-A_{22}-KA_{12})^{-1}\\
O&(\lambda I_{n-p}-A_{22}-KA_{12})^{-1}
\end{bmatrix}\hspace{-1mm}\times\nonumber\\
&\hspace{-1mm}\times\hspace{-1mm}\footnotesize\begin{bmatrix}
A_{11} -A_{12}K +A_x & B_1\\
KA_{11}-KA_{12}K+A_{21}-A_{22}K  & \hspace{-1mm}KB_1+B_2
\end{bmatrix}\hspace{-1mm}+\hspace{-1mm}\begin{bmatrix}
A_x-\lambda I_p&\hspace{-2mm}O
\end{bmatrix}.\label{eq:horror1}
\end{align}

Denoting $A_K:=K A_{11} - KA_{12}K + A_{21} -A_{22}K$, the TFM in \eqref{eq:horror1} is equal to
\begin{align}
\footnotesize\begin{array}{r}
	(\lambda I_p-A_x)U^{-1}
\end{array}\hspace{-2mm}\footnotesize\begin{bmatrix}
-\mathbf{M}(\lambda)&\hspace{-2mm}\mathbf{N}(\lambda)
\end{bmatrix}=&\footnotesize\begin{bmatrix}
A_x-\lambda I_p\hspace{-2mm}&O
\end{bmatrix}+\footnotesize\begin{bmatrix}
I_p & A_{12}(\lambda I_{n-p}-A_{22}-KA_{12})^{-1}
\end{bmatrix}\times\nonumber\\
&\qquad\qquad\quad\,\times\footnotesize\begin{bmatrix}
A_{11} -A_{12}K -A_x  & B_1\\
A_K  & KB_1+B_2
\end{bmatrix}\nonumber\\
 =&\footnotesize\begin{bmatrix}
-\lambda I_p\hspace{-2mm}&O
\end{bmatrix}+\left[\footnotesize\begin{array}{c|cc}
\hspace{-2mm}A_{22}+KA_{12}-\lambda I_{n-p}&A_K&\hspace{-2mm}KB_1+B_2\hspace{-2mm}\\\hline
A_{12}&A_{11}-A_{12}K&\hspace{-2mm}B_1
\end{array}\right]\hspace{-1mm}.\hspace{-1mm}\label{eq:horror2}
\end{align}

Recall Definition~\ref{obiectu} and Theorem~\ref{Main} to notice that the system from the right-hand term in the final equality of \eqref{eq:horror2}, which is in parallel connection to $\begin{bmatrix}
-\lambda I_p&O
\end{bmatrix}$, has a realization of type \eqref{VWdef}. Therefore, its TFM designates an SRTR pair of the system and we proceed to rewrite \eqref{eq:horror2} as follows
\begin{equation*}\label{eq:horror3}
(\lambda I_p-A_x)U^{-1}\begin{bmatrix}
-\mathbf{M}(\lambda)&\mathbf{N}(\lambda)
\end{bmatrix}=\begin{bmatrix}
\mathbf W(\lambda)-\lambda I_p &\mathbf V(\lambda)
\end{bmatrix},
\end{equation*}
from which we retrieve \eqref{eq:lcf2srtr}, along with the stable SRTR pair.

Point $\mathbf{c)}$ follows directly by inspecting the realization of the TFM which designates the SRTR pair obtained in \eqref{eq:horror2} and by recalling that $K$ is the solution of \eqref{Riccati}.
\medskip

\begin{figure}[t]
\centering
\begin{tikzpicture}[scale=1, every node/.style={transform shape}]
\node(n1)at(0,0){};
\node(n2)[right of=n1, node distance=4em]{};
\node(n21)[circle,draw,minimum height=8,above of=n2, node distance=1em]{};
\node(n22)[circle,draw,minimum height=8,below of=n2, node distance=1em]{};
\node(n3)[minimum height=0.5cm, minimum width=0.5cm,draw,right of=n2, node distance=5em]{$ \begin{bmatrix}
		\frac{1}{\lambda}\mathbf{W}^\top(\lambda)\\\frac{1}{\lambda}\mathbf{V}^\top(\lambda)
	\end{bmatrix}^\top $};
\node(n4)[circle,draw,minimum height=8,right of=n3, node distance=5em]{};
\node(n41)[circle,draw,minimum height=8,above of=n4, node distance=5em]{};
\node(n42)[circle,draw,minimum height=8,below of=n4, node distance=5em]{};
\node(n42m)[below of=n4, node distance=2.5em]{$\zeta$};
\node(n41m)[above of=n4, node distance=2.5em]{$w$};
\node(n5)[minimum height=0.5cm, minimum width=0.5cm,draw,right of=n4, node distance=4em]{$\begin{bmatrix}
		I_m\\\mathbf{G}(\lambda)
	\end{bmatrix}$};

\draw[-latex](n41m)--(n4) node[near end, right]{$ $};
\draw[-latex](n42m)--(n42) node[near end, left]{$ $};
\draw[-latex](n41m)--(n41) node[midway, left]{\bf{--}};

\draw[-latex](n3)--(n4) node[near end, above]{$ $};
\draw[-latex](n4)--(n5);

\draw[-latex](20em,1em)--(23em,1em)|-(n41) node[near end, below]{$v$} node[very near end, below]{\hspace{-2em}$ $};
\draw[-latex](20em,-1em)--(23em,-1em)|-(n42) node[near end, above]{$y-\zeta$} node[very near end, above]{\hspace{-2em}$ $};

\draw[-latex](n41)-|(n21) node[near start, below]{$u$};
\draw[-latex](n42)-|(n22) node[near start, above]{$z-r$} node[very near end, right]{};

\draw[-latex](n21)--(5.8em,1em);
\draw[-latex](n22)--(5.8em,-1em);

\draw[-latex](0em,1em)--(n21) node[near start,above]{$\delta_u$} node[near end, above]{$ $};
\draw[-latex](0em,-1em)--(n22) node[near start,below]{$r$} node[near end, below]{$ $};
\end{tikzpicture}
\caption{Equivalent feedback interconnection}\label{fig:eq_con}
\hrulefill
\end{figure}
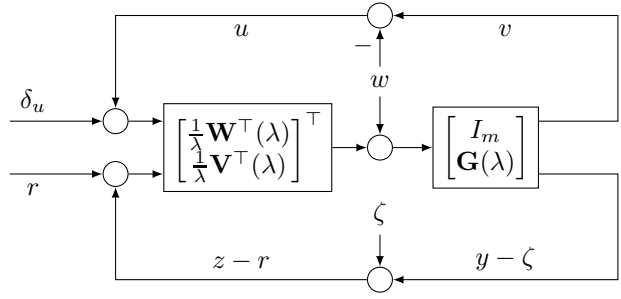

{\bf Proof of Theorem~\ref{thm:IO_stab}} The proof boils down to showing that $\mathbf{K}_d(\lambda)$ is a stabilizing controller for $\begin{bmatrix}
I_p&\mathbf{G}^\top(\lambda)
\end{bmatrix}^\top$, considering a standard feedback loop as in Fig.~\ref{fig:2Block}. This represents a sufficient condition for the closed-loop interconnection shown in Fig.~\ref{fig:eq_con}, located at the top of this page, to be internally stable. The desired result follows by noticing that Fig.~\ref{fig:eq_con} is equivalent to Fig.~\ref{fig:2BlockAgain}. In order to show that $\mathbf{K}_d(\lambda)$ is a stabilizing controller for $\begin{bmatrix}
I_p&\mathbf{G}^\top(\lambda)
\end{bmatrix}^\top$, we point out (see Section 5.3 of \cite{zhou}) that an equivalent condition for this is the fact that all the entries of $\mathbf{H}(\lambda)$ are stable, where
\begin{equation}\label{eq:all_cl}
\mathbf{H}(\lambda):=\small\begin{bmatrix}
I_p\\I_p\\\mathbf{G}(\lambda)
\end{bmatrix}\hspace{-1mm}\left(I_p-\begin{bmatrix}
\frac{1}{\lambda}\mathbf{W}(\lambda) & \frac{1}{\lambda}\mathbf{V}(\lambda)
\end{bmatrix}\begin{bmatrix}
I_p\\\mathbf{G}(\lambda)
\end{bmatrix}\right)^{-1}\begin{bmatrix}
I_p & \frac{1}{\lambda}\mathbf{W}(\lambda) & \frac{1}{\lambda}\mathbf{V}(\lambda)
\end{bmatrix}.\normalsize
\end{equation}

Express now any LCF over $\mathbb{S}$ of $\mathbf{G}^\top(\lambda)=\mathbf{M}_{\mathbf{G}}^{-1}(\lambda)\mathbf{N}_{\mathbf{G}}(\lambda)$ to get that $\mathbf{G}(\lambda)=\mathbf{N}^\top_{\mathbf{G}}(\lambda)\mathbf{M}_{\mathbf{G}}^{-\top}(\lambda)$, which designates a \emph{right coprime factorization (RCF) over $\mathbb{S}$} of $\mathbf{G}(\lambda)$. Then, select any stable $\mathbf{\Theta}(\lambda)$ as in Theorem~\ref{step2} and apply this result, to get that
\begin{equation}\label{eq:Theta_K}
\begin{bmatrix} \mathbf M_{\mathbf{K}}(\lambda) &  \mathbf N_{\mathbf{K}}(\lambda) \end{bmatrix}:= \mathbf \Theta(\lambda) \begin{bmatrix}   \lambda I_p -\mathbf W(\lambda)  &  \mathbf V(\lambda) \end{bmatrix}
\end{equation}
designates an LCF over $\mathbb{S}$ of the stabilizing controller $\mathbf{K}(\lambda)=\mathbf M_{\mathbf{K}}^{-1}(\lambda)\mathbf N_{\mathbf{K}}(\lambda)$.

Since $\mathbf{K}(\lambda)$ is a stabilizing controller for $\mathbf{G}(\lambda)$, we have by point $5$ of Lemma 5.10 in \cite{zhou} that $(\mathbf M_{\mathbf{K}}(\lambda)\mathbf M^\top_{\mathbf{G}}(\lambda)-\mathbf N_{\mathbf{K}}(\lambda)\mathbf N^\top_{\mathbf{G}}(\lambda))^{-1}$ must be stable. Finally, by \eqref{eq:Theta_K}, we have the fact that  $\mathbf{\Theta}(\lambda)=\mathbf{M}_{\mathbf{K}}(\lambda)(\lambda I_p-\mathbf{W}(\lambda))^{-1}$ and, by noticing that $\mathbf{M}_{\mathbf{K}}(\lambda)$, $\mathbf{W}(\lambda)$ along with $\mathbf{\Theta}(\lambda)$ are all stable TFMs, we get that $\lambda\mathbf{\Theta}(\lambda)=\mathbf{M}_{\mathbf{K}}(\lambda)+\mathbf{\Theta}(\lambda)\mathbf{W}(\lambda)$ is also a stable TFM. With this fact established, return to \eqref{eq:all_cl} and rewrite it as
\begin{align}
\mathbf{H}(\lambda)=&\footnotesize\begin{bmatrix}
I_p\\I_p\\\mathbf{G}(\lambda)
\end{bmatrix}\left(I_p-\left(\lambda I_p-\mathbf{W}(\lambda)\right)^{-1}\mathbf{V}(\lambda)\mathbf{G}(\lambda)\right)^{-1}\times\label{eq:Shai_Hulud}\\
&\qquad\qquad\qquad\qquad\qquad\qquad\qquad\times\left(\lambda I_p-
\mathbf{W}(\lambda)\right)^{-1}\begin{bmatrix}
\lambda I_p & \mathbf{W}(\lambda) & \mathbf{V}(\lambda)
\end{bmatrix}\nonumber\\
=&\footnotesize\begin{bmatrix}
\mathbf{M}_{\mathbf{G}}^{\top}(\lambda)\\\mathbf{M}_{\mathbf{G}}^{\top}(\lambda)\\\mathbf{N}_{\mathbf{G}}^{\top}(\lambda)
\end{bmatrix}\left(\mathbf{M}_{\mathbf{K}}(\lambda)\mathbf{M}_{\mathbf{G}}^{\top}(\lambda)-\mathbf{N}_{\mathbf{K}}(\lambda)\mathbf{N}^\top_{\mathbf{G}}(\lambda)\right)^{-1}\begin{bmatrix}
\lambda\mathbf{\Theta}(\lambda) & \mathbf{\Theta}(\lambda)\mathbf{W}(\lambda) & \mathbf{N}_{\mathbf{K}}(\lambda)
\end{bmatrix}.\nonumber
\end{align}

Notice that all the (block-)entries of the TFMs from the final equality in \eqref{eq:Shai_Hulud} are products of stable TFMs, which makes all of them stable. Thus, the closed-loop interconnections showcased in Figs.~\ref{fig:2BlockAgain} and~\ref{fig:eq_con} are internally stable. This confirms that the special implementation given in \eqref{eq:SRTR_ctl}-\eqref{eq:SRTR_conv} is indeed stabilizing, in both the continuous- and the discrete-time case. \medskip

{\bf Proof of Proposition~\ref{prop:not_stab}} We will prove our statement by showing that the TFM given in \eqref{eq:SRTR_ctl} is guaranteed to have $p$ poles in $s=0$ for any stable SRTR pair obtained via a minimal realization \eqref{ss2a}-\eqref{ss2b}, making it continuous-time unstable. Begin by considering the minimal state-space realization $\frac{1}{s}I_p=\left[\footnotesize\begin{array}{c|c}
-s I_p & I_p\\\hline I_p & O
\end{array}\right]$ and let the stable SRTR pair $(\mathbf{W}(s),\mathbf{V}(s))$ be as in \eqref{VWdef}. Then, a realization of $\mathbf{K}_d(s)$ is
\begin{multline} \label{VWalt}
	\mathbf{K}_d(s)=\left[\footnotesize\begin{array}{c|c}
		A_d-sI_n&B_d\\\hline C_d&D_d
	\end{array}\right]:=\\:= \left[\scriptsize\begin{array}{cc|cc}
		-sI_p & A_{12}   & A_{11} - A_{12}K & B_1  \\
		O & A_{22}+ K A_{12}  -s I_{n-p}  & K A_{11} - KA_{12}K + A_{21} -A_{22}K   &  KB_1+B_2  \\\hline 
		I_p & O & O & O
	\end{array}\right].\normalsize
\end{multline}
Notice now that, due to Assumption~\ref{Observability} and due to Remark~\ref{Leuenberger}, the matrix pencil
$$
\begin{bmatrix}
A_d-sI_n\\C_d
\end{bmatrix}=\footnotesize\begin{bmatrix}
I_p&O&O\\K&I_{n-p}&O\\O&O&I_p
\end{bmatrix}\footnotesize\begin{bmatrix}
-sI_p&A_{12}\\sK&A_{22}-sI_{n-p}\\I_p&O
\end{bmatrix}\normalsize
$$
has full column rank $\forall s\in\mathbb{C}$, making \eqref{VWalt} observable. Finally, consider the identity  
\begin{equation} \label{eq:Kd_ctrb}
	\small\begin{bmatrix}
		A_d-sI_n&\hspace{-2mm}B_d
	\end{bmatrix}=\footnotesize\begin{bmatrix}
		I_p&O\\K&I_{n-p}
	\end{bmatrix}\footnotesize\begin{bmatrix}
		-sI_p & \hspace{-1mm}A_{11} -sI_{p} & A_{12}    & B_1  \\
		sK & A_{21} & \hspace{-1mm}A_{22}-sI_{n-p}     &  B_2
	\end{bmatrix}\tiny\begin{bmatrix}
		I_p&O&-I_p&O\\
		O&O&\phantom{-}I_p&O\\
		O&I_{n-p}&-K&O\\
		O&O&\phantom{-}O&I_m
	\end{bmatrix}.\normalsize
\end{equation}
Due to Assumption~\ref{Controllability}, the left-hand pencil in \eqref{eq:Kd_ctrb} has full row rank $\forall s\in\mathbb{C}$, making the realization from the final equality in \eqref{VWalt} controllable. Thus, we conclude that the realization of $\mathbf{K}_d(s)$ from \eqref{VWalt} is minimal, which implies that all $p$ eigenvalues of $A_d$ at $s=0$ are poles of the system's TFM, making it continuous-time unstable.
\medskip

{\bf Proof of Theorem~\ref{thm:implem}} The proof largely follows the same reasoning as the one belonging to Theorem III.6 in \cite{NRF} and can be broken down into its two constituent parts: the discrete-time case $\mathbf{I)}$ and the continuous-time one $\mathbf{II)}$, respectively.\newpage

$\mathbf{I)}$ Begin with the discrete-time case and recall from Remark~\ref{rem:discr_stab} that $\mathbf{K}_d(z)$ is stable, when starting from a stable SRTR pair, and that each of its rows is also a stable TFM. Thus, the stabilizable and detectable realizations for each such row, represented by $e_i^\top\mathbf{K}_d(z)=\left[\footnotesize\begin{array}{c|c}
	\widetilde A_{di}-zI_{n_i}&\widetilde B_{di}\\\hline\vspace{-1.5mm}\\ \widetilde C_{di}&\widetilde D_{di}
\end{array}\right]$, possess pole pencils $\widetilde A_{di}-zI_{n_i}$ which have all their eigenvalues located in $\mathbb{S}$. Using these $p$ representations, form the realization
\begin{equation}
\mathbf{K}_d(z)=\left[\begin{array}{c:c:c}
	\mathbf{K}_d^\top(z)e_1&\dots&\mathbf{K}_d^\top(z)e_p
\end{array}\right]^\top=\left[\scriptsize\begin{array}{ccc|c}
	\widetilde A_{d1}-zI_{n_1}&&&\widetilde B_{d1}\\
	&\ddots&&\vdots\\
	&&\widetilde A_{dp}-zI_{n_p}&\widetilde B_{dp}\\\hline&&&\vspace{-1mm}\\
	\widetilde C_{d1}&&&\widetilde D_{d1}\\
	&\ddots&&\vdots\\
	&&\widetilde C_{dp}&\widetilde D_{dp}
\end{array}\right].\label{eq:z_col_concat}
\end{equation}

Due to its block-diagonal structure, the pole pencil of the realization from \eqref{eq:z_col_concat} will have all its eigenvalues located in $\mathbb{S}$, which ensures that the obtained realization for $\mathbf{K}_d(z)$ will be stabilizable and detectable, while Theorem~\ref{thm:IO_stab} shows that $\mathbf{K}_d(z)$ is a stabilizing controller for $\begin{bmatrix}
I&\mathbf{G}^\top(z)
\end{bmatrix}^\top$. Recall that $\mathbf{G}(z)=\left[\footnotesize\begin{array}{c|c}
A-zI_n&B\\\hline C&D
\end{array}\right]$ has a stabilizable and detectable realization and let $\begin{bmatrix}
I&\mathbf{G}^\top(z)
\end{bmatrix}^\top=\left[\footnotesize\begin{array}{c|cc}
A^\top-zI_n&O&C^\top\\\hline\vspace{-2mm}\\ B^\top&I_m&D^\top
\end{array}\right]^\top$, with the latter having a realization possessing the same state variables as the realization of $\mathbf{G}(z)$ and which is both stabilizable and detectable (since so is the realization of $\mathbf{G}(z)$). For these realizations of $\mathbf{G}(z)$ and $\begin{bmatrix}
I&\mathbf{G}^\top(z)
\end{bmatrix}^\top$, it is straightforward to check that the closed-loop state dynamics in Figs.~\ref{fig:2BlockAgain} and~\ref{fig:eq_con} are described by the same set of equations (with respect to the same sets of inputs and outputs). Thus, by the same arguments as those in part $\mathbf{(IV)}$ of the proof of Theorem III.6 in \cite{NRF}, we employ the discrete-time versions of Definition 5.2 and of Lemmas 5.2 and 5.3 from \cite{zhou} for the closed-loop interconnection in Fig.~\ref{fig:eq_con} to get the desired result.

$\mathbf{II)}$ For the continuous-time case, start by recalling the realization from \eqref{VWalt}, which was shown in the proof of Proposition~\ref{prop:not_stab} to be minimal. Since the chosen SRTR pair is stable, Lemma~\ref{lem:SRTR_stab} ensures that $A_{22}+KA_{12}$ has all its eigenvalues located in $\mathbb{S}$. Thus, $\mathbf{K}_d(s)$ has $p$ unstable poles in $s=0$ along with $n-p$ poles located in $\mathbb{S}$. Obtain now a realization for each $e_i^\top\mathbf{K}_d(s)$ as follows\vspace{-1mm}
\begin{multline} \label{VWalti}
	e_i^\top\mathbf{K}_d(s)=\left[\small\begin{array}{c|c}
		A_{di}-sI_{n-p+1}&B_{di}\\\hline C_{di}&D_{di}
	\end{array}\right]:=\\:= \left[\scriptsize\begin{array}{cc|cc}
		-s & e_i^\top A_{12}   & e_i^\top (A_{11} - A_{12}K) & e_i^\top B_1  \\
		O & A_{22}+ K A_{12}  -s I_{n-p}  & K A_{11} - KA_{12}K + A_{21} -A_{22}K   &  KB_1+B_2  \\\hline 
		1 & O & O & O
	\end{array}\right],\normalsize
\end{multline}
and notice that \eqref{VWalti} is implicitly controllable and observable at $\forall s\in\mathbb{C}\backslash\{\mathbb{S}\cup\{0\}\}$.

We now prove controllability and observability at $s=0$ for the realizations of type \eqref{VWalti}. The latter property follows from the fact that $A_{22}+KA_{12}$ has all its eigenvalues located in $\mathbb{S}$. Since the pencil from \eqref{eq:ctrb} has full row rank for $s=0$, then
\begin{equation*}
\begin{bmatrix}
	A_{di}&B_{di}
\end{bmatrix}=\begin{bmatrix}
	e_i^\top&O\\O&I_{n-p}
\end{bmatrix}\begin{bmatrix}
	O&S_1(0)\\O&S_2(0)
\end{bmatrix}\small\begin{bmatrix}
	1&O&O&O\\O&O&I_p&O\\O&I_{n-p}&O&O\\O&O&O&I_m
\end{bmatrix}\normalsize
\end{equation*}
has full row rank, which makes \eqref{VWalti} controllable at $s=0$.

We conclude that the realization from \eqref{VWalti} is both stabilizable and detectable and, since $A_{22}+KA_{12}$ has all its eigenvalues located in $\mathbb{S}$, each $e_i^\top\mathbf{K}_d(s)$ has one pole in $s=0$. Therefore, the pole pencil of any stabilizable and detectable realization belonging to $e_i^\top\mathbf{K}_d(s)$ will have only one unstable eigenvalue, at $s=0$. Compute now just such realizations for each $e_i^\top\mathbf{K}_d(s)=\left[\footnotesize\begin{array}{c|c}
	\widetilde A_{di}-sI_{n_i}&\widetilde B_{di}\\\hline\vspace{-1.5mm}\\ \widetilde C_{di}&\widetilde D_{di}
\end{array}\right]$ and use them to form
\begin{equation}
	\mathbf{K}_d(s)=\left[\begin{array}{c:c:c}
		\mathbf{K}_d^\top(s)e_1&\dots&\mathbf{K}_d^\top(s)e_p
	\end{array}\right]^\top=\left[\scriptsize\begin{array}{ccc|c}
		\widetilde A_{d1}-sI_{n_1}&&&\widetilde B_{d1}\\
		&\ddots&&\vdots\\
		&&\widetilde A_{dp}-sI_{n_p}&\widetilde B_{dp}\\\hline\vspace{-1mm}&&&\\
		\widetilde C_{d1}&&&\widetilde D_{d1}\\
		&\ddots&&\vdots\\
		&&\widetilde C_{dp}&\widetilde D_{dp}
	\end{array}\right].\label{eq:s_col_concat}
\end{equation}

The pole pencil of the realization from \eqref{eq:s_col_concat} will have, due to its block-diagonal structure, exactly $p$ eigenvalues at $s=0$, with the rest of them being located in $\mathbb{S}$. However, recall that $\mathbf{K}_d(s)$ has only $p$ unstable poles, with all of them being located at $s=0$. Thus, by part $\mathbf{(II)}$ in the proof of Theorem III.6 from \cite{NRF}, the obtained realization of $\mathbf{K}_d(s)$ must be both stabilizable and detectable and, once again, Theorem~\ref{thm:IO_stab} shows that $\mathbf{K}_d(s)$ is a stabilizing controller for $\begin{bmatrix}
	I&\mathbf{G}^\top(s)
\end{bmatrix}^\top$. Recall the fact that $\mathbf{G}(s)=\left[\footnotesize\begin{array}{c|c}
	A-sI_n&B\\\hline C&D
\end{array}\right]$ has a stabilizable and detectable realization and let $\begin{bmatrix}
	I&\mathbf{G}^\top(s)
\end{bmatrix}^\top=\left[\footnotesize\begin{array}{c|cc}
	A^\top-sI_n&O&C^\top\\\hline\vspace{-1mm}\\ B^\top&I_m&D^\top
\end{array}\right]^\top$, with the latter having a realization possessing the same state variables as the realization of $\mathbf{G}(s)$ and which is both stabilizable and detectable (since so is the realization of $\mathbf{G}(s)$). As in part $\mathbf{I)}$, the closed-loop state dynamics in Figs.~\ref{fig:2BlockAgain} and~\ref{fig:eq_con} are described by the same set of equations (with respect to the same sets of inputs and outputs). By the same arguments as those in part $\mathbf{(IV)}$ of the proof of Theorem III.6 in \cite{NRF}, we employ Definition 5.2 and Lemmas 5.2 and 5.3 from \cite{zhou} for the closed-loop interconnection in Fig.~\ref{fig:eq_con} to get the desired result.

\medskip

{\bf Proof of Theorem~\ref{thm:MM}} The proof hinges upon first obtaining the representations provided in point $\mathbf{a)}$ of the result. Therefore, we will begin by addressing those rows of $\begin{bmatrix} \mathbf{W}(\lambda)&\mathbf{V}(\lambda) \end{bmatrix}$ for which the corresponding rows of $A_{12}$ are identically zero, in order to obtain \eqref{eq:linec}, which constitutes the first part of point $\mathbf{a)}$. Regarding the second part of the latter point, we will show how the matrices $Q_i^\top$, which compress the individual rows of $A_{12}$, along with condition $v)$ enable us to obtain the realizations from \eqref{eq:linenc}, thus wrapping up with point $\mathbf{a)}$. Finally, we will prove point $\mathbf{b)}$ by showing how condition $vi)$ ensures stability, for those rows described by \eqref{eq:linenc}, whereas conditions $i)$ through $iv)$ are sufficient to impose the sparsity structure of the two binary matrices, from the result's statement, upon the obtained SRTR pair. 

We begin by considering, for the remainder of the proof, some fixed $K\in\mathbb{R}^{(n-p)\times p}$. For each $i\in1:p$, we employ this $K$ in order to form, via \eqref{VWdef}, the realizations
\begin{equation}\label{eq:row_real}
		e_{p,i}^\top\small\begin{bmatrix}
		\mathbf W(\lambda) &\hspace{-2mm} \mathbf V(\lambda)
	\end{bmatrix}= \left[\footnotesize\begin{array}{c|cc}
		A_{22}+ K A_{12}  -\lambda I_{n-p}  & K A_{11} - KA_{12}K + A_{21} -A_{22}K   &  KB_1+B_2  \\ 
		\hline\vspace{-2mm}\\ e_{p,i}^\top A_{12}   & e_{p,i}^\top(A_{11} - A_{12}K) & e_{p,i}^\top B_1  \\ \end{array}\right].\normalsize
\end{equation}

To prove the first part of point $\mathbf{a)}$, let us assume that $\exists\ i\in1:p$ for which $\left\|e_{p,i}^\top A_{12}^{}\right\|= 0$. Then, the identity given in \eqref{eq:linec} follows directly from the elimination of all the dynamical modes in \eqref{eq:row_real}, by noticing that all of them will be unobservable.

To prove the second part of point $\mathbf{a)}$, we consider instead some $i\in1:p$ for which $\left\|e_{p,i}^\top A_{12}^{}\right\|\neq 0$. Notice that, in contrast to the previous part, such an $i$ must always exist since, otherwise, this would imply that $A_{12}=O$, which would contradict Assumption~\ref{Observability} (as per Remark~\ref{Leuenberger}). For the chosen $i$, define the following matrices\vspace{-1mm}
\begin{align*}
	\left[\begin{array}{c:c}
		\widetilde{A}_{11}&\widetilde{A}_{12}\\\hdashline\widetilde{A}_{21}&\widetilde{A}_{22}
	\end{array}\right]:=&\left[\begin{array}{cc}
	I_{n-p-n_i} & O \\ \hdashline O & I_{n_i}
\end{array}\right] Q_i(A_{22}+KA_{12})Q_i^\top\left[\begin{array}{c:c}
I_{n-p-n_i} & O \\ O & I_{n_i},
\end{array}\right],\\
	\widetilde{B}_{i1}:=&\begin{bmatrix}I_{n-p-n_i} & O\end{bmatrix} Q_i\begin{bmatrix}
		K A_{11} - KA_{12}K + A_{21} -A_{22}K   &  KB_1+B_2
	\end{bmatrix},\vspace{-1mm}
\end{align*}
along with $\widetilde{B}_{i2}$ and $\widetilde{D}_i$, as given in \eqref{eq:aux_defa}-\eqref{eq:aux_defb}. Apply now a state equivalence transformation given by $T:=Q_i$ to \eqref{eq:row_real} and recall the identity from \eqref{eq:compress} to notice that, if condition $v)$ from the result's statement holds, then we obtain the identity
\begin{equation}\label{eq:obsv_desc}
	e_{p,i}^\top\begin{bmatrix}
		\mathbf W(\lambda) &\hspace{-2mm} \mathbf V(\lambda)
	\end{bmatrix}\hspace{-1mm}=\hspace{-1mm} \left[\footnotesize\begin{array}{cc|c}
							\hspace{-1mm}\widetilde{A}_{11}-\lambda I_{n-p-n_i} & \widetilde{A}_{12} & \widetilde{B}_{i1}\\
							O & \widetilde{A}_{22}-\lambda I_{n_i} & \widetilde{B}_{i2}\hspace{-1mm}\\\hline
							O & \left\|e_{p,i}^\top A_{12}^{}\right\|e_{n_i,n_i}^\top & \widetilde{D}_i
		 \end{array}\right]
	 			\hspace{-1mm}=\hspace{-1mm} \left[\footnotesize\begin{array}{c|c}
	 				\widetilde{A}_{22}-\lambda I_{n_i} & \widetilde{B}_{i2}\hspace{-1mm}\\\hline
	 				\hspace{-1mm}\left\|e_{p,i}^\top A_{12}^{}\right\|e_{n_i,n_i}^\top & \widetilde{D}_i
	 			\end{array}\right]\hspace{-1mm},\vspace{-1mm}
\end{equation}
after eliminating the first $n-p-n_i$ unobservable components from the state vector of the larger state-space realization from \eqref{eq:obsv_desc}. Note that this elimination must only be done for those rows where we have $n_i<n-p$. Moreover, by recalling the fact that $\widetilde{A}_{22}:=\begin{bmatrix}O & I_{n_i}\end{bmatrix} Q_i(A_{22}+KA_{12})Q_i^\top\begin{bmatrix}O & I_{n_i}\end{bmatrix}^\top$ along with the identity from \eqref{eq:compress}, we get that the smaller state-space realization from \eqref{eq:obsv_desc} is precisely the one in \eqref{eq:linenc}.

To prove point $\mathbf{b)}$, we will address the stability and sparsity constraints separately, beginning with the former. To this end, recall that any row satisfying $\left\|e_{p,i}^\top A_{12}^{}\right\|= 0$ must have a constant matrix as its TFM, given by the row vector from \eqref{eq:linec}, which is stable in both continuous- and discrete-time. Recall, additionally, that we wish to express all those rows satisfying $\left\|e_{p,i}^\top A_{12}^{}\right\|\neq 0$ via the realizations given in \eqref{eq:linenc}. If condition $vi)$ from our result's statement holds, then note that the pole pencils of all the aforementioned state-space realizations will have all their eigenvalues located in $\mathbb{S}$. Therefore, the TFMs belonging to these state-space realizations cannot have poles located outside of $\mathbb{S}$, making these TFMs stable, in either continuous- or discrete-time.

Regarding the sparsity structure of the SRTR pair, a critical observation is the fact that, for an arbitrary $\mathbf{K}\in\mathbb{R}(\lambda)^{p\times m}$ described by a generic state-space realization $\mathbf{K}(\lambda)=C_{\mathbf{K}}(\lambda I_{n_\mathbf{K}}-A_\mathbf{K})^{-1}B_\mathbf{K}+D_\mathbf{K}$ with $n_\mathbf{K}\neq 0$, we have, for any $j\in1:m$, that
\begin{equation}\label{eq:MM_ss}
	e_{m,j}^\top\begin{bmatrix}
		B_{\mathbf{K}}^\top & D_{\mathbf{K}}^\top
	\end{bmatrix} = O \Longrightarrow \mathbf{K}(\lambda)e_{m,j}\equiv O.
\end{equation}
Moreover, notice that if conditions $i)$ and $ii)$ from the result's statement hold, then\vspace{-1mm}
\begin{subequations}
	\begin{align}
		e_{p,i}^\top\mathcal{B}_{\mathbf{W}}^{}e_{p,j}^{}=0\ \Longrightarrow &\ \  e_{p,i}^\top(A_{11}-A_{12}K)e_{p,j}^{}=0,\label{eq:suf1}\\
		e_{p,i}^\top\mathcal{B}_{\mathbf{V}}^{}e_{m,k}^{}=0\ \Longrightarrow &\ \ e_{p,i}^\top B_1e_{m,k}^{}=0.\label{eq:suf2}
	\end{align}
\end{subequations}

For those rows satisfying $\left\|e_{p,i}^\top A_{12}^{}\right\|= 0$ and described by \eqref{eq:linec}, we point out that \eqref{eq:suf1} is equivalent to \eqref{eq:sufW}, whereas \eqref{eq:suf2} is equivalent to \eqref{eq:sufV}.  We now address those rows which satisfy $\left\|e_{p,i}^\top A_{12}^{}\right\|\neq 0$ in a similar way. Thus, if conditions $iii)$ and $iv)$ from the result's statement hold, then we obtain the implications\vspace{-1mm}
\begin{subequations}
	\begin{align}
		e_{p,i}^\top\mathcal{B}_{\mathbf{W}}^{}e_{p,j}^{}=0\Longrightarrow &\begin{bmatrix}O &  I_{n_i}\end{bmatrix} Q_i( K A_{11} - KA_{12}K + A_{21} -A_{22}K)e_{p,j}^{}=O,\label{eq:suf3}\\
		e_{p,i}^\top\mathcal{B}_{\mathbf{V}}^{}e_{m,k}^{}=0\Longrightarrow &\begin{bmatrix}O & I_{n_i}\end{bmatrix} Q_i(KB_1+B_2)e_{m,k}^{}= O.\label{eq:suf4}
	\end{align}
\end{subequations}
Since we express these rows via the realizations from \eqref{eq:linenc}, we use \eqref{eq:MM_ss} to obtain \eqref{eq:sufW}, from \eqref{eq:suf1} and \eqref{eq:suf3}, along with \eqref{eq:sufV}, from \eqref{eq:suf2} and \eqref{eq:suf4}.

\section*{Acknowledgment}

The authors would like to thank our colleague, Dr. Bogdan D. Ciubotaru, for the advice made during the elaboration of this manuscript.\\

\bibliographystyle{siamplain}
\bibliography{references}

@ARTICLE{Realiz,
	
	author={Mohan Vamsi, Andalam Satya and Elia, Nicola},
	
	journal={IEEE Transactions on Automatic Control}, 
	
	title={Optimal Distributed Controllers Realizable Over Arbitrary Networks}, 
	
	year={2016},
	
	volume={61},
	
	number={1},
	
	pages={129-144},
	
	doi={10.1109/TAC.2015.2434051}}

@ARTICLE{Explicit,
	
	author={Jensen, Emily and Bamieh, Bassam},
	
	journal={IEEE Transactions on Automatic Control}, 
	
	title={An Explicit Parametrization of Closed Loops for Spatially Distributed Controllers With Sparsity Constraints}, 
	
	year={2022},
	
	volume={67},
	
	number={8},
	
	pages={3790-3805},
	
	doi={10.1109/TAC.2021.3111863}}

@INPROCEEDINGS{Separ,
	
	author={Li, Jing Shuang and Ho, Dimitar},
	
	booktitle={2020 American Control Conference (ACC)}, 
	
	title={Separating Controller Design from Closed-Loop Design: A New Perspective on System-Level Controller Synthesis}, 
	
	year={2020},
	
	volume={},
	
	number={},
	
	pages={3529-3534},
	
	doi={10.23919/ACC45564.2020.9147736}}

@article{BBS,
	
	author={Boyd, S. and Baratt, C. and Norman, S.},
	
	journal={Proc. of the IEEE}, 
	
	title={Linear controller design: limits of performance via convex optimization}, 
	
	year={1990},
	
	volume={78},
	
	number={3},
	
	pages={529--574},
	
	doi={10.1109/5.52229}
}

@article{trian,  
	author={Rösinger, Christian A. and Scherer, Carsten W.},  
	journal={IEEE Trans. on Control of Network Systems},   
	title={{A Flexible Synthesis Framework of Structured Controllers for Networked Systems}},   year={2020},  
	volume={7},  
	number={1},  
	pages={6--18},  
	doi={10.1109/TCNS.2019.2914411}
}

@article{Lavaei,
	Author = {Ghazal Fazelnia and Ramtin Madani and Abdulrahman Kalbat and Javad Lavaei},
	Journal = {IEEE Transactions on Automatic Control},
	Number = {1},
	Pages = {206 -- 221},
	Title = {{Convex Relaxation for Optimal Distributed Control Problems}},
	Volume = {62},
	Year = {2017},
	doi={10.1109/TAC.2016.2562062}}

@article{Jovanovic,
	Author = {Fu Lin and Makan Fardad and Mihailo R. Jovanovic},
	Journal = {IEEE Transactions on Automatic Control},
	Pages = {2426--2431},
	Title = {{Design of Optimal Sparse Feedback Gains via the Alternating Direction Method of Multipliers}},
	Volume = {58},
	Year = {2013},
	doi={10.1109/TAC.2013.2257618}}

@article{Sznaier,
	Author = {Yin Wang  and Jose A. Lopez  and Mario Sznaier},
	Journal = {IEEE Transactions on Automatic Control},
	Number = {10},
	Pages = {3393 -- 3403},
	Title = {{Convex Optimization Approaches to Information Structured Decentralized Control}},
	Volume = {63},
	Year = {2018},
	doi={10.1109/TAC.2018.2830112}}

@article{DateChow,
	Author = {Ranjit A. Date and Joe H. Chow},
	Journal = {IEEE Transactions on Automatic Control,},
	Number = {2},
	Pages = {347--351},
	Title = {{Decentralized Stable Factors and a Parametrization of Decentralized Controllers}},
	Volume = {39},
	Year = {1994},
	doi={10.1109/9.272331}}

@article{DSF,
	title = {Robust dynamical network structure reconstruction},
	journal = {Automatica},
	volume = {47},
	number = {6},
	pages = {1230--1235},
	year = {2011},
	doi = {https://doi.org/10.1016/j.automatica.2011.03.008},
	author = {Ye Yuan and Guy-Bart Stan and Sean Warnick and Jorge Goncalves},
}

@INPROCEEDINGS{Tseng,
	
	author={Tseng, Shih-Hao},
	
	booktitle={2021 American Control Conference (ACC)}, 
	
	title={Realization, Internal Stability, and Controller Synthesis}, 
	
	year={2021},

	pages={3570--3577},
	
	doi={10.23919/ACC50511.2021.9482682}}

@article{CDC2013,
	
	author={S\u{a}bau, {\c{S}}erban and Oar\u{a}, Cristian and Warnick, Sean and Jadbabaie, Ali},
	
	journal={Proc. of the $52^\text{nd}$ IEEE Conference on Decision and Control}, 
	
	title={Structured coprime factorizations description of Linear and Time-Invariant networks}, 
	
	year={2013},
	
	pages={2720--2725},
	
	doi={10.1109/CDC.2013.6760294}}

@article{CF,
title = {Squaring down with zeros cancellation in generalized systems},
journal = {Systems \& Control Letters},
volume = {92},
pages = {5-13},
year = {2016},
author = {Cristian Oar{\u{a}} and Cristian Flutur and Marc Jungers},
doi = {https://doi.org/10.1016/j.sysconle.2016.02.019}
}

@book{genricc,
	title={{Generalized Riccati Theory and Robust Control: A Popov Function Approach}},
	author={Ionescu, V. and Oar{\u{a}}, C. and Weiss, M.},
	year={1999},
	publisher={John Wiley}
}

@article{reg,
	author={Matni, Nikolai and Chandrasekaran, Venkat},
	journal={IEEE Transactions on Automatic Control}, 
	title={{Regularization for Design}}, 
	year={2016},
	volume={61},
	number={12},
	pages={3991--4006},
	doi={10.1109/TAC.2016.2517570}
}

@article{Luca1,  
	author={Furieri, Luca and Zheng, Yang and Papachristodoulou, Antonis and Kamgarpour, Maryam},  
	journal={IEEE Control Systems Letters},   
	title={{An Input–Output Parametrization of Stabilizing Controllers: Amidst Youla and System Level Synthesis}},   
	year={2019},  
	volume={3},  
	number={4},  
	pages={1014--1019},
	doi={10.1109/LCSYS.2019.2920205}}

@article{Luca2,
title = {{System-level, input–output and new parameterizations of stabilizing controllers, and their numerical computation}},
journal = {Automatica},
volume = {140},
pages = {110211},
year = {2022},
author = {Yang Zheng and Luca Furieri and Maryam Kamgarpour and Na Li},
doi = {https://doi.org/10.1016/j.automatica.2022.110211}
}

@article{Luca3,
	author={Zheng, Yang and Furieri, Luca and Papachristodoulou, Antonis and Li, Na and Kamgarpour, Maryam},
	journal={IEEE Transactions on Automatic Control}, 
	title={{On the Equivalence of Youla, System-Level, and Input–Output Parameterizations}}, 
	year={2021},
	volume={66},
	number={1},
	pages={413--420},
	doi={10.1109/TAC.2020.2979785}}

@article{aug_sparse,
	author={Speril\u{a}, Andrei and Oar\u{a}, Cristian and Ciubotaru, Bogdan D. and Sab\u{a}u, {\c{S}}erban},
	journal={{IEEE Transactions on Automatic Control}}, 
	title={{Distributed Control of Descriptor Networks: A Convex Procedure for Augmented Sparsity}}, 
	year={2023},
	volume={68},
	number={12},
	pages={8067--8074},
	doi={10.1109/TAC.2023.3301949}
}

@book{zhou,
	title={{Robust and Optimal Control}},
	author={Zhou, Kemin and Doyle, John and Glover, Keith},
	year={1996},
	publisher={Prentice-Hall}
}

@book{Kai,
	title={{Linear Systems}},
	author={Kailath, Thomas},
	publisher={Prentice-Hall},
	year={1980}
}

@book{gantmacher,
	title={{The Theory of Matrices}},
	author={Gantmacher, Feliks},
	year={1959},
	publisher={American Math. Society}
}

@book{Rose70,
	title={State-space and multivariable theory},
	author={Rosenbrock, Howard Harry},
	year={1970},
	publisher={Nelson}
}

@article{Verg81,  
	author={Verghese, G. and Levy, B. and Kailath, T.},  
	journal={IEEE Transactions on Automatic Control},   
	title={A generalized state-space for singular systems},   
	year={1981},  
	volume={26},  
	number={4},  
	pages={811--831},  
	doi={10.1109/TAC.1981.1102763}
}

@book{Wonham,
	title={Linear Multivariable Control},
	author={Wonham, W Murray},
	subtitle={Optimal control theory and its applications},
	year={1985},
	publisher={Springer}
}

@book{V,
	title={Control System Synthesis: A Factorization Approach},
	author={Vidyasagar, Mathukumalli},
	journal={Synthesis lectures on control and mechatronics},
	year={1985},
	publisher={MIT Press}
}

@article{SIMAX,
	author = {Oar\u{a}, Cristian and Varga, Andras},
	title = {Minimal Degree Coprime Factorization of Rational Matrices},
	year = {1999},
	volume = {21},
	number = {1},
	doi = {10.1137/S0895479898339979},
	journal = {SIAM Journal on Matrix Analysis and Applications},
	pages = {245–-278},
}

@article{Voul1,  
	author={Voulgaris, P.G.},  
	journal={Proc. of the 2000 American Control Conference},   
	title={Control of nested systems},   
	year={2000},  
	volume={6},  
	pages={4442--4445},  
	doi={10.1109/ACC.2000.877064}
}

@article{Voul2,  
	author={Voulgaris, P.G.},  
	journal={Proc. of the 2000 American Control Conference},   
	title={A convex characterization of classes of problems in control with specific interaction and communication structures},   
	year={2001},  
	volume={4},  
	pages={3128--3133},  
	doi={10.1109/ACC.2001.946401}
}

@article{Rotko,  
	author={Rotkowitz, M. and Lall, S.},  
	journal={IEEE Transactions on Automatic Control},   
	title={{A Characterization of Convex Problems in Decentralized Control}},   year={2006},  
	volume={51},  
	number={2},  
	pages={274--286},  
	doi={10.1109/TAC.2005.860365}
}

@article{Shah,  
	author={Shah, Parikshit and Parrilo, Pablo A.},  
	journal={Proc. of the $50^{th}$ IEEE Conference on Decision and Control and European Control Conference},   
	title={An optimal controller architecture for poset-causal systems},   year={2011}, 
	pages={5522--5528},  
	doi={10.1109/CDC.2011.6160816}
}

@article{Nett,  
	author={Nett, C. and Jacobson, C. and Balas, M.},  
	journal={IEEE Transactions on Automatic Control},   
	title={A connection between state-space and doubly coprime fractional representations},   
	year={1984},  
	volume={29},  
	number={9},  
	pages={831--832},  
	doi={10.1109/TAC.1984.1103674}
}

@article{Lucic,  
	author={Lucic, V.M.},  
	journal={IEEE Transactions on Automatic Control},   
	title={Doubly coprime factorization revisited},   
	year={2001},  
	volume={46},  
	number={3},  
	pages={457--459},  
	doi={10.1109/9.911423}
}

@article{EJC,
	title = {{Continuous-Time Non-Symmetric Algebraic Riccati Theory: A Matrix Pencil Approach}},
	journal = {European Journal of Control},
	volume = {18},
	number = {1},
	pages = {74--81},
	year = {2012},
	doi = {https://doi.org/10.3166/ejc.18.74-81},
	author = {Marc Jungers and Cristian Oar\u{a}},
}

@article{Kron1,
	title = {{Towards Kron reduction of generalized electrical networks}},
	journal = {Automatica},
	volume = {50},
	number = {10},
	pages = {2586--2590},
	year = {2014},
	doi = {https://doi.org/10.1016/j.automatica.2014.08.017},
	author = {Sina Yamac Caliskan and Paulo Tabuada},
}

@article{Kron2,  
	author={Floriduz, Alessandro and Tucci, Michele and Riverso, Stefano and Ferrari-Trecate, Giancarlo},  
	journal={IEEE Trans. on Control Systems Technology}, 
	title={{Approximate Kron Reduction Methods for Electrical Networks With Applications to Plug-and-Play Control of AC Islanded Microgrids}},   year={2019},  
	volume={27},  
	number={6},  
	pages={2403--2416},  
	doi={10.1109/TCST.2018.2863645}
}

@article{Kron3,  
	author={Dorfler, Florian and Bullo, Francesco},  
	journal={IEEE Trans. on Circuits and Systems I: Regular Papers},   title={{Kron Reduction of Graphs With Applications to Electrical Networks}},
	year={2013},  
	volume={60},  
	number={1},  
	pages={150--163},  
	doi={10.1109/TCSI.2012.2215780}
}

@article{plutonizare,  
	author={Sab\u{a}u, {\c{S}}erban and Oar\u{a}, Cristian and Warnick, Sean and Jadbabaie, Ali},  
	journal={IEEE Transactions on Automatic Control},   
	title={{Optimal Distributed Control for Platooning via Sparse Coprime Factorizations}},   
	year={2017},  
	volume={62},  
	number={1},  
	pages={305--320},  
	doi={10.1109/TAC.2016.2572002}
}

@article{NRF,
	author={Sab\u{a}u, {\c{S}}erban and Speril\u{a}, Andrei and Oar\u{a}, Cristian and Jadbabaie, Ali},
	journal={{IEEE Transactions on Automatic Control}}, 
	title={{Network Realization Functions for Optimal Distributed Control}}, 
	year={2023},
	volume={68},
	number={12},
	pages={8059--8066},
	doi={10.1109/TAC.2023.3298549}
}

@article{SLS,  
	author={Wang, Yuh-Shyang and Matni, Nikolai and Doyle, John C.},  journal={IEEE Transactions on Automatic Control},   
	title={{A System-Level Approach to Controller Synthesis}},   
	year={2019},  
	volume={64},  
	number={10},  
	pages={4079--4093},  
	doi={10.1109/TAC.2018.2890753}
}
\end{document}